\documentclass[11pt,letterpaper]{article}
\usepackage{amsmath,amssymb,amsfonts,amsthm}
\usepackage{mathrsfs}
\usepackage{thm-restate}
\usepackage{tikz}
\usetikzlibrary{positioning,calc}
\usepackage{subcaption}
\usepackage{enumitem}
\usepackage{bbm}
\usepackage{pgfplots}
\pgfplotsset{compat=1.18}
\usepackage{xcolor}
\usepackage{extpfeil}
\usepackage[ruled,vlined,linesnumbered]{algorithm2e}
\usepackage[margin=1in]{geometry}
\usepackage[colorlinks, citecolor = blue, linkcolor = magenta]{hyperref}
\usepackage{cleveref}
\Crefname{subappendix}{Section}{Sections}
\usepackage[most]{tcolorbox}   
\tcbuselibrary{skins, breakable}
\numberwithin{equation}{section}
\newtheorem{theorem}{Theorem}[section]
\newtheorem{claim}[theorem]{Claim}

\newtheorem{proposition}[theorem]{Proposition}
\newtheorem{corollary}[theorem]{Corollary}
\newtheorem{lemma}[theorem]{Lemma}

\newtheorem{question}[theorem]{Question}

\theoremstyle{definition}
\newtheorem{definition}[theorem]{Definition}
\newtheorem{remark}[theorem]{Remark}
\newtheorem{notation}[theorem]{Notation}
\newtheorem{observation}[theorem]{Observation}

\allowdisplaybreaks

\newcommand{\bad}{\mathrm{bad}}
\newcommand{\leaf}{\mathrm{leaf}}
\newcommand{\rem}{\mathrm{rem}}
\newcommand{\anc}{\mathrm{anc}}

\renewcommand{\leq}{\leqslant}
\renewcommand{\le}{\leqslant}
\renewcommand{\geq}{\geqslant}
\renewcommand{\ge}{\geqslant}

\newcommand{\calA}{\mathcal{A}}

\newcommand{\calC}{\mathcal{C}}
\newcommand{\calD}{\mathcal{D}}
\newcommand{\calE}{\mathcal{E}}
\newcommand{\calF}{\mathcal{F}}

\newcommand{\calM}{\mathcal{M}}
\newcommand{\calN}{\mathcal{N}}
\newcommand{\calO}{\mathcal{O}}

\newcommand{\calR}{\mathcal{R}}
\newcommand{\calT}{\mathcal{T}}
\newcommand{\calU}{\mathcal{U}}
\newcommand{\calV}{\mathcal{V}}
\newcommand{\calI}{\mathcal{I}}
\newcommand{\calX}{\mathcal{X}}
\newcommand{\calY}{\mathcal{Y}}

\newcommand{\bC}{\mathbb{C}}
\newcommand{\bE}{\mathbb{E}}
\newcommand{\bF}{\mathbb{F}}
\newcommand{\bN}{\mathbb{N}}
\newcommand{\bP}{\mathbb{P}}
\newcommand{\bR}{\mathbb{R}}
\newcommand{\bZ}{\mathbb{Z}}

\newcommand{\bfa}{\mathbf{a}}

\newcommand{\bfx}{\mathbf{x}}
\newcommand{\bfy}{\mathbf{y}}
\newcommand{\bfz}{\mathbf{z}}
\newcommand{\bfP}{\mathbf{P}}
\newcommand{\bfY}{\mathbf{Y}}

\newcommand{\bdP}{\boldsymbol{P}}
\newcommand{\bdzeta}{\boldsymbol{\zeta}}

\newcommand{\sfe}{\mathsf{e}}
\newcommand{\sfv}{\mathsf{v}}
\newcommand{\sfH}{\mathsf{H}}

\newcommand{\fraki}{\mathfrak{i}}
\newcommand{\frakX}{\mathfrak{X}}

\newcommand{\scrN}{\mathscr{N}}

\newcommand{\yes}{\mathrm{yes}}
\newcommand{\no}{\mathrm{no}}

\newcommand{\nil}{\textup{\texttt{nil}}}

\newcommand{\loc}{\mathrm{loc}}

\newcommand{\proj}{\mathrm{proj}}

\newcommand{\ZmodN}{\bZ_{N}}
\newcommand{\ZNk}{\bZ_{N}^{k}}

\newcommand{\supp}{\mathrm{supp}}   
\DeclareMathOperator{\poly}{poly}

\newcommand{\lspan}{\mathrm{span}}

\newcommand{\lp}{{\textsc{BasicLP}}}
\newcommand{\lpval}{\mathrm{val}^{\mathrm{LP}}}

\newcommand{\val}{\mathrm{val}}
\newcommand{\lpcurve}{\vartheta_{\calF}}

\newcommand{\mcsp}{\mathsf{MaxCSP}(\mathcal{F})}
\newcommand{\cspF}{\mathrm{CSP}(\mathcal{F})}

\newcommand{\McspF}[2]{\mathsf{MaxCSP}(\calF)[#1, #2]}

\newcommand{\Map}[2]{\mathrm{Map}\left(#1, #2\right)}

\newcommand{\ind}[1]{\mathbbm{1}\left\{#1\right\}}

\newcommand{\tcup}{\textstyle\bigcup}
\newcommand{\tprod}{\textstyle\prod}

\newcommand{\DIHP}{\mathsf{DIHP}}
\newcommand{\CC}{\mathsf{CC}}

\newcommand{\reff}{\mathrm{ref}}
\newcommand{\adv}{\mathrm{adv}}
\newcommand{\Dno}{\mathcal{D}_{\mathrm{no}}}

\newcommand{\defeq}{\xlongequal{\text{def}}}

\newcommand{\Ex}[1]{\bE \left[ #1 \right]}
\newcommand{\Exu}[2]{\underset{#1} \bE \left[ #2 \right] }

\newcommand{\Exs}[2]{\bE_{#1}\left[ #2 \right]}
\renewcommand{\Pr}[1]{\bP \left[ #1 \right]} 
\newcommand{\Prs}[2]{\bP_{#1} \left[ #2 \right]}
\newcommand{\Pru}[2]{\underset{ #1 }\bP \left[ #2 \right]}

\newcommand{\inp}[2]{\left\langle #1, #2 \right\rangle}

\title{A Dichotomy Theorem for Multi-Pass Streaming CSPs}
\author{Yumou Fei\thanks{Department of EECS, Massachusetts Institute of Technology.}\and Dor Minzer\thanks{Department of Mathematics, Massachusetts Institute of Technology. Supported by NSF CCF award 2227876 and NSF CAREER award 2239160.}\and Shuo Wang\thanks{Department of Mathematics, Massachusetts Institute of Technology. Supported by NSF award 2239160.}}
\date{\vspace{-5ex}}

\begin{document}
\maketitle
\begin{abstract}
    In a constraint satisfaction problem (CSP) in the single-pass streaming model, an algorithm is given the constraints $C_1,\ldots,C_m$ of an instance one after another (in some fixed order), and its goal is to approximate the value of the instance, i.e., the maximum fraction of constraints that can be satisfied simultaneously. In the $p$-pass streaming model the algorithm is given $p$ passes over the input stream (in the same order), after which it is required to output an approximation of the value of the instance. We show a dichotomy result for $p$-pass streaming algorithms for all CSPs and for up to polynomially many passes. More precisely, we prove that for any arity parameter $k$, finite alphabet $\Sigma$, collection $\mathcal{F}$ of $k$-ary predicates over $\Sigma$ and any $c\in (0,1)$, there exists $0<s\leq c$ such that:
    \begin{enumerate}
        \item For any $\varepsilon>0$ there is a constant pass, $O_{\varepsilon}(\log n)$-space randomized streaming algorithm solving the promise problem $\McspF{c}{s-\varepsilon}$. That is, the algorithm accepts inputs with value  at least $c$ with probability at least $2/3$, and rejects inputs with value at most $s-\varepsilon$ with probability at least $2/3$.
        \item For all $\varepsilon>0$, any $p$-pass (even randomized) streaming algorithm that solves the promise problem $\McspF{c}{s+\varepsilon}$ must use $\Omega_{\varepsilon}(n^{1/3}/p)$ space.
    \end{enumerate}
    Our approximation algorithm is based on a certain linear-programming relaxation of the CSP and on a distributed algorithm that approximates its value. This part builds on the works [Yoshida, STOC 2011] and  [Saxena, Singer, Sudan, Velusamy, SODA 2025]. For our hardness result we show how to translate an integrality gap of the linear program into a family of hard instances, which we then analyze via studying a related communication complexity problem. That analysis is based on discrete Fourier analysis and builds on a prior work of the authors and on the work [Chou, Golovnev, Sudan, Velusamy, J.ACM 2024].
\end{abstract}

\newpage
{
\setcounter{tocdepth}{3} 
\tableofcontents
}

\newpage

\section{Introduction}
Constraint satisfaction problems (CSPs in short) are some of the most well studied problems in theoretical computer science, appearing in the context of algorithms design, approximation algorithms, hardness of approximation and more. Many prominent combinatorial optimization problems, such as the Max-Cut problem, the Vertex-Cover problem and various graph/hypergraph coloring problems, can be naturally formulated as CSPs. This paper focuses on the study of CSPs in the algorithmic model of multi-pass streaming, and our main result is an approximation dichotomy theorem in this setting.
\subsection{Constraint Satisfaction Problems}
\begin{definition}
    For a positive integer $k\geq 1$ and a finite alphabet $\Sigma$, a $k$-ary CSP is given by a family of predicates ${\cal{F}}\subseteq \{ f: \Sigma^k\rightarrow \{0,1\}\}$.\footnote{Without loss of generality, we assume that there exists at least one function $f\in \mathcal{F}$ such that $f^{-1}(1)\neq \emptyset$, otherwise the CSP is degenerate.} An instance of $\cspF$ is specified by $\mathcal{I} =({\cal{V}}, \mathcal{C})$, where $\mathcal{V}$ is a set of variables and $\mathcal{C} = (C_1,\dots, C_m)$ is a sequence of constraints. Each constraint $C_{i}$ is a pair $(\sfe_{i}, f_{i})$, where $\sfe_{i}\in \mathcal{V}^k$ is a tuple of distinct variables and $f_{i}\in \calF$ is a predicate. The tuple $\sfe_{i}=(\sfv_{i,1},\dots,\sfv_{i,k})$ specifies the variables involved in the constraint, and $f_{i}$ defines the condition that must be satisfied on those variables. We also write $\calI\in\cspF$.
\end{definition}

In the notations above, we say that an assignment $\tau: \calV\rightarrow\Sigma$ satisfies the constraint $C_{i}$ if $f_i\big(\tau(\sfv_{i,1}),\dots,\tau(\sfv_{i,k})\big)=1$. The value of the assignment $\tau$ is defined to be the fraction of the constraints satisfied by $\tau$, namely
\begin{align*}
    \val_{\calI }(\tau):= \frac{1}{m}\sum_{i=1}^{m} f_i\big(\tau(\sfv_{i,1}),\dots,\tau(\sfv_{i,k})\big).
\end{align*}
The value of the instance $\calI$ is the maximum possible value any assignment may achieve: 
\begin{align*}
    \val_{\calI}:= \max_{\tau:\calV\rightarrow\Sigma}\{ \val_{\calI}(\tau)\}. 
\end{align*} 
The following computational problems are often associated with $\cspF$:
\begin{enumerate}
    \item Decision version: given an instance $\calI\in\cspF$, decide if $\val_{\calI}=1$ or $\val_{\cal I} < 1$. Namely, design an algorithm 
    that given an instance $\calI\in \cspF$, accepts if it is fully satisfiable, and rejects otherwise.
    \item Optimization version: given an instance 
    $\calI\in \cspF$, output an approximation $\widehat{\val}_{\calI}$ of $\val_{\cal I}$. We say that the algorithm is a $\theta$-approximation algorithm for $\theta\in (0,1]$, if for all instances $\calI$ the output $\widehat{\val}_{\calI}$ satisfies that
    \[
    \theta\val_{\cal I}
    \leq 
    \widehat{\val}_{\cal I}
    \leq 
    \val_{\cal I}.
    \]
\end{enumerate}
CSPs have been studied in several different computational models. While the focus of this paper is on the streaming model, we first discuss the more popular model of polynomial time algorithms vs.~the class NP, from which one may seek to draw analogies.
\subsubsection{CSPs in the NP World} 
Both the decision and optimization versions above have 
been studied in the context of polynomial time algorithms over the last few decades. The Cook-Levin theorem, which is the basis of all of the theory of NP-hardness, can be equivalently seen as asserting that there exists a collection $\mathcal{F}$ such that the decision version of $\cspF$ is NP-hard. Similarly, the PCP theorem, which is the basis of all of the theory of NP-hardness for approximation problems, can be equivalently seen as seen as asserting that there exists a collection $\mathcal{F}$ and a constant $\theta<1$ such that getting a $\theta$-approximation for the optimization version of $\cspF$ is NP-hard. Subsequent 
research focused on getting a more detailed understanding of the complexity of $\cspF$ for all $\calF$. Namely:
\begin{enumerate}
    \item Dichotomy for decision problems: given a family of predicates $\calF$, what is the complexity of the decision version of $\cspF$? The dichotomy theorem of Zhuk and Bulatov~\cite{Zhuk,Bulatov} (which was previously known as the dichotomy conjecture of Feder and Vardi~\cite{FederVardi}) asserts that for any $\calF$, the complexity of the decision version of $\cspF$ is either polynomial time, or else it is NP-hard.
    \item Dichotomy for optimization problems: given a family of predicates $\calF$, what is the best possible approximation ratio $\theta$ that can be achieved for the optimization version of $\cspF$? Is it the case that there is always a number $\theta$ such that for all $\varepsilon>0$, there is a polynomial time $(\theta-\varepsilon)$-approximation algorithm, but getting a $(\theta+\varepsilon)$-approximation is already NP-hard? The dichotomy theorem of Raghavendra~\cite{Rag08} proves an assertion along these lines (assuming the Unique-Games Conjecture~\cite{Khot02}), and below we discuss his result in more detail.
\end{enumerate}
In both cases, the corresponding dichotomy result 
also specifies a concrete polynomial time algorithm (though not fully explicit) satisfying the guarantee of the theorem. 
In the case of decision problems, the algorithm is based on linear-programming hierarchies and linear equations over groups, and in the case of optimization problems, the algorithm is based on semi-definite programming and appropriate rounding schemes.
To discuss approximation problems further it is convenient to use the notion of gap problems, defined as follows.
\begin{definition}\label{def:McspF}
For a fixed finite predicate family $\calF\subseteq \{f:\Sigma^{k}\rightarrow\{0,1\}\}$, a completeness parameter $c\in (0,1]$ and a soundness parameter $s\in [0,c)$, the problem $\McspF{c}{s}$ is the promise problem where given an instance $\calI\in\cspF$, 
the algorithm should distinguish between the following two cases:
\begin{enumerate}[label = (\arabic*)]
    \item Yes case: if $\val_{\calI}\geq c$, then 
    the algorithm should accept.
    \item No case: if $\val_{\calI}\leq s$, then the algorithm should reject.   
\end{enumerate}
We will often consider randomized 
algorithms for $\McspF{c}{s}$, 
in which case in the ``yes case'' 
we require the algorithm accepts with probability at least $2/3$, 
and in the ``no case'' we require
that the algorithm rejects with probability at least $2/3$.
\end{definition}
In this language, Raghavendra shows that assuming the Unique-Games Conjecture, for all families $\calF$ and $c\in (0,1)$, there exists $s\in [0,c]$ such that for all $\varepsilon>0$, the problem $\McspF{c}{s-\varepsilon}$ can be solved in polynomial time
but $\McspF{c}{s+\varepsilon}$ is NP-hard. We remark that as far as 
approximation ratios are concerned, this result 
gives a full dichotomy result for the optimization 
problem associated with $\cspF$. However, we note that
this result does not address satisfiable instances, namely it does not make any assertion on the problems 
$\McspF{1}{s}$. Indeed, the approximability of satisfiable CSPs is still largely open in the NP world, and we will see that interestingly, it also presents some (different) challenges in the context of streaming algorithms.

\subsubsection{CSPs in the Streaming World}
The study of CSPs in the streaming model has seen a lot of activity over the past decade~\cite{KKS15,guruswami2017streaming,KK19,chou2020optimal,AKSY20,AN21,chou2022linear,saxena2023improved,hwang2024oblivious,CGSV24,saxena2025streaming,FMW25} (see~\cite{SudanSurvey,Assadi} for surveys). For a function $S\colon\mathbb{N}\to\mathbb{N}$, a space $S$ streaming algorithm has $S(n)$ cells of memory (where $n$ is the size of the CSP instance), and it receives the constraints of the instance one by one. Upon receiving an element in the stream it is allowed to make arbitrary computations involving that element and the current memory state, and then update its memory state. 
Typically, we think of space complexity $S(n) = \poly(\log n)$ as efficient, and of space complexity as $S(n) = n^{\Omega(1)}$ as being inefficient. For
the purposes of this paper we will not be concerned with any other complexity measure of the algorithm (such as run-time). Additionally, we allow our algorithm to be randomized, and the number of random bits it uses is included in its space complexity. There are a few variants of the streaming model that are often considered, depending on the stream order and on the number of stream passes the algorithm makes:
\begin{enumerate}
    \item {\bf Input order:} because an efficient streaming algorithm cannot store the entire CSP instance in its memory, the order in which it receives the constraints may matter. The two models that are most often considered are the ``random order model'', in which the constraints are given in a randomly chosen order, and the ``worst-case order model'', in which the order of the constraints is predetermined by an adversary. 
    \item {\bf Number of passes:} a single-pass streaming algorithm is an algorithm which is given the input stream once, after which it must produce an answer. A $p$-pass streaming algorithm is an algorithm which is given $p$ passes over the stream (according to the same order), after which it must produce an answer. When the number of passes $p$ is constant the model is often simply referred to as the multi-pass streaming model. It also makes sense however to allow the number of passes $p$ to increase with the input length.
\end{enumerate}

The first problem studied in this context is the Max-Cut problem, and the main result of~\cite{KKS15} is that for any $\varepsilon>0$, a single-pass $(1/2+\varepsilon)$-approximation algorithm for Max-Cut requires $\Omega_{\varepsilon}(\sqrt{n})$ memory (even under random ordering of constraints). Since the trivial algorithm that simply counts the edges of the input graph and outputs half their number yields a $(1/2)$-approximation using only $O(\log n)$ memory, this establishes a sharp threshold at the approximation ratio $1/2$ in the streaming model. This lower bound was later improved to a nearly optimal $\Omega_{\varepsilon}(n)$ in~\cite{KK19}.
Subsequent works have studied the approximability of other predicates in this model~\cite{chou2020optimal,chou2022linear,CGSV24,saxena2023improved,saxena2025streaming} and extended these results to
the multi-pass setting~\cite{AN21,AKSY20,FMW25}.

\subsubsection{Dichotomy Theorems for Streaming Algorithms?}
Of particular interest to the current paper is the result of~\cite{CGSV24}, which studies a sub-class of single-pass streaming algorithms called sketching algorithms. A space-$S$ sketching algorithm is an algorithm whose memory is thought of as a ``summary'' of the input stream read so far. Formally, it consists of a sketching function ${\sf Sketch}\colon \{\text{constraints}\}\to \{0,1\}^S$, and a combining function ${\sf Comb}\colon \{0,1\}^S\times\{0,1\}^S\to\{0,1\}^S$ such that for any two streams of constraints $\sigma, \tau$, it holds $
    {\sf Comb}({\sf Sketch}(\sigma),{\sf Sketch (\tau)})  = {\sf Sketch} (\sigma \circ \tau)$,
where $\sigma \circ \tau$ denotes the data stream obtained by concatenating $\sigma$ and $\tau$.  
If the current memory state of the algorithm is $x$, and it receives a constraint $C$, the new memory state will be ${\sf Comb}(x,{\sf Sketch}(C))$, which implies that a space-$S$ sketching algorithm can always be implemented by a space-$O(S)$ single-pass algorithm, but the reverse may not be true. The main result of~\cite{CGSV24} can be seen as an analog of the result of Raghavendra~\cite{Rag08} for sketching algorithms, reading as follows:
\begin{theorem}[{\cite{CGSV24}}]\label{thm:CGSV}
    For every $k\in\mathbb{N}$, a family $\calF$ of $k$-ary predicates and $0\leq s<c\leq 1$, either the problem $\McspF{c}{s}$ admits an $O(\log^3 n)$-space sketching algorithm, or else for all $\varepsilon>0$, the problem $\McspF{c-\varepsilon}{s+\varepsilon}$ requires $\Omega_{\varepsilon}(\sqrt{n})$ memory.
\end{theorem}
The result of~\cite{CGSV24} is in fact more detailed, and it specifies a polynomial space algorithm that is able to tell which one of the cases holds, as well as the sketching algorithm in the former case.
Their characterization relies on the definition of two convex sets: if these sets are disjoint there is an algorithm as in~\Cref{thm:CGSV}, and else an intersection point can be used to construct hard instances.

In light of~\Cref{thm:CGSV}, one may ask whether similar dichotomy results hold for other streaming models. The most natural models to consider are the single-pass streaming and the multi-pass streaming models. In both of these models there are examples of non-trivial algorithms and non-trivial hardness results, but a priori it is not clear how to unify them into full-blown dichotomy results. In a sense, a key challenge in proving such dichotomy theorems is that one typically has to come up with a ``single algorithm'' that works for all families of predicates $\calF$, and prove that hard instances for it can be used to prove hardness for any algorithm.

\subsection{Main Result}
The main result of this paper is a dichotomy theorem for the approximability of $\cspF$ in the multi-pass streaming model:

\begin{restatable}{theorem}{mainthm}\label{thm:main}
For any finite predicate family $\calF\subseteq \{f:\Sigma^{k}\rightarrow\{0,1\}\}$, there exists a non-decreasing continuous function $\lpcurve:(0,1)\rightarrow(0,1)$ satisfying $\lpcurve(c)\leq c$ for all $c\in (0,1)$, such that
\begin{enumerate}[label = (\arabic*)]
\item for any fixed rational numbers $c\in (0,1)$ and $s\in (0,\vartheta_{\calF}(c))$, there exists a constant-pass, $O(\log n)$-space randomized streaming algorithm for $\McspF{c}{s}$;\footnote{Here $n$ stands for the number of variables in the input instance. We additionally assume that the number of constraints is at most polynomial in $n$.}

\item for any fixed rational numbers $c\in (0,1)$ and $s\in (\vartheta_{\calF}(c),c)$, any $p$-pass streaming algorithm for $\McspF{c}{s}$ requires $\Omega_{c,s}(n^{1/3}/p)$ space.
\end{enumerate}
\end{restatable}
In words, up to the value of exponents,~\Cref{thm:main} is an analog of~\Cref{thm:CGSV} in the multi-pass setting, asserting that a given gap problem can either be solved by a streaming algorithm with $O(\log n)$-space and constantly many passes, or else requires a polynomial space (or polynomially many passes). We defer a detailed discussion of our techniques to~\Cref{sec:techniques}, but remark that our algorithm 
is based on the basic linear-programming formulation of
the CSP and a connection with distributed computation observed in~\cite{Yos11,saxena2025streaming}.

\subsubsection{Examples: DICUT and 2SAT}\label{sec:imp}
Using~\Cref{thm:main} together with the fact that the function $\vartheta_{\calF}$ can, for certain families $\calF$ of interest, be determined explicitly, one can obtain an almost complete characterization of the complexity of the problem 
$\McspF{c}{s}$ as the parameters $0 \leq s < c < 1$ vary.  
(This characterization is not fully complete, since the case $c=1$ remains unresolved.)  
Below we illustrate this with a few examples.

An instance of the maximum directed-cut problem is defined by the collection $\calF = \{f\}$, where $f\colon \{0,1\}^2\to\{0,1\}$ is the predicate whose unique satisfying assignment is $(1,0)$. Previous works~\cite{saxena2025streaming,FMW25}
have shown that for all $\varepsilon>0$, the problem admits an $(1/2-\varepsilon)$-approximation constant pass algorithm with $O(\log n)$ space, whereas $(1/2+\varepsilon)$-approximation requires either polynomial space or polynomially many passes. Using \Cref{thm:main} and the following result, we are able to determine the full approximability curve of this problem:
\begin{restatable}{theorem}{DICUTthm}\label{thm:DICUT}
    For the case of $\mathsf{Max}$-$\mathsf{DICUT}$, we have (see \Cref{fig:max-dicut})
    \begin{align*}
    \vartheta_\calF(c) = 
    \begin{cases}
    c&\text{if }0\leq c\leq1/4,\\
    1/4&\text{if }1/4< c\leq  1/2,\\
    (3c-1)/2&\text{if }1/2< c\leq 1.\\
    \end{cases}
    \end{align*}
\end{restatable}

An instance of the Max-$2$SAT problem is defined by the collection $\calF$ consisting the unary predicates $f^{(0)}, f^{(1)}\colon \{0,1\}\to\{0,1\}$ defined as $f^{(b)}(x) = 1_{x=b}$, as well as the binary predicates $f^{(b_1,b_2)}\colon \{0,1\}^2\to\{0,1\}$ defined as $f^{(b_1,b_2)}(x_1,x_2) = 1 - 1_{x_1=b_1, x_2=b_2}$.  Using \Cref{thm:main} and the following result, we are able to determine the full approximability curve of this problem:
\begin{restatable}{theorem}{thmtwoSAT}\label{thm:2SAT}
    For the case of $\mathsf{Max}$-$\mathsf{2SAT}$, we have (see \Cref{fig:2sat})
    \begin{align*}
    \vartheta_\calF(c) = 
    \begin{cases}
    c&\text{if }0\leq c\leq 1/2, \\
    (2c+1)/4&\text{if }1/2< c\leq 1.
    \end{cases}
    \end{align*}
\end{restatable}
In particular, \Cref{thm:main,thm:2SAT} imply that the optimal approximation ratio of Max-$2$SAT is $3/4$ for multi-pass streaming algorithms.

\begin{figure}[ht]
\centering
\begin{subfigure}{0.48\textwidth}
\centering
\begin{tikzpicture}
  \begin{axis}[
    width=\linewidth, height=\linewidth,
    axis lines=middle,
    xmin=0, xmax=1.05,
    ymin=0, ymax=1.05,
    xlabel={$c$}, ylabel={$s$},
    xtick={0,0.25,0.5,1},
    xticklabels={$0$,$\tfrac14$,$\tfrac12$,$1$},
    ytick={0,0.25,0.5,1},
    yticklabels={$0$,$\tfrac14$,$\tfrac12$,$1$},
    grid=none,
    clip=false,
    axis line style={thick},
    tick style={thick}
  ]

    \draw[gray, thick, dashed] (axis cs:0,0) -- (axis cs:1,1);

    \fill[green, fill opacity=0.2]
      (axis cs:0,0) -- (axis cs:0.25,0.25)
      -- (axis cs:0.5,0.25) -- (axis cs:1,1)
      -- (axis cs:1,0) -- cycle;

    \fill[red, fill opacity=0.2]
      (axis cs:0,0) -- (axis cs:0.25,0.25)
      -- (axis cs:0.5,0.25) -- (axis cs:1,1)
      -- (axis cs:0,0) -- cycle;

    \draw[thick] (axis cs:0,0) -- (axis cs:0.25,0.25);
    \draw[thick] (axis cs:0.25,0.25) -- (axis cs:0.5,0.25);
    \draw[thick] (axis cs:0.5,0.25) -- (axis cs:1,1);

  \end{axis}
\end{tikzpicture}
\caption{ $\mathsf{Max}$-$\mathsf{DICUT}$.}
\label{fig:max-dicut}
\end{subfigure}
\hfill
\begin{subfigure}{0.48\textwidth}
\centering
\begin{tikzpicture}
  \begin{axis}[
    width=\linewidth, height=\linewidth,
    axis lines=middle,
    xmin=0, xmax=1.05,
    ymin=0, ymax=1.05,
    xlabel={$c$}, ylabel={$s$},
    xtick={0,0.5,1},
    xticklabels={$0$,$\tfrac12$,$1$},
    ytick={0,0.25,0.5,0.75,1},
    yticklabels={$0$,$\tfrac14$,$\tfrac12$,$\tfrac34$,$1$},
    grid=none,
    clip=false,
    axis line style={thick},
    tick style={thick}
  ]

    \draw[gray, thick, dashed] (axis cs:0,0) -- (axis cs:1,1);

    \fill[green, fill opacity=0.2]
      (axis cs:0,0) -- (axis cs:0.5,0.5)
      -- (axis cs:1,0.75)
      -- (axis cs:1,0) -- cycle;

    \fill[red, fill opacity=0.2]
      (axis cs:0,0) -- (axis cs:0.5,0.5)
      -- (axis cs:1,0.75) -- (axis cs:1,1)
      -- (axis cs:0,0) -- cycle;

    \draw[thick] (axis cs:0,0) -- (axis cs:0.5,0.5);
    \draw[thick] (axis cs:0.5,0.5) -- (axis cs:1,0.75);

  \end{axis}
\end{tikzpicture}
\caption{$\mathsf{Max\text{-}2SAT}$.}
\label{fig:2sat}
\end{subfigure}

\caption{The threshold functions $\lpcurve$ are shown in solid black lines. Red region stands for hardness, while green region corresponds to parameters where efficient multi-pass streaming algorithms exist.}
\label{fig:overall}
\end{figure}
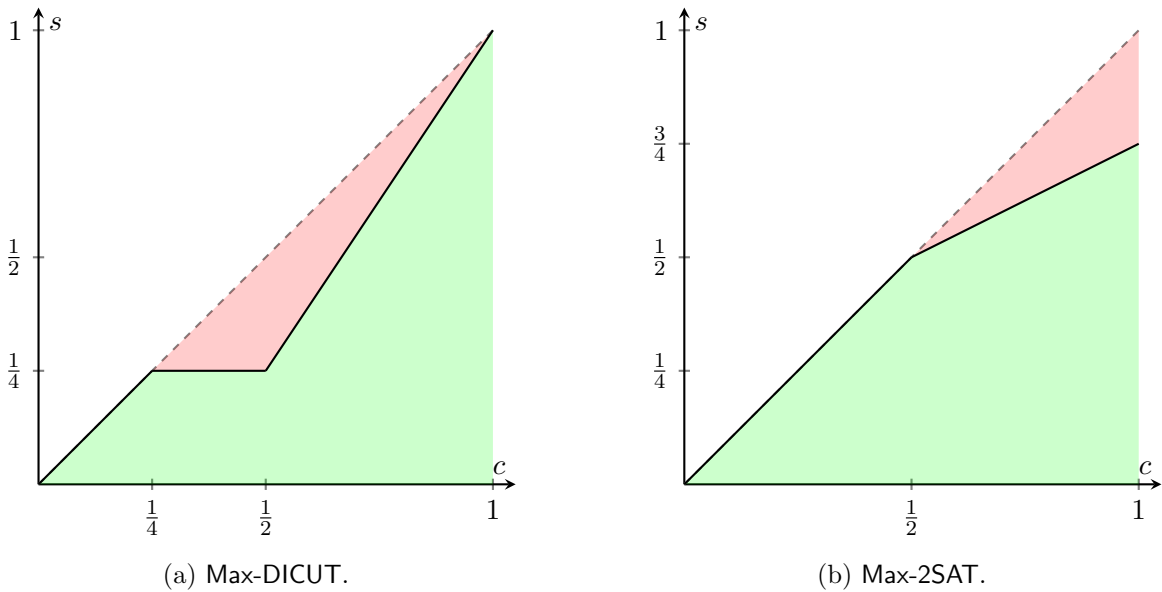

\subsubsection{Discussion: Sublinear Space vs. Sublinear Time}

Remarkably, the approximability threshold function $\vartheta_{\calF}$ arising in~\Cref{thm:main} exactly coincides with the one appearing in the analogous dichotomy theorem of~\cite{Yos11}.  
The dichotomy of~\cite{Yos11} concerns the same CSP approximation problem studied in~\cite{Rag08} (for polynomial time algorithms), \cite{CGSV24} (for sketching algorithms), and in this paper (for multi-pass streaming algorithms), but in yet another computational model --- namely, the sublinear-time model arising in bounded-degree graph property testing. We do not formally define this model here, but instead give an informal description of this connection.

In the bounded-degree query model (introduced by \cite{goldreich1997property}), an algorithm may adaptively query vertices of a graph, where each query reveals an edge incident to the queried vertex.  
The input graph is promised to have maximum degree at most a fixed constant $d$, and the algorithm’s goal is to decide whether the graph has a given property using as few queries as possible. The setting is similar in the context of CSPs: the algorithm may query variables of an input CSP instance, and each query reveals a constraint incident to the queried variable (there are promised to be no more than $d$ of them). The main result of \cite{Yos11} shows that for $c\in (0,1)$ and $s\in (0,\vartheta_{\calF}(c))$, there exists a randomized algorithm for deciding $\McspF{c}{s}$ using only a constant number of queries, whereas for $s\in (\vartheta_{\calF}(c),c)$ every such algorithm must make $\Omega_{c,s}(\sqrt{n})$ queries.

The coincidence between the threshold function in \cite{Yos11}'s dichotomy and the function $\vartheta_{\calF}$ arising in our result is intriguing: it suggests that, for the purpose of CSP approximation, the class of constant-query algorithms has essentially the same power as the class of efficient multi-pass streaming algorithms.

This is not accidental. One direction of the connection is more straightforward: any constant-query algorithm can be simulated by a constant-pass, logarithmic-space streaming algorithm. Consequently, on bounded-degree instances, multi-pass streaming algorithms are at least as powerful as constant-query algorithms. Since general CSP instances can be reduced to bounded-degree ones in the multi-pass streaming setting (formally argued in \Cref{sec:algorithm}), the performance of multi-pass streaming algorithms on general instances can match (or exceed) that of constant-query algorithms on bounded-degree instances. This is precisely the approach we take to prove our algorithmic result in \Cref{sec:algorithm}.

The reverse direction is considerably less trivial. Intuitively (and also formally, as we will argue in \Cref{sec:techniques,sec:communication_game}), multi-pass streaming algorithms correspond to communication protocols among several players, while query-based algorithms correspond to a more restricted class of ``structured'' protocols. 
Extending hardness from query-based algorithms to multi-pass streaming algorithms is thus analogous to establishing a query-to-communication lifting theorem, reminiscent of the lifting theorems in communication complexity (e.g., \cite{goos2017query}). Our result may therefore be viewed as such a lifting of \cite{Yos11}'s hardness result from a sublinear-time model to a sublinear-space model, albeit with a loss in the exponent: while \cite{Yos11} proves an $\Omega(\sqrt{n})$ time lower bound, we obtain only an $\Omega(n^{1/3})$ space lower bound. 

\begin{remark}\label{remark:general_graph_model}
The ``general graph model'' (introduced by \cite{parnas2002testing,kaufman2004tight}) is a sublinear-time algorithmic model that does not impose a bounded-degree assumption. This model may in fact be even closer to the multi-pass streaming model than the ``bounded-degree graph model'' is.
\end{remark}

\begin{remark}
The query vs.~communication connection has been exploited in the sublinear algorithm literature before this work. For instance, the result of \cite{blais2012property} is in some sense an ``un-lifting'' from communication lower bounds to query lower bounds. 
\end{remark}

\subsubsection{A Rich World of Approximability Hierarchy?}\label{subsec:hierarchy}

The relationships among the four dichotomy results discussed in \Cref{sec:imp} are illustrated in \Cref{fig:hiererachy}.  
The fact that multi-pass streaming algorithms are at least as powerful as sketching algorithms follows directly from the definitions of the models.  
By contrast, the second inequality in \Cref{fig:hiererachy} is not a priori obvious: the polynomial-time dichotomy is conditional on the Unique Games Conjecture and on $\mathrm{P}\neq\mathrm{NP}$, and $\mathrm{polylog}$-space algorithms need not run in polynomial time.  
Nevertheless, in light of \cite{Yos11} and our own results, we can safely conclude that CSP-approximation problems that are $\mathrm{UG}$-hard for polynomial-time algorithms are also unconditionally hard for multi-pass streaming and constant-query algorithms.  
This follows because the approximability threshold in both \cite{Yos11} and our work is determined by the basic linear programming relaxation (see \Cref{sec:techniques}), which is strictly weaker than the basic SDP relaxation featured in \cite{Rag08}.

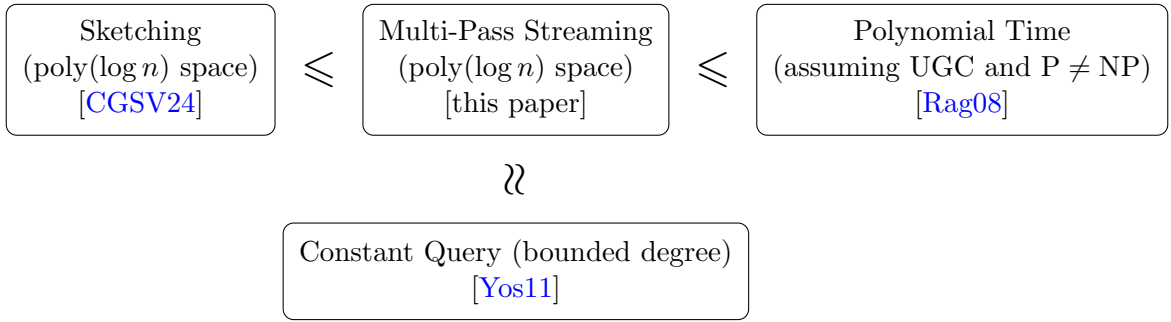
\begin{figure}[ht]
\begin{center}
\begin{tikzpicture}
\tikzset{
  bubble/.style={
    draw,
    rectangle,
    rounded corners,
    align=center,
    inner sep=6pt 
  }
}
\node[bubble] (A) {Sketching\\
($\poly(\log n)$ space)
\\ \cite{CGSV24}};
\node[bubble, right=3em of A] (B) {Multi-Pass Streaming\\($\poly(\log n)$ space)
\\ {[this paper]}};
\node[bubble, below=3em of B] (C) {Constant Query
(bounded degree)\\ \cite{Yos11}};
\node[bubble, right=3em of B] (D) {Polynomial Time\\ 
(assuming UGC and $\mathrm{P}\neq \mathrm{NP}$)
\\ \cite{Rag08}};

\node at ($(A.east)!0.5!(B.west)$) {\LARGE$\leq$};
\node[rotate = 90] at ($(B.south)!0.5!(C.north)$) {\LARGE$\approx$};
\node at ($(B.east)!0.5!(D.west)$) {\LARGE$\leq$};
\end{tikzpicture}
\end{center}
\caption{relative power of algorithmic models in approximating CSP value }\label{fig:hiererachy}
\end{figure}

An interesting open direction is whether the connection between the multi-pass streaming model and sublinear-time models extends to other approximation problems.  

A notable example is the Vertex-Cover minimization problem.  
On bounded-degree graphs without isolated vertices, \cite{parnas2007approximating} shows that for any $\varepsilon>0$ there exists a constant-query algorithm achieving a $(2+\varepsilon)$-approximation.  
On the other hand, the result of \cite{FMW25} implies that even in the (potentially more powerful) multi-pass streaming model, achieving a $(2-\varepsilon)$-approximation requires either polynomial memory or polynomially many passes (on bounded-degree graphs).  
Thus, the observed equivalence in power between multi-pass streaming and constant-query algorithms extends to the Vertex-Cover problem on bounded-degree graphs as well.

\begin{remark}
Unlike CSP-approximation problems, the Vertex-Cover problem becomes harder to approximate on graphs with unbounded degrees (see \cite{parnas2007approximating}). We do not claim that the equivalence in power extends also to graphs without a degree bound.
\end{remark}


As another example, the graph coloring problem appears to be very hard in both sublinear-time models and the multi-pass streaming model.  
In particular, the result of \cite{FMW25} implies that there exists a constant $\delta>0$ such that distinguishing $2$-colorable graphs from graphs that require at least $n^{\delta}$ colors requires either polynomial memory or polynomially many passes.\footnote{This relies on the observation that the communication lower bound in \cite{FMW25} depends only polynomially on the number of players $K$. The current paper contains a generalization (\Cref{thm:dihp_lowerbound}) of the communication lower bound of \cite{FMW25}, which likewise depends polynomially on the parameter $K$.} This lower bound also carries over to the general graph model (see \Cref{remark:general_graph_model}), implying that any such algorithm must make at least polynomially many queries.\footnote{This is because each query in the general graph model can be simulated using only one additional pass and $O(\log n)$ additional memory in the multi-pass streaming model.}


A more mysterious problem is the approximability of maximum matchings in graphs, which has recently attracted a lot of attention in both sublinear-time models (e.g. \cite{behnezhad2024approximating,mahabadi20250}) and streaming models (e.g. \cite{bury2015sublinear,AN21}). It is not clear yet whether similar equivalence in power holds for this problem, and we leave this question to future research.

\subsection{Techniques}\label{sec:techniques}
In this subsection, we present a high-level overview of the proof of \Cref{thm:main}. Although there are many high-level similarities between \Cref{thm:main} and \cite{Yos11}'s dichotomy theorem, we provide an independent exposition that does not assume prior familiarity with \cite{Yos11}.\footnote{The authors, in fact, learned about \cite{Yos11}'s result only after proving \Cref{thm:main}.}

As mentioned earlier, the key challenge in proving any dichotomy result (and in particular~\Cref{thm:main}) 
is in finding a candidate family of algorithms that are supposedly the best possible approximation algorithms in the computational model considered. Motivated by~\cite{saxena2025streaming}, we began our investigation by examining the work of Trevisan~\cite{Tre96}, who considered approximation algorithms for some CSPs based on positive linear programs, which is a certain sub-class of linear programs. Trevisan's motivation was that such algorithms are highly parallelized, and in particular they can be implemented using poly-logarithmic depth circuits. The value of a positive linear program (on bounded degree instances) was shown to be approximable using distributed algorithms~\cite{KMW06}, which given the result of~\cite{saxena2025streaming} raises the possibility that an algorithm based on positive linear programming captures some class of streaming model. We show that this is indeed the case.

\subsubsection{The Linear-Programming Relaxation}
The above discussion naturally leads one to consider the so-called basic linear-programming relaxation of a given CSP instance. Fix an instance $\calI=(\calV, \calC)$ of $\cspF$, write the constraints $\calC=(C_{1},\dots,C_{m})$, and for each $i\in[m]$ write the $i$th constraint $C_{i}=((\sfv_{i,1},\dots,\sfv_{i,k}),f_{i})$. The program $\lp_{\calI}$ has the variables $(x_{\sfv,\sigma})_{\sfv\in \calV,\,\sigma\in \Sigma}$ and $(z_{i,b})_{i\in [m],\, b\in \Sigma^{k}}$, and it proceeds as follows:
\begin{tcolorbox}[
  enhanced,
  breakable,
  title={$\lp_{\calI}$ for $\calI=(\calV, \calC)$},
  fonttitle=\bfseries,
  colframe=black,
  colback=white,
  boxrule=0.6pt,
  arc=6pt,
  left=6pt,
  right=6pt,
  top=4pt,
  bottom=4pt
]
\begin{alignat*}{3}
\textup{maximize} \quad     &&  \frac{1}{m}\sum_{i=1}^{m} 
            \sum_{b\in\Sigma^{k}} f_i(b)\,z_{i,b} \\
\textup{subject to} \quad && \sum_{\sigma\in\Sigma} x_{\sfv,\sigma}&=1
            \quad &\forall\,\sfv\in \calV \\
           && \sum_{b\in\Sigma^{k}}\mathbbm{1}\{b_{j}=\sigma\}\cdot
             z_{i,b}&=x_{\sfv_{i,j},\,\sigma}
            \quad &\forall\,i\in [m],\,j\in [k],\,\sigma\in \Sigma\\
           && x_{\sfv,\sigma}&\ge 0
            \quad &\forall\,\sfv\in \calV,\,\sigma\in\Sigma \\
           && z_{i,b}&\ge 0
            \quad &\forall\,i\in [m],\,b\in\Sigma^{k}
\end{alignat*}
\end{tcolorbox}

While the formulation of $\lp_{\calI}$ above itself is not strictly speaking a positive linear program, there are standard reductions that can turn it into one. Using this and the algorithm of~\cite{KMW06}, one can show that on bounded degree instances (by which we mean that each variable appears in $O(1)$ many constraints), the value of $\lp_{\calI}$ can be approximated within an additive error $\varepsilon>0$ by a local algorithm. This requires a nontrivial amount of work, and is already achieved by \cite{Yos11}, allowing Yoshida to build an algorithm with constant query complexity that approximates the value of $\lp_{\calI}$.

\subsubsection{Approximating the Value of $\lp_{\calI}$ via a Streaming Algorithm}
Building on top of~\cite{Yos11,saxena2025streaming} we show that the value of $\lp_{\calI}$ can be approximated within an additive error $\varepsilon>0$ by a streaming algorithm with $O(\log n)$ memory and constantly many passes. For this purpose, we first show how to handle instances with potentially unbounded degree, and for that we use an idea of~\cite{Tre01} that, in the context of the class NP, performs a degree reduction for CSPs while roughly maintaining its value. Second, we use the fact that Yoshida's algorithm is a local algorithm, which by results from~\cite{saxena2025streaming} can be therefore simulated with constantly many passes. The only missing piece from the above description is the integration of the degree-reduction step and the simulation step. While the simulation of~\cite{saxena2025streaming} requires explicit access to the bounded degree instance $\calI$, we cannot afford such an access (as our intended instance is a random sparsified version of $\calI$ as in~\cite{Tre01}). We overcome this difficulty by considering a slight variant of the sparsification of~\cite{Tre01} which is more amendable to the streaming model. At a high level, instead of generating the sparsified instance in a single shot, we only generate parts of it that are required to answer the queries made by Yoshida's algorithm. In particular, it is important that we can afford to store the randomness necessary to maintain that part of the graph. 

\subsubsection{From an Integrality Gap to Communication Complexity}
The function $\vartheta_\calF$ from~\Cref{thm:main} is defined using the program $\lp_{\calI}$: for all $c$, the value $\vartheta_\calF(c)$ is the infimum of $\val_{\calI}$ over all instances $\calI$ with $\lp_{\calI}$-value at least $c$ (see~\Cref{def:lpcurve}). 
With this in mind, the first item in~\Cref{thm:main} is a consequence of the algorithm discussed in the previous section. To complete the proof of~\Cref{thm:main} one must show the second item, which amounts to saying that an instance $\calI$ with integral value $s$ and $\lp_{\calI}$ value $c$ can be converted into a hardness result for the multi-pass streaming model.\footnote{An analogous assertion for semi-definite programs and polynomial time computation is made by Raghavendra~\cite{Rag08} in the proof of his dichotomy result.}

{\bf Distribution labeled graphs:} we first view the instance $\calI$ as a collection of local distributions with mild consistency between them. For each constraint $C_i$, the numbers $\{z_{i,b}\}_{b\in\Sigma^k}$ specify a local distribution of the assignment to the $k$-variables in $C_i$, and we denote this distribution by $\nu_{C_i}$. Analogously, for each variables $\sfv\in \calV$ the numbers $\{x_{\sfv,\sigma}\}_{\sigma\in \Sigma}$ specify a distribution over the assignments to $\sfv$, and we call this distribution $\nu_{\sfv}$. Thus, the second condition in $\lp_{\calI}$ can be seen as asserting that for each constraint $C_i$ and variable in it $\sfv_{i,j}$, it holds that the marginal distribution of $\nu_{C_i}$ on $\sfv_{i,j}$ is the same as $\nu_{\sfv_{i,j}}$. 

To apply Fourier analytic tools (which play a significant role in our analysis) it is more convenient to work with the Abelian groups $\mathbb{Z}_N$ and $\mathbb{Z}_N^k$ instead of 
$\Sigma$ and $\Sigma^k$, where $N$ is large enough. 
To do that we first observe that in a solution to $\lp_{\calI}$, the values of $x_{\sfv,\sigma}$ and $z_{i,b}$ can 
be taken to be rational. Thus, we can partition the set
$\mathbb{Z}_N$ into intervals $\{I_{\sfv,\sigma}\}$ where $|I_{\sfv,\sigma}| = x_{\sfv,\sigma}\cdot N$ and define a function $q_{\sfv}$ from $\mathbb{Z}_N$ to $\Sigma$ by $q_{\sfv}(i) = \sigma$ if $i\in I_{\sfv,\sigma}$. The function $q_{\sfv}$ maps the uniform distribution over $\mathbb{Z}_N$ to $\nu_{\sfv}$, so we have successfully converted $\nu_{\sfv}$ into the uniform distribution over $\mathbb{Z}_N$. To convert $\nu_{C_i}$ to a distribution $\mu_{C_i}$ over $\mathbb{Z}_N^k$ we first sample $b\in\Sigma^k$ with probability proportional to $z_{C_i,b}$, and then sample $(i_1,\ldots,i_k)\in q_{\sfv_1}^{-1}(b_1)\times\ldots\times q_{\sfv_k}^{-1}(b_k)$ uniformly, where $\sfv_1,\ldots,\sfv_k$ are the variables in $C_i$.

At the end of this step we get an object which we call a \emph{distribution labeled graph}, namely a $k$-uniform hypergraph $G$ (the constraint structure in $\calI$) whose hyperedges are labeled by distributions over $\mathbb{Z}_N^k$ that have uniform marginal on each variable. The next step in the proof is to transform this object into a communication complexity problem.

{\bf The communication problem:} following essentially all prior works, our lower bound is ultimately proved by establishing a communication complexity lower bound for a suitable problem. Towards this end we show how to transform a distribution labeled graph into a communication problem called DIHP, where:
\begin{enumerate}
    \item Any $p$-pass, $S$-space streaming algorithm for 
    $\McspF{c-o(1)}{s+o(1)}$ can be converted into a communication protocol $\Pi$ solving DIHP, whose communication complexity is $O(pS)$.
    \item The communication complexity of DIHP is lower bounded by $\Omega(n^{1/3})$.
\end{enumerate}
When combined, the two items clearly give the hardness part of~\Cref{thm:main}, and we next discuss the DIHP problem.

\subsubsection{The Distributional Implicit Hidden Partition (DIHP)}
The DIHP problem we consider is an appropriate analog of the one considered in~\cite{KK19} (or the signal detection problem in~\cite{CGSV24}), and it is 
defined on a blow-up of the graph $G$ above. To define it we need the notion of labeled matchings. For a vertex set $V$, a matching over $V$ is a collection of $k$-uniform hyperedges over $V$ that are vertex disjoint, and the size of the matching is the number of hyperedges in it. A labeled matching is one in which each hyperedge is labeled by a $\mathbb{Z}_N^k$ element.

With this in mind, let $\mathcal{E}$ be the edge set of $G$ and $\mathcal{V}$ be the vertex set of $G$, let $K$ be a large constant and let $n\in\mathbb{N}$ be thought of as large (we think of the other parameters, such as the size of $N$ and $G$, as being constant relative to $n$). For each $\sfv\in \calV$ we consider the cloud of $\sfv$, $U_{\sfv} = \{\sfv\}\times[n]$, and then the vertex set $V=\bigcup_{\sfv} U_{\sfv}$. 
In the DIHP problem each one of $K$ players receives as input a labeled (partial) matching over $V$. The labelings of these matchings are either correlated via some global $x\in \mathbb{Z}_N^{V}$, or else are fully independent of each other. The goal of the players is to distinguish between the two cases. 

More specifically, each hyperedge $\sfe\in\mathcal{E}$ has $K$ corresponding players, so that the total number of players is $|\mathcal{E}|K$, and we label them by $(\sfe,j)$ for $\sfe\in \mathcal{E}$ and $j\in [K]$. The player $(\sfe,j)$ will receive as input a partial labeled matching on the clouds corresponding to the hyperedge $\sfe$. Namely, letting $\sfv_1,\ldots,\sfv_k$ be the vertices in $\sfe$, the player $(\sfe,j)$ will receive a random matching $M^{(\sfe,j)}$ of size $\alpha n$ from the complete $k$-partite hypergraph on $\bigcup_{i=1}^{k}U_{\sfv_i}$ (where $\alpha$ is a small constant). The labeling of that matching will be drawn differently depending on whether we are drawing a YES instance or a NO instance:
\begin{enumerate}
    \item In the distribution $\mathcal{D}_{{\no}}$, the label of each hyperedge in $M^{(\sfe,j)}$ is chosen independently uniformly at random.
    \item In the distribution $\mathcal{D}_{{\yes}}$, we first sample $x\in \mathbb{Z}_N^{\calV\times [n]}$ uniformly. Then, for each hyperedge $e'$ in $M^{(\sfe,j)}$ we sample $w_{e'}\sim \mu_{\sfe}$ independently and label the hyperedge $e'$ by $x_{|e'} - w_{e'}$, where $\mu_{\sfe}$ denotes the label distribution associated with the edge $\sfe\in \mathcal{E}$ in the distribution labeled graph $G$. 
\end{enumerate}

We first observe the relation between the DIHP problem and $\McspF{c-o(1)}{s+o(1)}$. Consider the $\cspF$ instance on variables $ \calV\times [n]$ induced by the hyperedges in the above graph that are labeled by $\vec{0}$. It is not hard to see that with high probability, for an input sampled according to $\mathcal{D}_{{\yes}}$, when interpreted as a labeling of elements from $\Sigma$ to the variables (using the maps $q_{\sfv}$), the global assignment $x$ will have value at least $c-o(1)$. 

As a sanity check, fix any $e \in M^{(\sfe, j)}$ and compute the probability that the constraint on $e$ is satisfied, conditioned on its inclusion in the CSP instance. 
The random variable $x_{|e}$ is distributed as $\mu_{\sfe}$ conditioned on label $\vec{0}$. 
By our construction of the maps $\{q_{\sfv}\}$, the resulting distribution of assignments to the variables in $e$ over the alphabet 
$\Sigma$ coincides with $\nu_{\sfe}$. Consequently, the probability that $e$ is satisfied is exactly the probability that the corresponding constraint $C_i$ from the original 
CSP instance is satisfied under the distribution $\nu_{C_i}$. Combining a concentration inequality one can conclude the statement. 

Also, using standard arguments 
one can show that with high probability, an instance sampled according to $\mathcal{D}_{{\no}}$ will have value which is at most the integral value of $\calI$ plus $o(1)$. Combining this with standard techniques, one can transfer any streaming algorithm for $\McspF{c-o(1)}{s+o(1)}$ to a protocol for DIHP with the above guarantees.

\subsubsection{The DIHP Lower Bound}
Our lower bound for the above DIHP problem follows the method of~\cite{FMW25}, which we discuss next. As a first attempt at proving a lower bound one may try to use the discrepancy method and show that for any combinatorial rectangle $R$ whose mass under $\mathcal{D}_{\no}$ is at least $2^{-C}$ it holds that 
\begin{equation}\label{eq:disc_intro}
    |\mathcal{D}_{\yes}(R) - \mathcal{D}_{\no}(R)|\leq 0.01\mathcal{D}_{\no}(R).
\end{equation}
Indeed, if true, such assertion would give a lower bound of $\Omega(C)$ on the communication complexity of DIHP. Alas, this turns out to be false, and there are in fact rectangles $R$ for which this assertion fails. Indeed, it is not hard to engineer distributions $\mu_{\sfe}$ as above and rectangles $R$ defined by constantly many coordinates, such that $R$ has $0$ mass under $\mathcal{D}_{\yes}$ and $n^{-O(1)}$ mass under $\mathcal{D}_{\no}$.

This observation leads us to consider a refinement of the discrepancy method: instead of proving~\eqref{eq:disc_intro} for all rectangles, it is sufficient to prove it for a certain family of rectangles, so long as the rectangle structure of an arbitrary communication protocol can be further decomposed into rectangles from that family (with only a mild additional cost). The family of rectangles we work with here is the family of \emph{global rectangles}, analogously to~\cite{FMW25}. For a domain $\calU$ let $\Omega^{\calU,\alpha n}$ be the set of labeled matchings over $\calU$ of size $\alpha n$, and for a labeling $\bfz$ of at most $\alpha n$ hyperedges let $\Omega_{\bfz}^{\calU,\alpha n}$ be the set of labeled matchings consistent with $z$.
A pair $(A,\bfz)$, where $A\subseteq \Omega^{\calU,\alpha n}$  and $\bfz$ is restriction assigning labels to at most $\alpha n$ hyperedges, is called global if for any labeling $\bfz'$ extending $\bfz$ it holds that:
\[
\frac{\left|A \cap \Omega_{\bfz'}^{\calU,\alpha n}\right|}{\left|\Omega_{\bfz'}^{\calU,\alpha n}\right|}\leq 2^{|\supp(\bfz')|-|\supp(\bfz)|}\frac{\left|A \cap \Omega_{\bfz}^{\calU,\alpha n}\right|}{\left|\Omega_{\bfz}^{\calU,\alpha n}\right|}.
\]
In words, any further restriction of at most $r$ labels to hyperedges increases the relative density of $A$ by factor at most $2^r$. With this definition in mind, a rectangle-restriction pair $(R,\bdzeta)$ for $R = A^{(1)}\times\ldots\times A^{(|\mathcal{E}|K)}$ and 
$\bdzeta = (\bfz^{(1)},\ldots,\bfz^{(|\mathcal{E}|K)})$ is called global if $(A^{(i)},\bfz^{(i)})$ is global for each $i$.

Our communication complexity lower bound is proved 
by combining a decomposition lemma, and a discrepancy lemma. Our decomposition lemma asserts that a communication protocol $\Pi$ with complexity $C\leq o(\sqrt{n})$ for DIHP can be converted into a $C$-round communication protocol, such that in each round only a single player speaks, and at its end the set of inputs that reach there form a global rectangle-restriction pair (see~\Cref{lem:regularity_decomposition} for a more precise formulation). At a high level this lemma reduces the communication complexity lower bound  
down to two tasks:
\begin{enumerate}
    \item {\bf Handling the structured part:} namely, showing that strategies in which players expose $o(\sqrt{n})$ coordinates of their input fail to distinguish between the distributions $\mathcal{D}_{\yes}$ and $\mathcal{D}_{\no}$.
    \item {\bf Handling the global part:} namely, showing that an inequality such as~\eqref{eq:disc_intro} holds if $R$ is a global rectangle. 
\end{enumerate}
Our analysis of the structured part refines the analysis in~\cite{FMW25}, which handled structured strategies exposing up to $o(n^{1/3})$ of the coordinates. While this refinement is not necessary for our analysis it may be useful for future research. Our argument required a slight reformulation of the decomposition lemma of~\cite{FMW25}, so we chose to include the best analysis that we are aware of.

The analysis of the global part is done via our discrepancy lemma, which roughly speaking asserts that~\eqref{eq:disc_intro} holds for global rectangles. Technically, it proceeds via proving an appropriate global hypercontractive inequality for the space $\Omega^{\calU,\alpha n}$, using it to establish a certain variant of the level-$d$ inequality for functions as $1_A\colon \Omega^{\calU,\alpha n}\to \{0,1\}$ when $A$ is a global set, and then relating the left hand side of~\eqref{eq:disc_intro} to such quantities.
While there are high-level similarities between this result and the result proved in~\cite{FMW25}, we mention a few differences:
\begin{enumerate}
    \item The setting considered herein is more general than the one in~\cite{FMW25} (involving arbitrary predicates, $k$-uniform hypergraphs and large cyclic groups as opposed to $\mathbb{F}_2$), and this leads to a few technical complications. For example, the ``structured sets'' 
    in the context of~\cite{FMW25} are linear subspaces over $\mathbb{F}_2$, which is a substantially nicer structure than the type we need to study in the current paper.
    
    \item It is possible to deduce (in a black-box way) a level-$d$ inequality for $L_2(\Omega^{\calU,\alpha n})$ from existing global hypercontractive inequalities in the literature. Such a result can be used to prove weaker communication complexity lower bounds for DIHP, and the reduction we know achieves a bound of the form $n^{\Omega_{k}(1)}$ in~\Cref{thm:main}. This is the best type of result one may hope from black-box reductions (at least from the ones that we know). Indeed, such a reduction implies a level-$d$ inequality that is more general than the one we prove, which is tight in that generality. To get the better bound as in~\Cref{thm:main} we must therefore establish a level-$d$ inequality that is tailored to our setting.\footnote{Technically, our level-$d$ inequality only bounds the contribution of a certain subset of level-$d$ functions (as opposed to the entire level-$d$ mass of the function). This was also the case in~\cite{FMW25}, but there the gap between the performance of the two approaches is smaller.}
    \item The above description is somewhat of an oversimplification of our actual argument, and we do not know how to fully decouple the analysis of the ``structured part'' and the ``global part''. Instead, we need to sew the two arguments together in a suitable way. We handle this part somewhat differently (and perhaps more cleanly) than in~\cite{FMW25}. More specifically, we do not need to study the notion of ``unrefinements'' and the way they affect the level-$d$ weight of a function, which are important in~\cite{FMW25}.
\end{enumerate}

\subsection{Open Problems}
We finish this introductory section by stating a few open problems.
\begin{enumerate}
    \item {\bf Perfect completeness:} first, it would be interesting to investigate what happens for $c=1$ in~\Cref{thm:main}. Our main algorithmic result~\Cref{thm:streaming_solve_LP} and main hardness result~\Cref{thm:hardness} do give some guarantees in that case, but we do not know whether the two values they achieve match. If it is true that $\vartheta_{\mathcal{F}}$ (defined in \Cref{def:lpcurve}) is continuous at $1$, i.e., if it is true that $\lim_{c\rightarrow 1^{-}}\vartheta_{\mathcal{F}}(c) = \vartheta_{\mathcal{F}}(1)$, then the two results would match and~\Cref{thm:main} would automatically apply to $c=1$ as well. We do not know whether $\vartheta_{\mathcal{F}}$ is necessarily continuous at $c=1$, though. It is also possible that the case of perfect completeness case requires a different algorithm (such is the case of perfect completeness in the NP-world), but we do not have any better candidate algorithm in mind.
    \item {\bf Single-pass dichotomy theorem:} it is clear from definition that the power of single-pass streaming algorithms lies somewhere between sketching and multi-pass streaming (in \Cref{fig:hiererachy}). However, it remains largely mysterious whether CSP approximation exhibits a dichotomy behavior in the single-pass model as well. It would be interesting to prove a ${\poly}(\log n)$ vs $n^{\Omega(1)}$ space dichotomy result in this setting (or even more modestly, a $n^{o(1)}$ vs $n^{\Omega(1)}$ space dichotomy result). As a starting point it would be good to find a candidate class of optimal approximation algorithms for the single-pass model.
    \item {\bf Random order:} it is not clear if multi-pass streaming algorithms may have better performance when the constraints of CSP instances are given in random order than in the worst-case order setting, as our lower bound does not extend to the case of random-ordered inputs.
    \item {\bf Derandomization:} the approximation algorithm in~\Cref{thm:main} is randomized, and it will be interesting to see if similar guarantees can be achieved via a deterministic approximation algorithm with similar limit on space and number of passes.
    \item {\bf Other problems:} can the connection between sublinear-time models and sublinear-space models be extended to other problems (see \Cref{subsec:hierarchy}), such as approximating maximum matchings?
\end{enumerate}

\section{Preliminaries}

\subsection{General Notations}\label{subsec:general_notations}

In this subsection we summarize general notational conventions used throughout the paper. Additional notation will be introduced as needed, typically within dedicated ``Notation'' environments.

\paragraph{Probability.}  
For a finite set $\Lambda$, we write $\Exs{x \in \Lambda}{\cdot}$ and $\Prs{x \in \Lambda}{\cdot}$ to denote expectation and probability, respectively, when $x$ is drawn uniformly at random from $\Lambda$.  
If $x$ is sampled according to a specific distribution $\calD$ over $\Lambda$, we write $x \sim \calD$ in place of $x \in \Lambda$. A \emph{probability mass function} on $\Lambda$ is a function $p:\Lambda\rightarrow [0,\infty)$ such that $\sum_{x\in \Lambda}p(x)=1$, while a \emph{probability density function} is a function $f:\Lambda\rightarrow [0,\infty)$ such that $\Exs{x\in \Lambda}{f(x)}=1$. A \emph{right stochastic matrix}, or a \emph{Markov kernel}, is a matrix in which each row is a probability mass function on the set of columns. 

\paragraph{Hilbert space.}  
For a finite set $\Lambda$, we denote by $L^2(\Lambda)$ the (finite-dimensional) Hilbert space of complex-valued functions on $\Lambda$, equipped with the inner product
\[
\langle f, g \rangle := \Exu{x \in \Lambda}{f(x)\, \overline{g(x)}}.
\]

\paragraph{Fourier analysis.}  
We denote the finite cyclic group $\bZ / N\bZ$ by $\bZ_N$, where $N\geq 2$ is an integer. Throughout the paper, the capital letter $N$ is reserved exclusively for this notation. For any finite index set $\Lambda$, the collection of Fourier characters on the product group $\bZ_N^{\Lambda}$ is indexed by $\bZ_N^{\Lambda}$ itself. More precisely, for $b \in \bZ_N^{\Lambda}$, the associated character function $\chi_b : \bZ_N^{\Lambda} \to \bC$ is defined by
\[
\chi_b(x) := \exp\left( \frac{2\pi \mathrm{i}}{N} \sum_{v \in \Lambda} b_v x_v \right),
\]
where $\mathrm{i}$ denotes the imaginary unit.

\paragraph{Vectors and maps.}  
For a vector $x \in \bZ_N^{\Lambda}$ or $x \in [0,1]^{\Lambda}$, we denote its coordinates by subscripts: $x_v$ for each $v \in \Lambda$.  
A related notion is that of a \emph{map} $\bfy : \Lambda \to \Lambda'$. We use boldface symbols for maps, especially when their images are themselves vectors, to distinguish them from ordinary vectors.  
For $v \in \Lambda$, the value of the map at $v$ is denoted by $\bfy(v)$.  
The collection of all such maps is denoted by $\Map{\Lambda}{\Lambda'}$.

\paragraph{Support sets.}
Let $\Lambda'$ be a domain containing a distinguished \emph{nullity element}. For either a vector $x \in (\Lambda')^{\Lambda}$ or a map $\bfy \in \Map{\Lambda}{\Lambda'}$, the \emph{support} of $x$ or $\bfy$ --- denoted $\supp(x)$ or $\supp(\bfy)$ --- is the set of elements $v \in \Lambda$ such that $x_v$ or $\bfy(v)$ is not equal to the nullity element. For example, when $\Lambda' = \bZ_N$, the nullity element is the additive identity $0$. In some cases, the domain $\Lambda'$ is taken to be a disjoint union of an Abelian group and a special symbol --- such as $\bZ_N^k \cup \{\nil\}$ --- in which case the nullity element is the special symbol $\nil$, rather than the identity of the group.

\paragraph{Degree decomposition.} For a function $f:\ZmodN^{\Lambda}\rightarrow\bC$ and a nonnegative integer $d\leq |\Lambda|$, we will write the degree-$d$ part of $f$ as
\[
f^{=d}:=\sum_{b\in \ZmodN^{\Lambda},\, |\supp(b)|=d}\left\langle f,\chi_{b}\right\rangle\cdot \chi_{b}.
\]
We then have $f=\sum_{d=0}^{|\Lambda|}f^{=d}$.

\paragraph{CSPs and hypergraphs.}  
Throughout the paper, $\Sigma$ denotes the CSP alphabet, and the lowercase letter $k$ always refers to the arity of predicates. The calligraphic letter $\calF$ always denotes a nonempty finite set of predicates mapping from $\Sigma^{k}$ to $\{0,1\}$. 
Variables in a CSP instance are often identified with vertices of a hypergraph. In many constructions, these hypergraphs undergo a \emph{blow-up}, in which each original vertex is replaced by $n$ copies. We adopt the following notational convention: pre-blowup vertices and hyperedges are denoted using sans-serif font (e.g., $\mathsf{v}$ and $\mathsf{e}$), while post-blowup vertices and hyperedges are written in standard math font (e.g., $v$ and $e$). When the context is clear, hyperedges are sometimes simply refered to as edges. A set of hyperedges in a \(k\)-uniform hypergraph is said to contain a cycle if there exist \(\ell\) hyperedges in the set that cover at most \(\ell(k-1)\) vertices, for some \(\ell \geq 1\).

\paragraph{One-wise independence.} A probability distribution over $\ZNk$ is called \emph{one-wise independent} if its marginal on each of the $k$ coordinates is the uniform distribution on $\ZmodN$. For any one-wise independent distribution $\mu$, we denote by $\mu(\cdot)$ the probability mass function of $\mu$.

\subsection{Streaming Algorithms}

Suppose $\calF\subseteq \{f:\Sigma^{k}\rightarrow\{0,1\}\}$ is a fixed predicate family, and $\calV$ is a fixed variable set. Let 
$$\calF_{\calV}:=\left\{(\sfe,f):\sfe\in \calV^{k},\,f\in \calF\right\}$$
be the set of all possible constraints placed on $\calV$ using predicates from $\calF$. 
\begin{definition}
A deterministic space-$S$ streaming algorithm for $\mcsp$ over the variable set $\calV$ is specified by:
\begin{itemize}
    \item a transition function $\calT : \{0,1\}^{S} \times \calF_{\calV} \rightarrow \{0,1\}^{S}$, and
    \item an output function $\calO : \{0,1\}^{S} \rightarrow [0,1]$.
\end{itemize}
When the algorithm reads a constraint $C \in \calF_{\calV}$ while in memory state $x \in \{0,1\}^{S}$, it updates its memory to $\calT(x, C)$. After processing all constraints in the input stream and reaching a final memory state $x \in \{0,1\}^{S}$, it outputs the value $\calO(x)$.
\end{definition}

\begin{definition}\label{def:randomized_streaming}
A randomized space-$S$ streaming algorithm for $\mcsp$ over the variable set $\calV$ is specified by:
\begin{itemize}
    \item a transition function $\calT : \{0,1\}^{S} \times \calF_{\calV} \times \{0,1\}^{r} \rightarrow \{0,1\}^{S}$, and
    \item an output function $\calO : \{0,1\}^{S} \times \{0,1\}^{r} \rightarrow [0,1]$,
\end{itemize}
where $r$ is a nonnegative integer with $r \leq S$. When the algorithm reads a constraint $C \in \calF_{\calV}$ while in memory state $x \in \{0,1\}^{S}$, it samples a uniformly random string $z \in \{0,1\}^{r}$ and updates its memory state to $\calT(x, C, z)$. After processing the entire input stream and reaching a final memory state $x \in \{0,1\}^{S}$, it samples a fresh random string $z \in \{0,1\}^{r}$ and outputs the value $\calO(x, z)$.
\end{definition}

A more general notion of a randomized space-$S$ streaming algorithm is given by a probability distribution over deterministic space-$S$ streaming algorithms. Our lower bound results are proved against this broader model, whereas our algorithmic results use only the weaker model of \Cref{def:randomized_streaming}. Moreover, our algorithms are \emph{uniform} in the sense that there exists a $O(\log n)$-space Turing machine that, given the size-$n$ variable set $\calV$, outputs circuit representations of the transition and output functions.

\subsection{Concentration Inequalities}
We will make use of the following standard concentration bound known as Hoeffding's inequality:
\begin{proposition}[{\cite[Theorem D.2]{mohri2018foundations}}]\label{prop:Hoeffding}
Let $X_{1},\dots,X_{n}$ be independent random variables taking values in $[0,1]$. Then for any $\varepsilon\geq 0$ we have
\[
\Pr{\sum_{i=1}^{n}X_{i}\geq \varepsilon+\sum_{i=1}^{n}\Ex{X_{i}}}\leq \exp\left(-\frac{2\varepsilon^{2}}{n}\right).
\]
\end{proposition}

We also need the following martingale version of the Chernoff bound, which has previously been used for lower bounds against streaming algorithms by \cite{KK19,CGSV24}. 
\begin{proposition}[{\cite[Lemma 2.8]{CGSV24}}]\label{prop:concentration}
    Let $X_1,\dots,X_{n}$ be Bernoulli random variables such that for every $i\in [n]$, $\Ex{X_i\mid X_1,\dots, X_{i-1}}\leq p_i$ for some $p_i\in (0,1)$. For any $\Delta>0$, we have 
    \begin{align*}
        \Pr{\sum_{i=1}^{n}X_{i}\geq \varepsilon +\sum_{i=1}^{n}p_{i}} \leq \exp\left(-\frac{\varepsilon^2}{2\varepsilon+2\sum_{i=1}^{n}p_{i}}\right). 
    \end{align*}
\end{proposition}
\subsection{Hypercontractivity}
Hypercontractive inequalities on product spaces have been crucial tools in establishing streaming lower bounds for approximating CSPs. We need the following version in this paper:
\begin{proposition}
    [{\cite[Theorem 10.21]{odonnell2021analysisbooleanfunctions}}]\label{prop:classical_hypercontractivity}
    For any function $f:\bZ_N^\Lambda \rightarrow \bR$ with degree at most $d$ and any real number $q\geq 2$, we have
    \begin{align*}
        \lVert f \rVert_q \leq \left(\sqrt{N(q-1)} \right)^d \lVert f \rVert_2
    \end{align*}
\end{proposition}
As is standard in many applications, the above hypercontractivity result is used to obtain the following level-$d$ inequality.
\begin{proposition}\label{prop:level_d_classical}
    For any function $f: \bZ_N^{\Lambda} \rightarrow \bR$ and positive integer $d$, we have 
    \begin{align*}
        \left\| f^{=d}\right\|_2^2\leq \|f\|_{1}^{2}\cdot\left( 12N\log\left(\frac{2\|f\|_{2}}{\|f\|_{1}}\right)\right)^d . 
    \end{align*}
\end{proposition}
\begin{proof}
    For any $q\geq 2$, we have
    \begin{align*}
        \left\|f^{=d} \right\|_2^2 = \left\langle f,f^{=d}\right\rangle &\leq \left\| f^{=d}\right\|_q \cdot\lVert f\rVert_{q/(q-1)}\leq  \left\| f^{=d}\right\|_q \cdot \lVert f\rVert_{1}^{(q-2)/q}\lVert f\rVert_{2}^{2/q}\\
        &\leq \left(\sqrt{(q-1)N}\right)^d\left\|  f^{=d}\right\|_2 \cdot \lVert f\rVert_{1}^{(q-2)/q}\lVert f\rVert_{2}^{2/q},
    \end{align*}
    where the second and third transitions are by H\"{o}lder's inequality, and the fourth transition is by \Cref{prop:classical_hypercontractivity}. Thus, we have 
    \begin{align}\label{ineq:level_d_classical}
         \left\| f^{=d} \right\|_2^2 \leq \|f\|_{1}^{2}\cdot \left((q-1)N\right)^d \left(\frac{\|f\|_{2}}{\|f\|_{1}}\right)^{4/q}. 
    \end{align}
    Taking $q = 4\log \left(2\lVert f\rVert_{2}/\|f\|_{1}\right)$ yields the conclusion.
\end{proof}

\section{The Approximability Threshold}
\label{sec:LP}

As outlined in the introduction, the approximability threshold function $\lpcurve$ in \Cref{thm:main} is given by the basic linear programming relaxation of $\mcsp$. In \Cref{subsec:linearprogramming} we formally define $\lpcurve$, and in \Cref{subsec:main_results} we present the main algorithmic and hardness results (\Cref{thm:streaming_solve_LP,thm:hardness}) living on the two sides of the approximability threshold. 

Finally, in \Cref{subsec:examples}, we work out several concrete examples of predicate families $\calF$ for which the threshold function $\lpcurve$ can be computed explicitly, thereby determining the exact approximation ratio of $\mcsp$ in the multipass streaming model.

\subsection{The Basic Linear Program}\label{subsec:linearprogramming}

For any instance $\calI\in \cspF$ of the CSP maximization problem we recall the linear programming relaxation of $\mcsp$ from the introduction, which is termed ``the basic linear program'' and abbreviated as $\lp$.

\begin{definition}\label{def:BasicLP}
Let ${\cal{F}}\subseteq \{ f: \Sigma^k\rightarrow \{0,1\}\}$ be a predicate family and let $\calI=(\calV, \calC)$ be a $\cspF$ instance. Write $\calC=(C_{1},\dots,C_{m})$, and $C_{i}=((\sfv_{i,1},\dots,\sfv_{i,k}),f_{i})$ for each $i\in [m]$. We define $\lp_{\calI}$ to be the following linear program, with variables $(x_{\sfv,\sigma})_{\sfv\in \calV,\,\sigma\in \Sigma}$ and $(z_{i,b})_{i\in [m],\, b\in \Sigma^{k}}$.
\begin{tcolorbox}[
  enhanced,
  breakable,
  title={$\lp_{\calI}$ for $\calI=(\calV, \calC)$},
  fonttitle=\bfseries,
  colframe=black,
  colback=white,
  boxrule=0.6pt,
  arc=6pt,
  left=6pt,
  right=6pt,
  top=4pt,
  bottom=4pt
]
\begin{alignat*}{3}
\textup{maximize} \quad     &&  \frac{1}{m}\sum_{i=1}^{m} 
            \sum_{b\in\Sigma^{k}} f_i(b)\,z_{i,b} \\
\textup{subject to} \quad && \sum_{\sigma\in\Sigma} x_{\sfv,\sigma}&=1
            \quad &\forall\,\sfv\in \calV \\
           && \sum_{b\in\Sigma^{k}}\mathbbm{1}\{b_{j}=\sigma\}\cdot
             z_{i,b}&=x_{\sfv_{i,j},\,\sigma}
            \quad &\forall\,i\in [m],\,j\in [k],\,\sigma\in \Sigma\\
           && x_{\sfv,\sigma}&\ge 0
            \quad &\forall\,\sfv\in \calV,\,\sigma\in\Sigma \\
           && z_{i,b}&\ge 0
            \quad &\forall\,i\in [m],\,b\in\Sigma^{k}
\end{alignat*}
\end{tcolorbox}
\end{definition}
In words, the intention
here is that for each $\sfv\in\calV$, the variables 
$\{x_{\sfv,\sigma}\}_{\sigma\in \Sigma}$ represent a distribution over the labels of $\sfv$, and for each $i\in [m]$, the variables $\{z_{i,b}\}_{b\in\Sigma^k}$ 
represent a distribution over the assignments to $C_i$. The objective function counts
the total mass that is put on satisfying assignments, 
and the constraints asserts that the marginal distribution of $\{z_{i,b}\}_{b\in\Sigma^k}$ on each variable $\sfv$ appearing in $C_i$ is consistent with the distribution $\{x_{\sfv,\sigma}\}_{\sigma\in \Sigma}$.

\begin{observation}\label{obs:BasicLP}
Note that assigning value $1/|\Sigma|$ to every $x$ variable and $1/|\Sigma|^{k}$ to every $z$ variable always satisfies all the constraints in the linear program. Therefore, the feasible region of the linear program is nonempty. Furthermore, using $\sum_{b\in \Sigma^{k}}z_{i,b}=1$ for 
all $i\in [m]$ and the Booleanity of $f$ implies $\sum_{b\in \Sigma^{k}}f_{i}(b)z_{i,b}$ is in $[0,1]$ for all $i\in [m]$. Therefore, as the average of this quantity over $i\in [m]$, the objective function also always takes value in $[0,1]$ on feasible solutions. We 
also note that any integral solution $\tau$ to $\calI$ naturally corresponds to an assignment to $\lp_{\calI}$ with objective value $\val_{\calI}(\tau)$.
\end{observation}
\begin{notation}\label{notation:LP_value}
For a family of predicates ${\cal{F}}\subseteq \{ f: \Sigma^k\rightarrow \{0,1\}\}$ and an instance $\calI\in \cspF$, we let $\lpval_{\calI}$ denote the optimal value of the linear program $\lp_{\calI}$.
\end{notation}

We are now ready to define the approximability threshold function $\vartheta_{\calF}:[0,1]\rightarrow[0,1]$.

\begin{definition}\label{def:lpcurve}
Let ${\cal{F}}\subseteq \{ f: \Sigma^k\rightarrow \{0,1\}\}$ be a fixed predicate family. We define a function $\vartheta^{*}_{\calF}:[0,1]\rightarrow[0,1]$ as follows. If $\calF$ consists only of the all-0 predicate, then let $\vartheta^*_{\calF}(c)=1$ for any $c\in [0,1]$. Otherwise, we let
\begin{align*}
    \lpcurve^*(c):= \inf_{\substack{\calI\in \cspF\textup{ s.t. }\lpval_{\calI}\geq c}} \val_{\calI}.
\end{align*}
The function $\lpcurve:[0,1]\rightarrow[0,1]$ is then defined by $\lpcurve(c):=\min\left\{c, \lpcurve^*(c)\right\}$, for $c\in [0,1]$.
\end{definition}
It is not hard to see that the function $\vartheta^*_{\calF}$ has the following nice property.
\begin{lemma}\label{lem:convexity}
    The function $\vartheta_{\calF}^*$ is monotone nondecreasing and convex on $[0,1]$.
\end{lemma}
\begin{proof}
    Monotonicity of $\vartheta^*_{\calF}$ is immediate from its definition. 
    To prove convexity, it suffices to show that for any $c_0,c_1 \in [0,1]$ and $t\in (0,1)$,
    \[
        \vartheta_{\calF}^*\bigl(t c_0+(1-t) c_1\bigr)
        \leq t\, \vartheta_\calF^* (c_0) + (1-t)\, \vartheta_{\calF}^*(c_1).
    \]
    Without loss of generality, assume $c_0<c_1$. If $\calF$ contains no nonzero predicate, then $\vartheta^*_{\calF}$ is constant and hence convex, so we may assume otherwise.

    By the definition of $\vartheta^{*}_{\calF}$, for any $\varepsilon>0$ there exist finite instances $\calI_0,\calI_1 \in \cspF$ such that  
    \begin{enumerate}[label=(\arabic*)]
        \item $\lpval_{\calI_0} \ge c_0$ and $\lpval_{\calI_1} \ge c_1$;
        \item $\val_{\calI_0} \le \vartheta_{\calF}^*(c_0)+\varepsilon$ and 
        $\val_{\calI_1} \le \vartheta_{\calF}^*(c_1)+\varepsilon$. 
    \end{enumerate}

    Choose a rational $t'=\frac{p}{q}\in(t-\varepsilon,t]$ with $p,q \in \mathbb{N}$, and form a new instance $\calI'$ as the disjoint union of $p$ copies of $\calI_0$ and $q-p$ copies of $\calI_1$. Then
    \[
        \lpval_{\calI'} = \frac{p}{q}\,\lpval_{\calI_0}+\frac{q-p}{q}\,\lpval_{\calI_1} 
        \ge t c_0+(1-t)c_1,
    \]
    and
    \begin{align*}
        \val_{\calI'} 
        &= \frac{p}{q}\,\val_{\calI_0} + \frac{q-p}{q}\,\val_{\calI_1} \\
        &\le t'\,\bigl(\vartheta_{\calF}^*(c_0)+\varepsilon\bigr)
           + (1-t')\,\bigl(\vartheta_{\calF}^*(c_1)+\varepsilon\bigr) \\
        &\le t\, \vartheta_\calF^*(c_0) + (1-t)\,\vartheta_{\calF}^*(c_1) + 2\varepsilon,
    \end{align*}
    where the last inequality uses $t'\in(t-\varepsilon,t]$.

    Since $\varepsilon>0$ was arbitrary, the desired convexity inequality follows.
\end{proof}
\Cref{lem:convexity} has the following immediate corollary.
\begin{corollary}\label{lem:weak_continuity}
    The function $\vartheta^*_\mathcal{F}:[0,1]\rightarrow[0,1]$ is continuous on $[0,1)$. As a consequence, the threshold function $\vartheta_{\calF}:[0,1]\rightarrow[0,1]$ is also continuous on $[0,1)$.  
\end{corollary}

\subsection{Main Results}\label{subsec:main_results}

In this subsection, we present the main algorithmic and hardness results of the paper. 

On the algorithmic side, we show that for approximating the value of CSP instances, a multipass-streaming algorithm can achieve performance matching the basic linear programming relaxation of \Cref{def:BasicLP}. Importantly, we do not claim that such an algorithm can directly compute the LP value $\lpval_{\calI}$ for a given instance $\calI \in \cspF$. 
Instead, as shown in \Cref{sec:algorithm}, our algorithm estimates the LP value of a suitably modified instance whose \emph{value} is close to that of $\calI$, thereby achieving the same approximation ratio as the integrality gap of $\lp$.

\begin{restatable}[Main algorithm]{theorem}{thmalgo}\label{thm:streaming_solve_LP}
    For any fixed predicate family ${\cal{F}}\subseteq \{ f: \Sigma^k\rightarrow \{0,1\}\}$, fixed completeness parameter $c\in [0,1]$ and fixed error parameter $\varepsilon \in (0,1)$, there exists a randomized streaming algorithm $\mathcal{A}$ using $O_{\varepsilon}(\log n)$ space and $O_{\varepsilon} (1)$ passes such that for any instance $\calI\in \cspF$,
    \begin{enumerate}[label = (\arabic*)]
    \item if $\val_{\calI}\geq c+\varepsilon$ then
    $\Prs{\calA}{\calA(\calI)=1}\geq 2/3$;
    \item if $\val_{\calI}\leq \lpcurve(c)-\varepsilon$ then $\Prs{\calA}{\calA(\calI)=0}\geq 2/3$.
    \end{enumerate}
\end{restatable}

On the hardness side, we show that any integrality gap instance of $\lp$ serves as a witness for the hardness of approximation in the multipass-streaming setting. Consequently, no efficient multipass-streaming algorithm can achieve an approximation ratio better than the integrality gap of $\lp$, establishing the optimality of our algorithm in \Cref{thm:streaming_solve_LP}.

\begin{restatable}[Main hardness]{theorem}{thmhardness}\label{thm:hardness}
Fix a nonempty instance $\calI\in \cspF$. Let $s:=\val_{\calI}$ and $c:=\lpval_{\calI}$. Then the following statements hold:
\begin{enumerate}[label = (\arabic*)]
\item If $c<1$, then for any fixed error parameter $\varepsilon\in (0,1)$, any $p$-pass streaming algorithm for $\McspF{c-\varepsilon}{s+\varepsilon}$ requires $\Omega_{\varepsilon}(n^{1/3}/p)$ bits of memory.
\item If $c=1$, then for any fixed error parameter $\varepsilon\in (0,1)$, any $p$-pass streaming algorithm for $\McspF{1}{s+\varepsilon}$ requires $\Omega_{\varepsilon}(n^{1/3}/p)$ bits of memory.
\end{enumerate}
\end{restatable}
The proof of Theorem \ref{thm:streaming_solve_LP} is given in Section \ref{sec:algorithm} and the proof of Theorem \ref{thm:hardness} will take up \Cref{sec:communication_game,sec:communication_lower_bound,sec:global_rectangle,sec:Fourier_decay}.  

Our main result, \Cref{thm:main} restated below, follows as a consequence of \Cref{thm:streaming_solve_LP,thm:hardness}.

\mainthm*
\begin{proof}[Proof]
    Let the threshold function $\vartheta_{\calF}$ be as defined in \Cref{def:lpcurve}. Its continuity on $[0,1)$ is already proved in \Cref{lem:weak_continuity}.
    
    For fixed rational numbers $c\in (0,1)$ and $s\in (0,\vartheta_\calF(c))$, there exists $c^*$ such that $0<c^*<c$, and $\vartheta_\calF(c^*)>s$, due to the continuity of $\vartheta_\calF$ on $[0,1)$. Applying \Cref{thm:streaming_solve_LP} with the completeness parameter $c^{*}$ and the error parameter $\varepsilon:= \min \{c-c^*, \vartheta_\calF(c^*) - s\}$, we obtain an $O_\varepsilon(1)$-pass, $O_\varepsilon(\log n)$-space algorithm that solves the gap problem $\McspF{c^*+\varepsilon}{\vartheta_\calF(c^*)-\varepsilon}$ with probability at least $2/3$. Therefore, it also solves $\McspF{c}{s}$ with probability at least $2/3$, as desired. 

    For fixed rational numbers $c\in (0,1)$ and $s\in (\vartheta_\calF(c),c)$, similarly, there exists $c^*$ such that $c<c^*<1$ and $\vartheta_\calF(c^*)<s$, due to the continuity of $\vartheta_\calF$. Since $$\vartheta_{\calF}(c^*)<s<c<c^*,$$
    by the definition of $\vartheta_{\calF}$, we also have $\vartheta^*_{\calF}(c^*)=\vartheta_{\calF}(c^*)$. By the definition of $\vartheta^*_{\calF}$, there exists an instance $\calI\in \cspF$ such that $\lpval_{\calI}\geq c^*$ and $\val_{\calI}< s$. Applying the first statement of Theorem \ref{thm:hardness} with the gap instance $\calI$ and the error parameter $\varepsilon:=\min\{\lpval_{\calI}-c, s-\val_{\calI}\}$, we know that any $p$-pass streaming algorithm that solves $\McspF{\lpval_{\calI}-\varepsilon}{\val_{\calI}+\varepsilon}$ with probability at least $2/3$ requires $\Omega_{\varepsilon}(n^{1/3}/p)$ bits of memory. By comparison of parameters, the same lower bound holds for $\McspF{c}{s}$. 
\end{proof}

\begin{remark}
    The second statement of \Cref{thm:hardness} implies that for any \(s \in \big(\vartheta_\calF(1),1\big)\), every \(p\)-pass streaming algorithm for \(\McspF{1}{s}\) must use \(\Omega\big(n^{1/3}/p\big)\) bits of space. On the other hand, the algorithmic result in \Cref{thm:streaming_solve_LP} guarantees efficient streaming algorithms for $\McspF{1}{s}$ only when $s < \lim_{c \to 1^{-}} \vartheta_{\calF}(c)$. Since it remains unknown whether \(\lim_{c\to 1^{-}}\vartheta_{\calF}(c) = \lpcurve(1)\), i.e., whether
    \(\vartheta_\calF\) is continuous at \(1\), we cannot yet establish a full dichotomy for approximating satisfiable CSPs. This motivates the following open question.
\end{remark}

\begin{question}
Is it true that for any alphabet $\Sigma$, arity $k$ and predicate family $\calF\subseteq \{f:\Sigma^{k}\rightarrow \{0,1\}\}$, the function $\lpcurve$ (defined in \Cref{def:lpcurve}) is continuous at 1?
\end{question}

\subsection{Examples}\label{subsec:examples}
In this section, we discuss two example CSP problems, $\mathsf{Max}$-$\mathsf{DICUT}$ and $\mathsf{MAX}$-$\mathsf{2SAT}$, and determine their threshold function $\vartheta_\calF$. In particular, we show that the former has approximation ratio $1/2$ (as already proved by \cite{saxena2025streaming,FMW25}) and the latter has approximation ratio $3/4$. 

\subsubsection{MAX-DICUT}\label{subsubsec:DICUT}

In \Cref{subsubsec:DICUT}, we consider the alphabet $\Sigma=\{0,1\}$ and the singleton predicate family $\calF=\{f\}$, where $f:\{0,1\}^{2}\rightarrow\{0,1\}$ is given by $f(\sigma_{1},\sigma_{2})=1$ if and only if $\sigma_{1}=1$ and $\sigma_{2}=0$. In this case, the problem $\mcsp$ is also known as $\mathsf{Max}$-$\mathsf{DICUT}$. 

Let $\calI=(\calV,\calC)$ be a $\mathsf{Max}$-$\mathsf{DICUT}$ instance. We write $C=(C_{1},\dots,C_{m})$ and $C_{i}=((\sfv_{i,1},\sfv_{i,2}),f)$. The basic linear program $\lp_{\calI}$ (defined in \Cref{def:BasicLP}) can be simplified into the following:

\begin{tcolorbox}[
  enhanced,
  breakable,
  title={BasicLP for $\mathsf{Max}\text{-}\mathsf{DICUT}$ instance $\calI=(\calV,\calC)$},
  fonttitle=\bfseries,
  colframe=black,
  colback=white,
  boxrule=0.6pt,
  arc=6pt,
  left=6pt,
  right=6pt,
  top=4pt,
  bottom=4pt
]
\[
\begin{alignedat}{3}
\text{maximize}\quad     & \frac{1}{m}\sum_{i=1}^{m} z_i \\
\text{subject to}\quad
& 0\leq z_i \le x_{\sfv_{i,1}}, \quad && z_i \le 1 - x_{\sfv_{i,2}} \quad && \forall\, i\in[m]\\
& 0 \le x_{\sfv} \le 1     &&                               && \forall\, \sfv\in\calV
\end{alignedat}
\]
\end{tcolorbox}

For Max-DICUT, it is straightforward to verify that the above formulation of the basic linear program is equivalent to that in \Cref{def:BasicLP}. Specifically, the variable $x_{\sfv}$ here corresponds to $x_{\sfv,1}$ in \Cref{def:BasicLP}, while $z_{i}$ corresponds to $z_{i,(1,0)}$.

The key observation is that every vertex of the polytope defined by the constraints of $\lp$ has coordinates taking values only in $\{0,\tfrac{1}{2},1\}$.

\begin{lemma}\label{lem:solution_structure}
    For any $\mathsf{Max}\text{-}\mathsf{DICUT}$ instance $\calI=(\calV,(C_{1},\dots,C_{m}))$, the linear program $\lp_{\calI}$ has an optimal solution $\left((x^*_{\sfv})_{\sfv\in \calV}, (z^*_i)_{i\in [m]}\right)$ (achieving value $\lpval_{\calI}$) such that $x_{\sfv},z_{i}\in \{0,\frac{1}{2},1\}$ for any $\sfv\in\calV$ and $i\in [m]$. 
\end{lemma}
\begin{proof}
Consider the polytope in $\mathbb{R}^{\calV \cup [m]}$ defined by the constraints of $\lp_{\calI}$:
\[
P := \Big\{ \big((x_{\sfv})_{\sfv \in \calV}, (z_i)_{i \in [m]}\big) \in [0,1]^{\calV \cup [m]} \;\Big|\; z_i \le \min\left(x_{\sfv_{i,1}}, 1 - x_{\sfv_{i,2}}\right) \;\; \forall\, i \in [m] \Big\}.
\]
We claim that every vertex of $P$ has all coordinates in $\{0, \tfrac{1}{2}, 1\}$, which yields the lemma.

Let $\big((x^*_{\sfv})_{\sfv \in \calV}, (z^*_i)_{i \in [m]}\big)$ be a vertex of $P$.  
Define the set of non-integral variables
\[
\calV' := \{\sfv \in \calV \mid x^*_{\sfv} \notin \{0,1\} \},
\]
and construct a graph $G'$ with vertex set $\calV'$ and edge set
\[
\calE' := \big\{ \{\sfv_{i,1}, \sfv_{i,2}\} \;\big|\; i \in [m], \;
x^*_{\sfv_{i,1}} = z^*_i = 1-x^*_{\sfv_{i,2}} \notin \{0,1\} \big\}.
\]

Since $\big((x^*_{\sfv}), (z^*_i)\big)$ is a vertex of the polytope, the system of linear equations
\begin{equation}\label{eq:DICUT_equations}
x_{\sfv_1} + x_{\sfv_2} = 1 \qquad \text{for all } \{\sfv_1,\sfv_2\} \in \calE'
\end{equation}
must have a unique solution for $(x_{\sfv})_{\sfv \in \calV'}$.\footnote{%
Strictly speaking, one should consider the full system of linear equations in all $|\mathcal{V}|+m$ variables (both the $x$- and $z$-variables) obtained by taking those linear inequality constraints defining $P$ that are tight at $(x^*,z^*)$.  
Since $(x^*,z^*)$ is a vertex of $P$, this system has a unique solution.  
It is, however, straightforward to eliminate the $z$-variables (and the variables $x_{\sfv}$ for $\sfv\in \calV\setminus\calV'$) from this system, yielding precisely the equations in~\eqref{eq:DICUT_equations}.%
}
This is possible only if $G'$ is connected and non-bipartite; otherwise, the solution space would have at least one degree of freedom.  
The existence of an odd cycle in $G'$ forces $x^*_{\sfv} = \tfrac{1}{2}$ for every vertex $\sfv$ on the cycle, and by connectedness, $x^*_{\sfv} = \tfrac{1}{2}$ for all $\sfv \in \calV'$.  
Thus every coordinate of the vertex lies in $\{0, \tfrac{1}{2}, 1\}$.
\end{proof}
\begin{corollary}\label{cor:dicutcurve}
    For any $\mathsf{Max}\text{-}\mathsf{DICUT}$ instance $\calI$, we have $\val_{\calI}\geq \frac{3}{2}\cdot \lpval_{\calI}-\frac{1}{2}$. 
\end{corollary}
\begin{proof}
    Using Lemma \ref{lem:solution_structure}, we obtain a solution $\big((x^*_{\sfv})_{\sfv \in \calV}, (z^*_i)_{i \in [m]}\big)$ with $x^*_{\sfv}, z^*_i\in \{0,\frac{1}{2},1\}$ for all $\sfv\in \calV,i\in [m]$ such that
    \begin{equation}\label{eq:lpval_dicut}
    \frac{1}{m}\sum_{i=1}^{m}\min\left(x^*_{\sfv_{i,1}}, 1 - x^*_{\sfv_{i,2}}\right)=\lpval_{\calI}.
    \end{equation}
    We define a random integral assignment $\tau:\calV\rightarrow\{0,1\}$ by independently assigning $\tau(\sfv)=1$ with probability $x^*_{\sfv}$. For any constraint $C_i =(\sfv_{i,1},\sfv_{i,2})$, we know that $\tau$ satisfies $C_{i}$ with probability $x^*_{\sfv_{i,1}}\left(1-x^*_{\sfv_{i,2}}\right)$.
    We claim that
    \[\Exu{\tau}{\val_{\calI}(\tau)}\geq \frac{3}{2}\cdot \lpval_{\calI}-\frac{1}{2},
    \]
    which yields the conclusion. Indeed, due to \eqref{eq:lpval_dicut}, it suffices to verify that
    \[
    x_{1}(1-x_{2})\geq \frac{3}{2}\cdot \min(x_{1},1-x_{2})-\frac{1}{2}
    \]
    for all $x_{1},x_{2}\in \{0,\frac{1}{2},1\}$, which is straightforward.
\end{proof}

We now prove~\Cref{thm:DICUT}, restated below.
\DICUTthm*
\begin{proof}
    Observe that every nonempty instance $\calI$ of $\mathsf{Max}$-$\mathsf{DICUT}$ satisfies $\val_\calI \geq 1/4$, because the random assignment where each variable value is independently and uniformly sampled from $\{0,1\}$ satisfies each constraints with probability $1/4$. By the definition of $\vartheta_{\calF}$, this implies that $\vartheta_{\calF}(c)=c$ for $c\in [0,\frac{1}{4}]$.

    For each integer $n\geq 1$, consider the $\mathsf{Max}$-$\mathsf{DICUT}$ instance $\calI_n$ on the variable set $[n]$ where for each pair $\{i,j\}\subseteq [n]$, we have two DICUT constraints on $(i,j)$ and $(j,i)$. It is easy to see that 
    \begin{align*}
        \val_{\calI_n} = \frac{\lfloor\frac{n}{2}\rfloor\cdot \lceil\frac{n}{2}\rceil}{2\binom{n}{2}}  =1/4+ o(1). 
    \end{align*}
    Since $\lpval_{\calI_{n}}\geq \frac{1}{2}$ (one can assign value $\frac{1}{2}$ to every variable of $\lp_{\calI}$), this implies that $\vartheta_\calF(c) = 1/4$ for all $c\in [\frac{1}{4},\frac{1}{2}]$.

    Applying \Cref{lem:convexity} to $\lpcurve(\frac{1}{2})=\frac{1}{4}$ and $\lpcurve(1)\leq 1$, we obtain $\vartheta_{\calF}(c)\leq \frac{3}{2}c -\frac{1}{2}$ for $c\in [\frac{1}{2},1]$. Together with \Cref{cor:dicutcurve}, this implies $\vartheta_{\calF}(c)= \frac{3}{2}c -\frac{1}{2}$ for $c\in [\frac{1}{2},1]$. 
\end{proof}

\subsubsection{MAX-2SAT}\label{subsubsec:2SAT}
In this subsection we consider the alphabet $\Sigma=\{0,1\}$ and the predicate family 
\[\calF=\left\{f^{(0)},f^{(1)},f^{(0,0)},f^{(0,1)},f^{(1,0)},f^{(1,1)}\right\},\]
where the functions $f^{(0)},f^{(1)},f^{(0,0)},f^{(0,1)},f^{(1,0)},f^{(1,1)}:\{0,1\}^{2}\rightarrow\{0,1\}$ are given by

\[\begin{alignedat}{3}
f^{(b)}(\sigma_{1},\sigma_{2})&=1\quad &&\text{if and only if  }\sigma_{1}=b,\quad&\text{for all }b\in \{0,1\},\\
f^{(b_{1},b_{2})}(\sigma_{1},\sigma_{2})&=0\quad &&\text{if and only if }\sigma_{1}=b_{1}\text{ and }\sigma_{2}=b_{2},\quad &\text{for all }b_{1},b_{2}\in \{0,1\}.
\end{alignedat}\]
In this setting, the problem $\mcsp$ is also known as $\mathsf{Max}$-$\mathsf{2SAT}$. If the essentially unary predicates $f^{(0)}$ and $f^{(1)}$ are removed from the family $\calF$, the resulting problem is called $\mathsf{Max}$-$\mathsf{E2SAT}$, where the letter ``$\mathsf{E}$'' stands for ``exact.''

From \Cref{thm:hardness} (or \cite{FMW25}), it follows that $\mathsf{Max}$-$\mathsf{E2SAT}$ is approximation resistant\footnote{That is, $\mathsf{Max}$-$\mathsf{E2SAT}[1,\tfrac{3}{4}+\varepsilon]$ is hard for every $\varepsilon>0$, while every nonempty instance has value at least $\tfrac{3}{4}$.} under multipass streaming, achieving an approximation ratio of $3/4$. However, it is not immediately clear whether the non-exact version $\mathsf{Max}$-$\mathsf{2SAT}$ admits the same ratio. We show in \Cref{thm:2SAT} that $\mathsf{Max}$-$\mathsf{2SAT}$ indeed has the same approximation ratio of $3/4$. 

This contrasts with the single-pass streaming setting, where the approximation ratios for $\mathsf{Max}$-$\mathsf{2SAT}$ and $\mathsf{Max}$-$\mathsf{E2SAT}$ are $\sqrt{2}/2$ and $3/4$, respectively, as shown in \cite{chou2020optimal}. Thus, in the single-pass regime, the non-exact version $\mathsf{Max}$-$\mathsf{2SAT}$ is strictly harder to approximate than $\mathsf{Max}$-$\mathsf{E2SAT}$, whereas in the multipass regime this gap disappears.

We are now ready to prove~\Cref{thm:2SAT}, restated below.
\thmtwoSAT*
\begin{proof}[Proof Sketch]
By considering a random assignment where each variable is independently assigned a value from $\{0,1\}$ uniformly at random, we observe that every nonempty instance $\calI$ of $\mathsf{Max}\text{-}\mathsf{2SAT}$ satisfies $\val_{\calI} \ge 1/2$. Hence, $\lpcurve(c) = c$ for $c \in [0,\tfrac{1}{2}]$.

To bound $\lpcurve(c)$ for larger $c$, consider two extremal constructions:

\begin{itemize}
    \item First instance: A variable set $\{1,2\}$ with two constraints, $((1,2),f^{(0)})$ and $((1,2),f^{(1)})$.  
    Both the integral value and LP value of this instance are $1/2$, implying $\vartheta_{\calF}^*(1/2)\le 1/2$.

    \item Second instance: A variable set $[n]$ with $4\binom{n}{2}$ constraints, where on each pair $\{i,j\} \subseteq [n]$ all four E2SAT constraints $f^{(0,0)}, f^{(0,1)}, f^{(1,0)}, f^{(1,1)}$ are placed.  
    This instance has LP value $1$ and integral value approaching $3/4$ as $n \to \infty$.
\end{itemize}

These two constructions imply $\vartheta_{\calF}^*(1/2)\le 1/2$ and $\vartheta^*_{\calF}(1)\le 3/4$. By \Cref{lem:convexity}, we then obtain
\[
\lpcurve(c) \le \frac{2c+1}{4} \quad \text{for } c \in [\tfrac{1}{2},1].
\]
Thus, it remains to prove the reverse inequality \(
\lpcurve(c) \ge (2c+1)/4
\)
for $c \in [\tfrac{1}{2},1]$. By \Cref{def:lpcurve}, this is equivalent to showing that for every nonempty $\mathsf{Max}\text{-}\mathsf{2SAT}$ instance $\calI$, 
\begin{equation}\label{eq:2SAT_desired}
\val_{\calI}\ge\frac{1}{2}\cdot  \lpval_{\calI} + \frac{1}{4}.
\end{equation}

Let $\calI = (\calV, (C_1, \dots, C_m))$ be a $\mathsf{Max}\text{-}\mathsf{2SAT}$ instance, where $C_i = ((\sfv_{i,1}, \sfv_{i,2}), f_i)$.  
For convenience, define $R^{(1)}:[0,1]\to [0,1]$ as the identity map and $R^{(0)}:[0,1]\to [0,1]$ by $R^{(0)}(x) = 1-x$. For any $(b_1,b_2)\in \{0,1\}^2$ and $i\in [m]$, if $f_i = f^{(b_1,b_2)}$, define
\[
g_i(x_1,x_2) := 1 - R^{(b_1)}(x_1)\,R^{(b_2)}(x_2), 
\qquad
h_i(x_1,x_2) := \min\!\left(1, 2 - R^{(b_1)}(x_1) - R^{(b_2)}(x_2)\right).
\]
Similarly, for $i\in [m]$ and $b\in \{0,1\}$, if $f_i = f^{(b)}$, set 
\[
g_i(x_1,x_2) = h_i(x_1,x_2) = R^{(b)}(x_1).
\]
The basic linear program $\lp_\calI$ can then be written as:
\[
\begin{alignedat}{3}
\text{maximize}\quad     & \frac{1}{m}\sum_{i=1}^{m} z_{i}  \\
\text{subject to}\quad & 0\le z_{i}\le h_{i}\big(x_{\sfv_{i,1}},x_{\sfv_{i,2}}\big) \qquad && \forall i\in [m], \\
& 0 \le x_{\sfv} \le 1 && \forall \sfv\in\calV.
\end{alignedat}
\]
Let $P \subseteq [0,1]^{\calV \cup [m]}$ be the polytope defined by these constraints.  
By an argument analogous to \Cref{lem:solution_structure}, every vertex of $P$ has coordinates in $\{0,\tfrac12, 1\}$.  
Hence, there exists an optimal solution
\[
(x^*, z^*) \in \{0,\tfrac12, 1\}^{\calV \cup [m]}
\]
such that
\begin{equation}\label{eq:2SAT_optimality}
\frac{1}{m}\sum_{i=1}^{m} h_{i}\big(x^*_{\sfv_{i,1}}, x^*_{\sfv_{i,2}}\big) = \lpval_\calI.
\end{equation}
Define a random assignment $\tau: \calV \to \{0,1\}$ by setting $\tau(\sfv) = 1$ with probability $x^*_{\sfv}$, independently for all $\sfv$.  
Then each constraint $C_i$ is satisfied with probability $g_i(x^*_{\sfv_{i,1}}, x^*_{\sfv_{i,2}})$.  
Hence,
\[
\Exu{\tau}{\val_\calI(\tau)}
    = \frac{1}{m}\sum_{i=1}^{m} g_i\big(x^*_{\sfv_{i,1}}, x^*_{\sfv_{i,2}}\big).
\]
Thus, to prove \eqref{eq:2SAT_desired}, it suffices to establish
\begin{equation}\label{eq:2SAT_claim}
\frac{1}{m}\sum_{i=1}^{m} g_i\big(x^*_{\sfv_{i,1}}, x^*_{\sfv_{i,2}}\big)
    \ge \frac{1}{m}\sum_{i=1}^{m}\left(\frac{1}{2}\cdot h_{i}\big(x_{\sfv_{i,1}}^*,x_{\sfv_{i,2}}^*\big)+\frac{1}{4}\right).
\end{equation}

To prove \eqref{eq:2SAT_claim}, we let $T\subseteq [m]$ be the set of constraint indices $i\in [m]$ where $x_{\sfv_{i,1}}^*=x_{\sfv_{i,2}}^*=\frac{1}{2}$. For $i\in T$, a direct calculation shows
\begin{equation}\label{eq:2SAT_T}
g_i\big(x^*_{\sfv_{i,1}}, x^*_{\sfv_{i,2}}\big) = \frac{3}{4}
    = \frac{1}{2} \cdot h_i\big(x^*_{\sfv_{i,1}}, x^*_{\sfv_{i,2}}\big) + \frac{1}{4}.
\end{equation} We claim that
\begin{equation}\label{eq:2SAT_not_T}
\frac{1}{m-|T|}\sum_{i\in [m]\setminus T}g_{i}\big(x_{\sfv_{i,1}}^*,x_{\sfv_{i,2}}^*\big)=\frac{1}{m-|T|}\sum_{i\in [m]\setminus T}h_{i}\big(x_{\sfv_{i,1}}^*,x_{\sfv_{i,2}}^*\big)\geq \frac{1}{2},
\end{equation}
which combined with \eqref{eq:2SAT_T} would yield \eqref{eq:2SAT_claim}. The first transition in \eqref{eq:2SAT_not_T} can be justified by verifying $1-xy=\min \{1,2-x-y\}$ for all $(x,y)\in \{0,\frac{1}{2},1\}^{2}\setminus\{ (\frac{1}{2},\frac{1}{2})\}$. The second transition holds because otherwise, reassigning the value $1/2$ to all $x^*_{\sfv}$, for $\sfv\in [n]$, would increase the sum on the left hand side of \eqref{eq:2SAT_optimality}, contradicting the optimality of the LP solution $(x^*,z^*)$.
\end{proof}

\section{The Multi-Pass Algorithm}\label{sec:algorithm}
In this section we prove the main algorithmic result of the paper, \Cref{thm:streaming_solve_LP}. Our approach combines the high-level strategy of~\cite{Yos11} and~\cite{saxena2025streaming}. Given a bounded-degree instance $\calI \in \cspF$ of the CSP maximization problem, \cite{Yos11} showed that the linear programming relaxation $\lp_{\calI}$ (defined in \Cref{def:BasicLP}) can be approximately solved by a constant-round local algorithm, and that such local algorithms can in turn be simulated by constant-query property testers. A key observation, exploited by~\cite{saxena2025streaming}, is that multi-pass streaming algorithms can likewise simulate constant-round local algorithms. This essentially yields \Cref{thm:streaming_solve_LP}, with one caveat: unlike the bounded-degree instances considered in~\cite{Yos11} for property testing, our streaming setting must handle instances of unbounded degree.

To bridge this gap, we show that the classic reduction from general CSPs to bounded-degree CSPs, originally developed for polynomial-time algorithms in~\cite{Tre01}, can in fact be adapted to the streaming model. Combined with Yoshida’s local algorithm for solving the LP relaxation on bounded-degree instances, this yields the desired streaming algorithm.

The section is organized as follows. 
\Cref{subsec:bounded_degree} presents Yoshida’s local algorithm for approximately solving the LP on bounded-degree instances.
\Cref{subsec:reduction_bounded_degree} shows how general CSP maximization problems can be reduced to the bounded-degree setting, and finally, \Cref{subsec:algo} implements this reduction in the streaming model.

\subsection{Yoshida's Local Algorithm}\label{subsec:bounded_degree}

In order to define local algorithms, we first define the notion of \emph{degrees} for CSP instances.

\begin{definition}\label{def:variable_degree}
Consider a \(\cspF\) instance \(\mathcal{I} = (\mathcal{V}, (C_1, \dots, C_m))\), where \({\mathcal{F}} \subseteq \{ f : \Sigma^k \rightarrow \{0,1\} \}\). The \emph{degree} of a variable \(\sfv \in \mathcal{V}\), denoted by $\deg_{\calI}(\sfv)$, is defined as the number of pairs \((i, \ell) \in [m] \times [k]\) such that \(\sfv\) appears as the \(\ell\)-th variable in the scope of constraint \(C_i\). The \emph{maximum degree} of \(\mathcal{I}\) is the maximum degree over all variables in \(\mathcal{V}\) and is denoted by $\deg_{\calI}$. 
\end{definition}

A large CSP instance with bounded maximum degree naturally corresponds to a sparse hypergraph, where variables are represented as vertices and constraints as hyperedges. This perspective motivates the study of distributed algorithms over the hypergraph structure, in which each variable acts as an agent that communicates with its ``neighbors'' to determine its assignment.

One drawback of the hypergraph view, however, is that the notion of a vertex's neighborhood is less canonical than in standard graphs. To address this, we alternatively model the CSP instance as a bipartite graph capturing the incidence relation between variables and constraints. In this setting, distances between vertices are defined using standard graph distance, and distributed algorithms can be analyzed within the standard ``LOCAL model'' of distributed computing~\cite{peleg2000distributed}.

\begin{definition}
Given a $\cspF$ instance \(\mathcal{I} = (\mathcal{V}, \mathcal{C})\) where \(\mathcal{C} = (C_1, \dots, C_m)\), we define the associated \emph{auxiliary (labeled) bipartite graph} as follows:
\begin{enumerate}[label=(\arabic*)]
\item The vertex set of the graph is partitioned into two parts: the set of variables \(\mathcal{V}\) on the left side, and the index set \([m]\) on the right, representing the constraints.
\item For each variable \(\sfv \in \mathcal{V}\) and each constraint index \(i \in [m]\), we add an edge between \(\sfv\) and \(i\) labeled with \(j\) if \(\sfv\) appears as the \(j\)-th variable in the scope of the constraint \(C_i\).
\item Each vertex \(i \in [m]\) on the right side is additionally labeled with the predicate \(f_i\) associated with the constraint \(C_i\).
\end{enumerate}
\end{definition}

Note that the degree of each variable $\sfv \in \calV$, as defined in \Cref{def:variable_degree}, coincides with the usual notion of the degree when $\sfv$ is viewed as a vertex on the left side of the associated auxiliary bipartite graph.

As is standard in the LOCAL model, a local algorithm is formalized as a function from neighborhood profiles to output values at a designated root vertex. The round complexity of such an algorithm corresponds to the radius of the neighborhood it inspects.

We now define the notion of a neighborhood in the auxiliary bipartite graph representation of a CSP instance. For our purposes --- particularly in describing the local algorithm from~\cite{Yos11} --- we are primarily interested in neighborhoods rooted at vertices representing constraints.

\begin{definition}
Fix an instance \(\mathcal{I} = (\mathcal{V}, \mathcal{C})\in \cspF\), where \(\mathcal{C} = (C_1, \dots, C_m)\), and let \(G\) denote its associated auxiliary bipartite graph. For any vertex $u$ in $G$ and any positive integer \(r\), we define \(\mathcal{N}_{\calI}(u, r)\) to be the radius-\(r\) labeled neighborhood of the vertex $u$. This neighborhood includes all vertices within graph distance at most \(r\) from \(u\), all edges among them, and the associated labels on both vertices and edges.

We let \(\mathscr{N}_{\mathcal{F}}(r)\) denote the collection of all possible radius-\(r\) labeled neighborhoods with a distinguished root vertex \emph{on the right side} (i.e. corresponding to constraints, not variables) that can arise from \(\cspF\) instances. Formally speaking, \(\mathscr{N}_{\mathcal{F}}(r)\)  is defined as 
\begin{align*}
    \mathscr{N}_{\mathcal{F}}(r):=\{ \calN_\calI(i,r): \calI=(\calV,\calC=(C_1,\dots,C_m))\in \mathrm{CSP}(\calF), i\in [m]\}.
\end{align*}
\end{definition}

For instance, $\scrN_{\calF}(1)$ consists of a single element: a star graph with $k$ edges labeled $1$ through $k$. For any fixed $r \ge 2$, however, the collection $\scrN_{\calF}(r)$ is infinite, since a variable in a $\cspF$ instance --- equivalently, a vertex on the left side of the auxiliary bipartite graph --- can have arbitrarily large degree. If we instead restrict our attention to neighborhoods arising from $\cspF$ instances with maximum degree at most $B$ (which we will do in \Cref{lem:local_approx_lp}), then we are concerned with only a finite subset of $\scrN_{\calF}(r)$.

\begin{remark}
    One simple but useful observation about neighborhoods is that, when $\calI=(\calV,\calC=(C_1,\dots, C_m))$ has maximum degree at most $B$, then for any $\sfv\in \calV$ (or $i\in [m]$), the number of vertices and edges contained in $\calN_\calI(\sfv,r)$ (or $\calN_\calI(i,r)$) is at most $2\left(\max\{B,k\}\right)^r+1$. 
\end{remark}

We now state the result of Yoshida \cite{Yos11}, which says that for constant arity $k$, alphabet-size $|\Sigma|$, and degree bound $B$, for any instance $\calI\in \mathrm{CSP}(\calF)$ with degree at most $B$, the objective value  of $\lp_\calI$ can be approximated within an additive error $\varepsilon$ by an local algorithm with constant locality.

\begin{lemma}[{\cite[Theorem 3.1]{Yos11}}]\label{lem:local_approx_lp}
Let \(\mathcal{F} \subseteq \{ f : \Sigma^k \rightarrow \{0,1\} \}\) be a fixed family of predicates, and fix a positive integer \(B\) and an error parameter \(\varepsilon \in (0,1)\). Then there exists a positive integer \(r \leq \exp\big(\mathrm{poly}(kB|\Sigma|/\varepsilon)\big)\) and a deterministic map \(\mathcal{A}_{\loc} : \mathscr{N}_{\mathcal{F}}(r) \rightarrow [0,1]^{\Sigma^k}\) such that the following holds:

Given any \(\cspF\) instance \(\mathcal{I} = (\mathcal{V}, (C_1, \dots, C_m))\) with maximum degree at most $B$, the output vectors \(\hat{z}^{(i)} = \mathcal{A}_{\loc}(\mathcal{N}_{\mathcal{I}}(i, r))\) (for each \(i \in [m]\)) satisfies
\[
\lpval_{\mathcal{I}} - \varepsilon\leq  \frac{1}{m}\sum_{i=1}^{m} \sum_{b \in \Sigma^k} f_i(b) \, \hat{z}^{(i)}_b \leq \lpval_{\mathcal{I}} + \varepsilon,
\]
where \(f_i \in \mathcal{F}\) is the predicate associated with constraint \(C_i\).
\end{lemma}
In words, for constant $k,|\Sigma|$, and $B$, the above lemma offers a local algorithm with constant locality that can approximate $\lpval_{\calI}$ with an additive error $\varepsilon$ given any instance $\calI\in \mathrm{CSP}(\calF)$ with maximum degree at most $B$. 

We briefly sketch the main idea behind Lemma~\ref{lem:local_approx_lp}, referring the reader to \cite{Yos11} for full details. The algorithm consists of two main steps. First, the original LP relaxation is reduced to a \emph{fractional packing problem} (that is, maximizing $c \cdot x$ subject to $Ax \preceq b$ and $x \succeq 0$, where $A$, $b$, and $c$ have only nonnegative entries), also known as a \emph{positive linear program}. Similar reductions appear in earlier works such as \cite{Tre96,fotakis1997linear}. Second, Yoshida applies the distributed algorithm of \cite{KMW06}, which approximates the value of a bounded-degree positive linear program. This algorithm simultaneously maintains primal and dual solutions, and guarantees that their values converge to each other as the number of rounds (and hence the locality) increases.

\subsection{Reduction to Bounded-Degree Instances}\label{subsec:reduction_bounded_degree}

In the polynomial-time setting, the approximability of CSPs is known to reduce to instances with bounded degree, due to the reduction of~\cite{Tre01}. In this subsection, we present a slightly different reduction that is more amenable to implementation in the streaming model. Our (randomized) reduction map is described in \Cref{alg:I_BD_from_I}. 

\begin{algorithm}
    \DontPrintSemicolon
    \SetKwInOut{Input}{Input}\SetKwInOut{Output}{Output}
    \caption{Definition of the Random Bounded-Degree Instance $\calI_{B,D}$}\label{alg:I_BD_from_I}
    \Input{a $\cspF$ instance $\calI=(\calV,(C_{1},\dots,C_{m}))$ and integer parameters $B,D\geq 1$}
    \Output{a $\cspF$ instance $\calI_{B,D}$}
    Let $\calV_{D}:=\left\{(\sfv,j)\mid \sfv\in \calV, \,j\in [D\cdot \deg_{\calI}(\sfv)]\right\}$\;
    \For{$\ell \in [B]$\label{line:ell_loop}}{
        Initialize the pool of available variables $U\gets \calV_{D}$\;
        \For{$i\in [m]$}{\tcp*{next: define the constraint $C_{i,\ell}$}
        \For{$t\in [k]$}{
        Suppose $\sfv$ is the $t$-th variable in the scope of $C_{i}$\;
        Pick a uniformly random variable $v$ from $U\cap \{(\sfv,j)\mid j\in [D\cdot \deg_{\calI}(\sfv)]\}$\label{line:randomness_in_I_BD}\;
        Let the $t$-th variable of $C_{i,\ell}$ be $v$\;
        $U\gets U\setminus \{v\}$ \label{line:removal_U}\tcp*{ensures bounded maximum degree in $\calI_{B,D}$}
        }
        }
    }
    \Return{$\calI_{B,D}=\big(\calV_{D},(C_{i,\ell})_{i\in [m],\, \ell\in [B]}\big)$}
\end{algorithm}
In words, for each $\ell$ 
in the outer loop, we take 
a copy of $\calI$ where each variable $\sfv$ is replaced with one of its copies $(\sfv,j)$.
In the following lemma, we show that this reduction preserves the value of the original instance with high probability.
\begin{lemma}\label{lem:reduction_bounded_degree}
    Let \(\mathcal{F} \subseteq \{ f : \Sigma^k \rightarrow \{0,1\} \}\) be a fixed family of predicates, and fix an error parameter \(\varepsilon \in (0,1)\). There exist positive integers \(B,D\leq \poly(k|\Sigma|/\varepsilon)\) such that the random instance $\calI_{B,D}$ sampled by \Cref{alg:I_BD_from_I} satisfies
    \begin{enumerate}[label=(\arabic*)]
        \item $\Prs{\calI_{B,D}}{\val_{\calI_{B,D}}\geq \val_{\calI}}=1$;
        \item $\Prs{\calI_{B,D}}{\val_{\calI_{B,D}}\geq \val_{\calI}+\varepsilon}\leq 0.01$;
        \item $\Prs{\calI_{B,D}}{\deg_{\calI_{B,D}}\leq B}=1$. 
    \end{enumerate}
\end{lemma}
\begin{proof}
    
Due to the removal step in \Cref{line:removal_U} of \Cref{alg:I_BD_from_I}, every variable $v\in \calV_{D}$ is picked at most once in each iteration of the outer-most for-loop (on \Cref{line:ell_loop}). Since there are $B$ iterations of the outer-most for-loop, any variable $v\in \calV_{D}$ is used at most $B$ times in total, and hence the maximum degree of $\calI_{B,D}$ is always at most $B$. It remains to prove the first and the second items in the statement.

We first show that $\val_{\calI_{B,D}}\geq \val_{\calI}$ always holds. For an assignment $\tau:\calV\rightarrow \Sigma$, we can lift it to an assignment $\widetilde{\tau}: \calV_{D}\rightarrow \Sigma$ by setting $\widetilde{\tau}((\sfv,j)):=\tau(\sfv)$ for each $\sfv\in \calV$ and $j\in [D\cdot \deg_{\calI}(\sfv)]$. By the construction of $\calI_{B,D}$, it is easy to see that $\val_{\calI_{B,D}}(\widetilde{\tau})=\val_{\calI}(\tau)$ always holds.

Next, we prove that $\Prs{\calI_{B,D}}{\val_{\calI_{B,D}}\geq \val_{\calI}+ \varepsilon} \leq 0.01$. Consider any fixed assignment $\widetilde{\tau}:\calV_{D}\rightarrow \Sigma$. For each index pair $(i,\ell)\in [m]\times [B]$, define a Bernoulli random variable
\[X_{(i,\ell)}=\begin{cases}
1, &\text{if the constraint }C_{i,\ell}\text{ is satisfied by $\widetilde{\tau}$ in }\calI_{B,D},\\
0, &\text{otherwise}.
\end{cases}\]
Note that the randomness in $X_{(i,\ell)}$ comes from $\calI_{B,D}$ but not from $\widetilde{\tau}$, which is \emph{fixed}. On the other hand, if we define a \emph{random} assignment $\tau:\calV\rightarrow \Sigma$ by assigning value $\sigma\in \Sigma$ to $\sfv$ with probability
\[
\frac{\Big|\big\{j\in [D\cdot \deg_{\calI}(\sfv)]\,\big|\, \widetilde{\tau}((\sfv,j))=\sigma\big\}\Big|}{D\cdot \deg_{\calI}(\sfv)},
\]
independently for each $\sfv\in \calV$, it is easy to see that
\[
\Exu{\calI_{B,D}}{X_{(i,1)}}=\Exu{\calI_{B,D}}{X_{(i,2)}}=\dots=\Exu{\calI_{B,D}}{X_{(i,\ell)}}=\Pru{\tau}{C_{i}\text{ is satisfied by }\tau}
\]
holds for each $i\in [m]$. Therefore, we have
\begin{equation}\label{eq:sampling_assignments}
\frac{1}{m|B|}\sum_{i=1}^{m}\sum_{\ell=1}^{B}\Exu{\calI_{B,D}}{X_{(i,\ell)}}=\frac{1}{m}\sum_{i=1}^{m}\Pru{\tau}{C_{i}\text{ is satisfied by }\tau}=\Exu{\tau}{\val_{\calI}(\tau)}\leq \val_{\calI}.
\end{equation}

In order to apply~\Cref{prop:concentration}, we must specify a total order on the index set $[m]\times [B]$. We consider the order in which the constraints $C_{i,\ell}$ are specified during the execution of \Cref{alg:I_BD_from_I}: let $(i,\ell)<(i',\ell')$ if either (1) $\ell<\ell'$ or (2) $\ell=\ell'$ and $i<i'$. Each time \Cref{line:randomness_in_I_BD} of \Cref{alg:I_BD_from_I} is executed, the number of available copies $(\sfv,j)$ of $\sfv$ is
\begin{align*}
\big|U\cap \{(\sfv,j)\mid j\in [D\cdot \deg_{\calI}(\sfv)]\}\big|&\geq D\cdot \deg_{\calI}(\sfv) - \deg_{\calI}(\sfv) \\
&\geq (1-D^{-1})\cdot \big|\{(\sfv,j)\mid j\in [D\cdot \deg_{\calI}(\sfv)]\}\big|.
\end{align*}
This means that even if the algorithm sampled the variables of $C_{i,\ell}$ without avoiding those that have already been occupied (i.e.~those that are not in $U$), the probability that none of the $k$ sampled variables would actually be occupied is still at least $(1-D^{-1})^{k}$. Therefore, we have
\begin{align*}
     \Exu{\calI_{B,D}}{X_{(i,\ell)}}\geq (1-D^{-1})^{k}\cdot\Exu{\calI_{B,D}}{X_{(i,\ell)} \mid (X_{(i',\ell')})_{(i',\ell')<(i,\ell)}  }
\end{align*}
for each $(i,\ell)\in [m]\times [B]$. Picking $D=\lceil 10k/\varepsilon\rceil$, we have
\[
\Exu{\calI_{B,D}}{X_{(i,\ell)} \mid (X_{(i',\ell')})_{(i',\ell')<(i,\ell)}  }\leq \Exu{\calI_{B,D}}{X_{(i,\ell)}}+\frac{\varepsilon}{2}, 
\]
since $(1-D^{-1})^{-k}\leq e^{2k/D} \leq 1+\varepsilon/2$.

Let $X=\sum_{i=1}^{m}\sum_{\ell=1}^{B}X_{(i,\ell)}$. Now we can apply \Cref{prop:concentration} to $(X_{(i,\ell)})_{(i,\ell)\in [m]\times [B]}$ and get 
\[
    \Pru{\calI_{B,D}}{X\geq \frac{mB\varepsilon}{2}+\sum_{i=1}^{m}\sum_{\ell=1}^{B}\left(\Exu{\calI_{B,D}}{X_{(i,\ell)}}+\frac{\varepsilon}{2}\right)}
    \leq \exp\left( - \frac{(mB\varepsilon/2)^2}{4mB}\right)= \exp\left( - \frac{mB\varepsilon^{2}}{16}\right). 
\]
By the definition of the variables $X_{(i,\ell)}$ and applying \eqref{eq:sampling_assignments}, we arrive at
\[
\Pru{\calI_{B,D}}{\val_{\calI_{B,D}}(\widetilde{\tau})\geq \val_{\calI}+\varepsilon}\leq \exp\left( - \frac{mB\varepsilon^{2}}{16}\right).
\]
Taking union bound over all $|\Sigma|^{|\calV_D|} = |\Sigma|^{mkD}$ possibilities of $\widetilde{\tau}$, we have
\begin{align*}
    \Pru{\calI_{B,D}}{\val_{\calI_{B,D}}\geq \val_{\calI}+\varepsilon} \leq |\Sigma|^{mkD}\cdot  \exp\left( - \frac{mB\varepsilon^2}{16}\right) \leq 0.01, 
\end{align*}
for some constant $B\leq \poly(kD|\Sigma|/\varepsilon)\leq \poly(k|\Sigma|/\varepsilon)$. 
\end{proof}
    
\subsection{Efficient Implementation in Multi-Pass Streaming}\label{subsec:algo}
In this subsection, we combine Yoshida's local algorithm (\Cref{lem:local_approx_lp}) with the reduction in \Cref{lem:reduction_bounded_degree} to obtain an efficient streaming algorithm for~\Cref{thm:streaming_solve_LP} (restated below).

\thmalgo*


The idea behind~\Cref{thm:streaming_solve_LP} is, given an instance $\calI$, to approximate the LP-value of a randomly sampled instance bounded degree instance $\calI_{B,D}$ and accept if its value exceeds $c+\varepsilon/2$. Using~\Cref{lem:reduction_bounded_degree} we know that (if we are able to produce a good such approximation) with high probability we will accept if $\val_{\calI}\geq c+\varepsilon$, and reject if $\val_{I}\leq \lpcurve(c)-\varepsilon$.


To approximate the LP value of $\calI_{B,D}$ we use
Yoshida's algorithm (\Cref{lem:local_approx_lp}). We integrate the local algorithm into the streaming setting using a similar approach to that of~\cite{saxena2025streaming}. More precisely, we begin by uniformly sampling a set of constraints from $\calI_{B,D}$ and, using a constant number of passes, recover the constant-radius neighborhoods of these constraints. The local algorithm $\calA_{\loc}$ from Lemma~\ref{lem:local_approx_lp} is then applied to approximate the contribution of each sampled constraint to the LP objective, and their average is used to estimate the LP value $\lpval_{\calI_{B,D}}$. 

The main difference between our setting and the one in~\cite{saxena2025streaming} is that our method requires an \emph{implicit} construction of the bounded degree instance $\calI_{B,D}$ using only logarithmic space. This is necessary for us, as we cannot afford to sample a full instance $\calI_{B,D}$ and store it on the memory. In contrast,~\cite{saxena2025streaming} does not apply generic degree-reduction transformations, and instead works directly with the original, potentially unbounded-degree instance $\calI$. Their approach, however, is tailored to the specific class of local algorithms they consider, whereas our reduction applies in a black-box manner to any local algorithm, including Yoshida’s.

We now give the formal proof of~\Cref{thm:streaming_solve_LP}.

\begin{proof}[Proof of \Cref{thm:streaming_solve_LP}]
    We will present the proof in several parts.

    \paragraph{Part 1: the reduction oracle.} Ideally, we would like to first transform the input stream $\calI=(\calV,(C_{1},\dots,C_{m}))$ into the data stream corresponding to a bounded-degree instance $\calI_{B,D}$, as defined by the random reduction procedure in \Cref{alg:I_BD_from_I}. However, due to constraints on space and the number of passes, we cannot afford to run \Cref{alg:I_BD_from_I} in its entirety. Instead, we simulate the reduction \emph{on the fly} --- computing only local portions of the random instance $\calI_{B,D}$ when needed and storing them in memory. This approach leads to \Cref{alg:oracle_I_BD}, a local, streaming-compatible variant of \Cref{alg:I_BD_from_I}.
    \begin{algorithm}
    \DontPrintSemicolon
    \SetKwBlock{Conditioned}{Conditioned on \(L\) do}{}
    \SetKwInOut{Input}{Input}
    \SetKwInOut{Output}{Output}
    \caption{Oracle for Answering Queries on \(\mathcal{I}_{B,D}\) --- Local Version of \Cref{alg:I_BD_from_I}}\label{alg:oracle_I_BD}
    \Input{a \(\cspF\) instance \(\mathcal{I} = (\mathcal{V}, (C_{1},\dots,C_{m}))\) presented as a data stream, and integers \(B, D \geq 1\)}
    \Output{answers queries on \(\mathcal{I}_{B,D}\), each given as an index tuple \((i,\ell,t) \in [m]\times [B]\times [k]\)}
    Initialize an empty list $L$ \tcp*{the memory kept by the algorithm}
    \While{receiving a query  \((i,\ell,t) \in [m]\times [B]\times [k]\)}{
        Suppose $\sfv$ is the $t$-th variable in the scope of $C_{i}$\label{line:oracle_get_variable}\;
        \Conditioned{\tcp*{using fresh randomness}
        Determine a $j\in [D\cdot \deg_{\calI}(\sfv)]$ such that $(\sfv,j)$ is the $t$-th variable of $C_{i,\ell}$\label{line:oracle_count_degree}\;
        Determine the radius-1 neighborhood \(\mathcal{N} = \mathcal{N}_{\mathcal{I}_{B,D}}\big((\sfv,j), 1\big)\)\label{line:oracle_sample_neighborhood}\;
        }
        Output \(\mathcal{N}\), including the root vertex $(\sfv,j)$, as the answer to the query\;
        Append \(\mathcal{N}\) to \(L\) \tcp*{memory update}
    }
    \end{algorithm}
    
    In \Cref{alg:oracle_I_BD}, answering each query requires exactly three passes over the input stream: the first occurs in \Cref{line:oracle_get_variable}, where the variable $\sfv$ is retrieved; the second in \Cref{line:oracle_count_degree}, where the degree of $\sfv$ in $\calI$ is counted; and the third in \Cref{line:oracle_sample_neighborhood}, where the relevant constraints to include in the neighborhood of the variable $(\sfv,j)$ in $\calI_{B,D}$ are determined.
    
    Importantly, the responses of \Cref{alg:oracle_I_BD} are always consistent with some fixed instantiation of the random instance $\calI_{B,D}$, although this instantiation is never computed explicitly. Moreover, its answers to any (possibly adaptive) sequence of queries are indistinguishable from those of a hypothetical oracle that first samples the entire instance $\calI_{B,D}$ and then responds deterministically to each query.
    
    \begin{algorithm}
    \DontPrintSemicolon
    \SetKwInOut{Input}{Input}\SetKwInOut{Output}{Output}
    \SetKwBlock{Qtimes}{repeat \(Q\) times}{}
        \caption{$\textsc{ApproxLP}(\calI,B,D,Q,r)$}\label{alg:approx_lp}
        \Input{a \(\cspF\) instance \(\mathcal{I} = (\mathcal{V}, (C_{1},\dots,C_{m}))\) presented as a data stream, and integers \(B, D, Q, r \geq 1\)}
        \Output{an approximation of $\lpval_{\calI}$ with high probability}
        Initialize a counter $\widetilde{\val} \leftarrow 0$\;
        Start running \Cref{alg:oracle_I_BD} in parallel\;
        \label{line:Qtimes}\Qtimes{
            Sample $(i,\ell)$ uniformly at random from $[m]\times [B]$ \label{line:sample_constraint}
            \tcp*{using fresh randomness}
            Obtain $\calN_{\calI_{B,D}}\big((i,\ell),r\big)$ using multiple queries to \Cref{alg:oracle_I_BD}\label{line:invoking_oracle}\;
            \tcp*{randomness of \Cref{alg:oracle_I_BD} involved}
            Using the map $\calA_{\loc}$ from \Cref{lem:local_approx_lp} to compute $\hat{z}\gets \calA_{\loc}\Big(\calN_{\calI_{B,D}}\big((i,\ell),r\big)\Big)$\;
            \tcp*{deterministic}
            Suppose $f_{i}\in \calF$ is the predicate used by the constraint $C_{i}$\; 
            \label{line:add_counter} Add the result to the counter: \(\widetilde{\val} \leftarrow \widetilde{\val} + \sum_{b\in \Sigma^{k}}f_{i}(b)\hat{z}_{b}\)\;
        }
        \Return $\widetilde{\val}/Q$\;
    \end{algorithm}

    We can then use the local reduction oracle provided by \Cref{alg:oracle_I_BD} to build \Cref{alg:approx_lp}, which samples a constant number of constraints in the bounded-degree instance $\calI_{B,D}$ (which is in turn implicitly sampled by \Cref{alg:oracle_I_BD}) to approximate the LP value $\lpval_{\calI_{B,D}}$. 

    \paragraph{Part 2: correctness of \Cref{alg:approx_lp}.} We argue that \Cref{alg:approx_lp} correctly approximates the LP value of $\calI_{B,D}$. First, we let $Q = \lceil 10/\varepsilon_0^2\rceil$ and let $r \leq \exp(\poly(kB|\Sigma|/\varepsilon_0))$ be as in Lemma~\ref{lem:local_approx_lp}, where $\varepsilon_0$ is a parameter that only depends on $\varepsilon$ and to be determined later. We then pick $B,D\leq \poly(k|\Sigma|/\varepsilon_{0})$ as in \Cref{lem:reduction_bounded_degree}.
    
    Conditioned on a fixed instantiation of $\calI_{B,D}$, the $Q$ iterations of the loop on \Cref{line:Qtimes} of \Cref{alg:approx_lp} are independent of each other. Furthermore, for a fixed instantiation of $\calI_{B,D}$, if we let $\hat{z}^{(i,\ell)}$ be the vector $\calA_{\loc}\Big(\calN_{\calI_{B,D}}\big((i,\ell),r\big)\Big)$, the expected value of each increment to the counter on \Cref{line:add_counter} of \Cref{alg:approx_lp} is
    \[
    \Exu{(i,\ell)\text{ chosen on \Cref{line:sample_constraint}}}{\sum_{b\in \Sigma^{k}}f_{i}(b)\hat{z}_{b}}=\frac{1}{mB}\sum_{(i,\ell)\in[m]\times [B]}\left(\sum_{b\in \Sigma^{k}}f_{i}(b)\hat{z}^{(i,\ell)}_{b}\right),
    \]
    We denote this expected value by $\val^{\calA_{\loc}}_{\calI_{B,D}}$. By Hoeffding's inequality (\Cref{prop:Hoeffding}), it follows that for a fixed instantiation of $\calI_{B,D}$,
    \begin{equation}\label{eq:alg_correctness_1}
        \Pru{\textsc{ApproxLP}}{\left| \textsc{ApproxLP}(\calI,B,D,Q,r)- \val^{\calA_{\loc}}_{\calI_{B,D}}\right|\geq \varepsilon_0\,\Bigg|\,\calI_{B,D}} \leq 2\exp(-2\varepsilon_0^2Q) \leq 0.01,
    \end{equation}
    Recall that Lemma \ref{lem:local_approx_lp} guarantees 
    \begin{equation}\label{eq:alg_correctness_2}
        \lpval_{\calI_{B,D}}-\varepsilon_{0} \leq \val^{\calA_{\loc}}_{\calI_{B,D}}\leq \lpval_{\calI_{B,D}}+\varepsilon_{0}.
    \end{equation}
    Combining \eqref{eq:alg_correctness_1} and \eqref{eq:alg_correctness_2}, we conclude that
    \begin{equation}\label{eq:alg_correctness_3}
        \Pru{\textsc{ApproxLP}}{\left| \textsc{ApproxLP}(\calI,B,D,Q,r)- \lpval_{\calI_{B,D}}\right|\geq 2\varepsilon_0\,\Bigg|\, \calI_{B,D}} \leq 0.01.
    \end{equation}

    \paragraph{Part 3: efficiency of \Cref{alg:approx_lp}.} We begin by providing an upper bound on the number of passes used by the algorithm $\textsc{ApproxLP}(\cdot, B, D, Q, r)$. Aside from the initial pass to determine the number of constraints $m$, only \Cref{line:invoking_oracle} in \Cref{alg:approx_lp} requires access to the input stream. In other words, after the first pass, all subsequent passes over the input are delegated to \Cref{alg:oracle_I_BD}.

    Each execution of \Cref{line:invoking_oracle} makes at most $(\max\{B, k\})^{r+1}$ queries to \Cref{alg:oracle_I_BD}, and each query incurs 3 passes over the input stream. Since \Cref{line:invoking_oracle} is executed $Q$ times, the total number of passes required by \Cref{alg:approx_lp} is $O(Q \cdot (Bk)^{r+1}) \leq  \exp(\exp(\poly(k |\Sigma| / \varepsilon_{0})))$.

    Regarding space complexity, it is straightforward to observe that the memory usage of \Cref{alg:oracle_I_BD} grows linearly with the number of queries made, with each query contributing an additional $O(\log n)$ bits. Moreover, the memory maintained internally by \Cref{alg:approx_lp} (i.e., not delegated to \Cref{alg:oracle_I_BD}) is at most comparable to that used by the oracle. Therefore, the overall space complexity of \Cref{alg:approx_lp} is $O_{k,|\Sigma|,\varepsilon_0}(\log n)$.

    Finally, we note that \Cref{alg:approx_lp} never runs out of fresh random bits during the execution of \Cref{alg:approx_lp}. The only steps requiring fresh randomness are \Cref{line:oracle_count_degree,line:oracle_sample_neighborhood} of \Cref{alg:oracle_I_BD} and \Cref{line:sample_constraint} of \Cref{alg:approx_lp}. These steps are executed only constant times in total. Furthermore, each such operation consumes only $O(\log n)$ random bits, which can be readily generated within the $O(\log n)$-space budget of our streaming algorithm (see \Cref{def:randomized_streaming}). 
    
    \paragraph{Part 4: wrapping up.} We now summarize the preceding components into a complete algorithm for \Cref{thm:streaming_solve_LP}, presented below as \Cref{alg:main}.

\begin{algorithm}
    \DontPrintSemicolon
    \SetKwInOut{Input}{Input}
    \SetKwInOut{Output}{Output}
    \caption{Final Algorithm \(\mathcal{A}\) for \Cref{thm:streaming_solve_LP}}\label{alg:main}
    \Input{A \(\cspF\) instance \(\mathcal{I}\) presented as a data stream}
    \Output{0 if \(\val_{\mathcal{I}} \leq \lpcurve(c) - \varepsilon\); 1 if \(\val_{\mathcal{I}} \geq c + \varepsilon\); both with probability at least \(2/3\)}
    Let \(c'\) be any rational number in \([c + 2\varepsilon/5,\, c + 3\varepsilon/5]\)\;
    Choose a rational number \(\varepsilon_0 \in (0, \varepsilon/5)\) and set \(Q = \lceil 10/\varepsilon_0^2 \rceil\)\;
    Select parameters \(r, B, D\) according to \Cref{lem:local_approx_lp,lem:reduction_bounded_degree}\;
    Run \(\textsc{ApproxLP}(\mathcal{I}, B, D, Q, r)\) (\Cref{alg:approx_lp})\;
    \eIf{the output is at least \(c'\)}{
        \Return{1}
    }{
        \Return{0}
    }
\end{algorithm}

If \(\val_{\mathcal{I}} \geq c + \varepsilon\), then since \(\lpval_{\mathcal{I}_{B,D}} \geq \val_{\mathcal{I}_{B,D}} \geq \val_{\mathcal{I}}\) always holds (by \Cref{lem:reduction_bounded_degree}), we get from \eqref{eq:alg_correctness_3} that 
\[\Pr{\textsc{ApproxLP}(\mathcal{I}, B, D, Q, r) \geq c + \varepsilon - 2\varepsilon_0}\geq 0.99.\]
Since \(c + \varepsilon - 2\varepsilon_0 \geq c'\), it follows that \(\Prs{\mathcal{A}}{\mathcal{A}(\mathcal{I}) = 1} \geq 0.99\geq 2/3\).

Conversely, if \(\val_{\mathcal{I}} \leq \lpcurve(c) - \varepsilon\), then by \Cref{lem:reduction_bounded_degree}, we have \(\val_{\mathcal{I}_{B,D}} \leq \lpcurve(c) - \varepsilon + \varepsilon_0\) with probability at least \(0.99\). From \Cref{def:lpcurve}, whenever \(\val_{\mathcal{I}_{B,D}} < \lpcurve(c)\), it holds that \(\lpval_{\mathcal{I}_{B,D}} < c\), and in that case, we get from \eqref{eq:alg_correctness_3} that \[\Pr{\textsc{ApproxLP}(\mathcal{I}, B, D, Q, r) < c + 2\varepsilon_0\,\Big|\,\calI_{B,D}}\geq 0.99.\]
Since \(c + 2\varepsilon_0 \leq c'\), we conclude that \(\Prs{\mathcal{A}}{\mathcal{A}(\mathcal{I}) = 0} \geq 0.99^2 \geq 2/3\).

Finally, by the analysis in part 3 of this proof, the algorithm $\calA$ requires $O_{\varepsilon}(1)$ passes and $O_{\varepsilon}(\log n)$ bits of memory. 
\end{proof}

\section{Streaming Lower Bound from Communication Complexity}\label{sec:communication_game}
As shown in~\Cref{thm:streaming_solve_LP}, efficient constant-pass streaming algorithms can approximately match the performance of the basic linear programming relaxation for approximating CSPs. In this section, we establish the complementary hardness result, \Cref{thm:hardness}, which says that multi-pass streaming algorithms cannot outperform the linear programming relaxation by a constant margin.


In \Cref{thm:hardness}, we are given a fixed integrality gap instance $\calI\in\cspF$, with $\val_{\calI}=s$ and $\lpval_{\calI}=c$. To establish the hardness of the gap problem $\McspF{c - \varepsilon}{s + \varepsilon}$, we construct two probability distributions over $\cspF$ instances, referred to as the YES and NO distributions. Instances drawn from the YES distribution typically have value at least $c - \varepsilon$, while those from the NO distribution typically have value at most $s + \varepsilon$. The goal is to show that any $p$-pass streaming algorithm with limited memory cannot reliably distinguish between instances sampled from these two distributions.

The construction of the YES and NO distributions largely follows the ideas of~\cite{Yos11}. Starting from the fixed integrality gap instance $\calI$, we perform a blow-up: each variable is replaced with $n$ copies, and each constraint is replaced with $O(n)$ copies. We then define two different but streaming-indistinguishable methods for selecting which variable copies appear in each constraint copy. These two selection procedures yield the YES and NO distributions, respectively.

The key distinction between our hardness result and that of~\cite{Yos11} lies in the model of computation: we aim to establish lower bounds against multi-pass streaming algorithms, which are potentially more powerful than the query-based property testing algorithms considered in~\cite{Yos11}. To more cleanly capture the broader range of behaviors that a streaming algorithm might exhibit when processing inputs from the YES and NO distributions, we introduce an abstract communication game, called $\DIHP(G,n,\alpha,K)$, that models these possibilities. This communication game is in turn based on an abstract object $G$ that we call a \emph{distribution-labeled $k$-graph}.

This section is structured as follows. In \Cref{subsec:labeled_matchings,subsec:Markov_kernel,subsec:distribution_labeled_graph}, we introduce the abstract mathematical objects underlying the communication game $\DIHP(G,n,\alpha,K)$. The formal definition of the communication game is given in~\Cref{subsec:communication_game}. In \Cref{subsec:streaming_lower_bound}, we explain how the communication complexity of $\DIHP(G, n, \alpha, K)$ yields the desired lower bound against streaming algorithms. Proving a communication lower bound of $\DIHP(G,n,\alpha,K)$ will be the subject of \Cref{sec:communication_lower_bound,sec:global_rectangle,sec:Fourier_decay}.

\subsection{Labeled Matchings}\label{subsec:labeled_matchings}

Labeled matchings are combinatorial objects that have been widely used in establishing streaming lower bounds for approximating CSPs (e.g. \cite{KKS14,KK19,CGSV24,FMW25}). In \cite{FMW25} especially, the \emph{space} of labeled matchings plays a prominent role in proof of the lower bound, and a significant emphasis was placed on the exploration of Fourier analytic properties of this space. While \cite{FMW25} studies the space of labeled matchings on a \emph{complete graph}, which is tailored to the specific CSP of Max-Cut, our goal of analyzing general CSPs requires considering labeled matchings on a \emph{complete $k$-partite hypergraph}. The following preliminary definition captures the set-theoretic structure of complete $k$-partite hypergraphs.

\begin{definition}
For finite sets $U_{1},\dots,U_{k}$ of equal cardinality, we call the tuple $\calU=(U_{1},\dots,U_{k})$ a \emph{$k$-universe}. The \emph{cardinality} of $\calU$, denoted by $|\calU|$, is defined to be the common cardinality of the sets $U_{i}$. For convenience, we use the shorthand $\tcup \calU$ for the union $\bigcup_{i\in [k]}U_{i}$ and $\tprod \calU$ for the Cartesian product $\prod_{i=1}^{k}U_{i}$.
\end{definition}

Before introducing labeled matchings, we first introduce a convenient notation for the collection of unlabeled matchings.

\begin{definition}
For a $k$-universe $\calU$ and a nonnegative integer $m\leq |\calU|$, we let $\calM_{\calU,m}$ denote the collection of all matchings (without labels) in the complete $k$-partite hypergraph $(\tcup \calU,\tprod\calU)$ (the hypergraph with vertex set $\tcup \calU$ and edge set $\tprod\calU$) with $m$ edges. We also write $\calM_{\calU,\leq m}:=\bigcup_{d=0}^{m}\calM_{\calU,d}$.
\end{definition}

We now define the space of labeled matchings as follows.

\begin{definition}
For a $k$-universe $\calU$ and a nonnegative integer $m\leq |\calU|$, we define the following space of labeled matchings:
$$\Omega^{\calU,m}:=\left\{\bfy\in\Map{\tprod \calU}{\ZNk\cup\{\nil\}}:\supp(\bfy)\text{ is a matching with }m\text{ edges}\right\}.$$
Here, $\supp(\bfy)$ denotes the support of $\bfy$, i.e., the edges in $\prod \calU$ mapped to $\bZ_N^k$ (see \Cref{subsec:general_notations}). 
\end{definition}

Note that the labels on edges of the matchings are elements of $\bZ_N^k$. This differs from the usual $\bF_{2}$-labels considered for Max-Cut, and aligns with what has been used for general CSPs \cite{CGSV24}.  

\subsection{The Markov Kernel}\label{subsec:Markov_kernel}

In previous works on the Max-Cut problem \cite{KKS14,KK19,FMW25}, an important concept used in the construction of the YES distribution is random generation of labeled matchings that are \emph{compatible} with a given bipartition of a vertex set. However, as already evidenced by \cite{CGSV24}, it turns out that for studying general CSPs, the black-and-white notion of compatibility needs to be relaxed into a range of probabilities in $[0,1]$. The more general formalism is that of a \emph{Markov transition} from the space of bipartitions (or, for us, $N$-partitions) to the space of labeled matchings. 

The following notation will be helpful in defining the Markov transition, as well as in later parts of the paper.

\begin{notation}\label{notation:vector_subscript}
Suppose $\Lambda$ is a ground set and \(x \in \mathbb{Z}_N^{\Lambda}\) is a $\ZmodN$-vector indexed by $\Lambda$. If \(e = (v_1, \dots, v_k)\) is a tuple of elements with each \(v_i \in \Lambda\) for \(i \in [k]\), we denote by \(x_{|e}\) the vector \((x_{v_1}, \dots, x_{v_k}) \in \mathbb{Z}_N^k\). 
\end{notation}

We now define the matrix specifying the probability that we want to draw each labeled matching given an ``$N$-partition'' of the vertices. Such matrices are known as \emph{Markov kernels}.

\begin{definition}\label{def:Markov_kernel}
Fix a $k$-universe $\calU$, a positive integer $m\leq |\calU|$, and a one-wise independent distribution $\mu$ over $\ZNk$. We define a right stochastic matrix $\bfP^{\calU,m}_{\mu}:\ZmodN^{\bigcup \calU}\times \Omega^{\calU,m}\rightarrow[0,\infty)$ as follows.  For each $x\in \bZ_{N}^{\bigcup \calU}$ and $\bfy\in \Omega^{\calU,m}$, the entry $\bfP^{\calU,m}_{\mu}(x,\bfy)$ is the probability that the output of the following process equals $\bfy$:
\begin{enumerate}
\item sample a matching $M$ uniformly at random from $\calM_{\calU,m}$; 
\item let $\bfz\in \Omega^{\calU,m}$ have support $\supp(\bfz)=M$, and
\item for each edge $e\in M$, draw $w_{e}\in \ZNk$ independently from $\mu$ and set $\bfz(e)=x_{|e}-w_{e}$, where the subtraction is performed in the Abelian group $\ZNk$;
\item output $\bfz$.
\end{enumerate}
\end{definition}

As noted earlier, the seed vector $x\in \ZmodN^{\bigcup\calU}$ can be viewed as an $N$-partition of the vertices in $\tcup\calU$. A uniformly random matching $M$ is drawn, and the label $\bfy(e)$ on each edge $e\in M$ reveals certain information about $x_{|e}$, i.e., how the vertices of $e$ are partitioned. Due to the ``masking vector'' $w_{e}$ drawn from the one-wise independent distribution $\mu$, no information is revealed about $x_{v}$ for any single vertex $v$. However, information about correlations may be revealed --- for instance, if $u$ and $v$ are two vertices of a same edge $e\in M$ and $\mu$ is the uniform distribution on diagonal elements of $\ZNk$, the difference $x_{u}-x_{v}\pmod N$ might be completely recoverable from $\bfy(e)$. We will formalize this intuition using Fourier analysis later in the paper (see \Cref{subsec:SVD}).

\subsection{Distribution-Labeled $k$-Graphs}\label{subsec:distribution_labeled_graph}

As promised in the introductory text of \Cref{sec:communication_game}, the communication game $\DIHP(G,n,\alpha,K)$ (to be defined in \Cref{subsec:communication_game}) is based on an abstract structure $G$ called a distribution-labeled $k$-graph, which we now define as follows.

\begin{definition}\label{def:distribution_k_graph}
A \emph{distribution-labeled $k$-graph} $G$ consists of the following data: a vertex set $\calV$; a multi-set $\calE$ of hyperedges, each an ordered $k$-tuple of distinct vertices in $\calV$; a positive integer $N$; and a collection of probability distributions $(\mu_{\sfe})_{\sfe\in \calE}$, where each $\mu_{\sfe}$ is a one-wise independent probability distribution on the Abelian group $\ZNk$. 
\end{definition}

Having made clear the first parameter $G$ in $\DIHP(G,n,\alpha,K)$, we now turn to the second parameter $n$: this is the blow-up factor of the distribution-labeled $k$-graph $G$. Indeed, the communication game is not played over $G$ itself, but rather over the $n$-fold blow-up of $G$. The set-theoretic structure of the blow-up is captured by the following definition.

\begin{definition}\label{def:associated_combinatorial}
Given a distribution-labeled $k$-graph $G=(\calV,\calE,N,(\mu_{\sfe})_{\sfe\in \calE})$ and a positive integer $n$, we define the following associated combinatorial objects.
\begin{enumerate}
\item The set $\calV\times [n]$, i.e. the $n$-blow-up of the vertex set $\calV$, will be referred to as the ground set.
\item For each $\sfv\in \calV$, let $U_{\sfv}:=\{\sfv\}\times [n]$ be the subset of $\calV\times [n]$ consisting of the $n$ copies of $\sfv$.
\item We associate with each hyperedge $\sfe=(\sfv_{1},\dots,\sfv_{k})\in \calE$ the $k$-universe $\calU_{\sfe}:=(U_{\sfv_{1}},\dots,U_{\sfv_{k}})$.
\end{enumerate}
\end{definition}

\subsection{The Communication Game}\label{subsec:communication_game}

The following notation will be helpful in defining the communication game, as well as in later parts of the paper.

\begin{definition}
Fix a distribution-labeled \(k\)-graph \(G = (\calV, \calE, N, (\mu_{\sfe})_{\sfe \in \calE})\). The Abelian group \(\mathbb{Z}_N^{\calV \times [n]}\) will play a central role throughout \Cref{sec:communication_game,sec:global_rectangle}. For each edge \(\sfe \in \calE\), recall from \Cref{def:associated_combinatorial} that \(\bigcup \calU_{\sfe} \subseteq \calV \times [n]\). We denote by \(\proj_{\sfe}\) the canonical projection from \(\mathbb{Z}_N^{\calV \times [n]}\) onto \(\mathbb{Z}_N^{\bigcup \calU_{\sfe}}\).
\end{definition}

We are now ready to define the communication game $\DIHP(G,n,\alpha,K)$.

\begin{definition}\label{def:communication_game}
Given a distribution-labeled $k$-graph $G=(\calV,\calE,N,(\mu_{\sfe})_{\sfe\in \calE})$, parameters $n,K\in \bN$ and $\alpha\in (0,1)$, we define the communication game $\DIHP(G,n,\alpha,K)$ as follows:

\begin{enumerate}
\item There are $|\calE|\cdot K$ players, each indexed by a pair $(\sfe,j)$, where $\sfe\in \calE$ and $j\in [K]$.

\item Each player $(\sfe,j)$ receives as input a labeled matching in $\Omega^{\calU_{\sfe},\alpha n}$. 

\item \textbf{The no distribution:} define $\calD_{\no}$ to be the uniform distribution on the Cartesian product $\prod_{(\sfe,j)\in \calE\times [K]}\Omega^{\calU_{\sfe},\alpha n}$, i.e. each player gets a independent uniformly random input.

\item \textbf{The yes distribution:} define $\calD_{\yes}$ to be the joint distribution of $(\bfy^{(\sfe,j)})_{(\sfe,j)\in \calE\times [K]}$ obtained by the following procedure:
\begin{itemize}
\item Sample a uniformly random vector $\widetilde{x}\in \ZmodN^{\calV\times [n]}$. 
\item For each player $(\sfe,j)\in \calE\times [K]$, independently draw a labeled matching $\bfy^{(\sfe,j)}\in \Omega^{\calU_{\sfe},\alpha n}$ according to the distribution given by the probability mass function $\bfP^{\calU_{\sfe},\alpha n}_{\mu_{\sfe}}\big(\proj_{\sfe}(\widetilde{x}),\cdot\big)$.
\end{itemize}
\end{enumerate}
The goal of the players is to decide whether their inputs $(\bfy^{(\sfe,j)})_{(\sfe,j)\in \calE\times [K]}$ comes from $\calD_{\yes}$ or $\calD_{\no}$.
\end{definition}

\begin{remark}
Throughout this paper, whenever we refer to the communication game $\DIHP(G,n,\alpha,K)$, we treat $G,\alpha$ and $K$ as fixed parameters, and consider the asymptotic regime $n\rightarrow\infty$. 
\end{remark}

As is standard in distributional communication complexity, we measure the performance of a communication protocol by its ``advantage'', defined as follows.

\begin{definition}\label{def:advantage}
A deterministic communication protocol $\Pi$ for $\DIHP(G,n,\alpha ,K)$ computes a function $\Pi:\prod_{(\sfe,j)\in \calE\times [K]}\Omega^{\calU_{\sfe},\alpha n}\rightarrow\{0,1\}$. We define its \emph{advantage} in the communication game as 
\begin{align*}
    \mathrm{adv}(\Pi) : = \left| \Pru{\bfY\sim \mathcal{D}_\yes}{\Pi(\bfY)  =1 } - \Pru{\bfY\sim \mathcal{D}_\no}{\Pi(\bfY) =1 }\right|,
\end{align*}
where $\bfY$ denotes a joint input $\bfY=(\bfy^{(\sfe,j)})_{(\sfe,j)\in \calE\times [K]}$.
\end{definition}

Since a randomized communication protocol is a distribution over deterministic protocols, it may well be replaced by the deterministic protocol in its support with the highest advantage. It is therefore without loss of generality to only consider deterministic protocols for $\DIHP(G,n,\alpha,K)$.

The communication complexity of $\DIHP$ is then defined as follows. 

\begin{definition}
The \emph{communication cost} of a protocol~$\Pi$, denoted by~$|\Pi|$, is the total number of bits broadcasted by all players across all rounds during its execution.  
The \emph{communication complexity} of the game $\DIHP(G, n, \alpha, K)$, denoted by~$\CC(G, n, \alpha, K)$, is the minimum communication cost over all protocols~$\Pi$ that satisfy $\adv(\Pi) \geq 0.1$.
\end{definition}
We have the following communication lower bound for the $\DIHP(G,n,\alpha,K)$ problem, the proof of which will occupy \Cref{sec:communication_lower_bound,sec:global_rectangle,sec:Fourier_decay}.

\begin{theorem}\label{thm:dihp_lowerbound}
    Fix a distribution-labeled $k$-graph $G$, an integer $K>0$ and a parameter $\alpha\in \big(0,10^{-8}k^{-3}\big]$. There exists a constant $\gamma=\gamma(G,\alpha,K)>0$ such that $\CC(G,n,\alpha,K)\geq \gamma n^{1/3}$.
\end{theorem}

\subsection{Streaming Lower Bound}\label{subsec:streaming_lower_bound}

It is now time to demonstrate how the communication complexity of $\DIHP(G,n,\alpha,K)$ is related to lower bounds for streaming approximation of CSPs and give a proof of \Cref{thm:hardness}. We present the general reduction lemma from $\DIHP(G,n,\alpha,K)$ to $\McspF{c-\varepsilon}{c+\varepsilon}$:

\begin{lemma}\label{lem:communication_reduction}
Fix a nonempty $\cspF$ instance $\calI = (\calV, (C_1, \dots, C_m))$, and let $s := \val_{\calI}$ and $c := \lpval_{\calI}$.  
Then there exists a distribution-labeled $k$-graph~$G=(\calV,\calE,N,(\mu_{\sfe})_{\sfe\in \calE})$ such that for any fixed error parameter $\varepsilon \in (0,1)$ and constants
\begin{equation}\label{eq:reduction_alpha_K}
\alpha \leq (100k)^{-1}\varepsilon
\quad \text{and} \quad 
K \geq 100 \alpha^{-1} \varepsilon^{-2} N^{2k} \cdot |\calV| \log |\Sigma|,
\end{equation}
the following holds for sufficiently large $n$:
\begin{enumerate}[label=(\arabic*)]
    \item If $c < 1$, then any $p$-pass algorithm for $\McspF{c - \varepsilon}{s + \varepsilon}$ requires at least $(pmK)^{-1} \cdot \CC(G, n, \alpha, K)$ bits of memory on input instances with $|\calV| \cdot n$ variables and at most $mK\cdot n$ constraints.
    
    \item If $c = 1$, then any $p$-pass algorithm for $\McspF{1}{s + \varepsilon}$ requires at least $(pmK)^{-1} \cdot \CC(G, n, \alpha, K)$ bits of memory on input instances with $|\calV| \cdot n$ variables and at most $mK\cdot n$ constraints.
\end{enumerate}
\end{lemma}

We observe that \Cref{thm:hardness} (restated below) follows immediately from this reduction lemma:

\thmhardness*

\begin{proof}[Proof of \Cref{thm:hardness} assuming \Cref{lem:communication_reduction}]
Suppose $\calI=(\calV,(C_{1},\dots,C_{m}))$. We take constants $\alpha \leq  \min\{10^{-8}k^{-3}, (100k)^{-1}\varepsilon\}$, and $K\geq 100\alpha^{-1}$. The conclusions then follow by first applying \Cref{lem:communication_reduction} and then applying \Cref{thm:dihp_lowerbound} (note that $m,|\calV|,K$ are all constants).
\end{proof}

Finally, we arrive at the main task of this section, which is to prove~\Cref{lem:communication_reduction}. The following observation will be useful for the proof.

\begin{proposition}\label{prop:rational_solution}
    For every $\mathrm{CSP}(\calF)$ instance $\calI$, there exists an optimal solution to $\lp_{\calI}$ (achieving the optimal value $\lpval_{\calI}$) where all variables take rational values.  
\end{proposition}
\begin{proof}
    Since all coefficients of the $\lp_\calI$ are rational and the feasible region is nonempty, by a folklore\footnote{See Section 3.7 of \cite{laurent2012semidefinite} for a reference. } result of linear programming, there exists at least one rational-valued optimal solution to $\lp_\calI$. 
\end{proof}

The proof of \Cref{lem:communication_reduction} is rather lengthy and consists of 5 steps.
\begin{proof}[Proof of \Cref{lem:communication_reduction}] 
    The organization of the proof is as follows. We first construct a distribution-labeled $k$-graph $G$ from the gap instance $\calI$. Then we present a general scheme of mapping a joint input $\bfY$ of $\DIHP(G,n,\alpha,K)$ to a $\mathrm{CSP}(\calF)$ instance --- this also translates $\calD_{\yes}$ and $\calD_{\no}$ to distributions over $\mathrm{CSP}(\calF)$. Finally, we prove soundness and completeness of the reduction.
    \paragraph{Step 1: construction of $G$.}
    We construct the distribution-labeled $k$-graph $G$ by specifying its four components as follows:
    
    \begin{enumerate}[label=(\arabic*)]
        \item The vertex set of $G$ is simply the variable set $\calV$.
        
        \item For each constraint $C_i = ((\sfv_{i,1}, \dots, \sfv_{i,k}), f_i)$, define the hyperedge $\sfe_i = (\sfv_{i,1}, \dots, \sfv_{i,k})$. Let the multi-set $\calE = \{\sfe_1, \dots, \sfe_m\}$ be the edge set of $G$.

\item By~\Cref{prop:rational_solution}, there exist a rational solution $\left((x^*_{\sfv,\sigma})_{\sfv \in \calV,\, \sigma \in \Sigma}, (z^*_{i,b})_{i \in [m],\, b \in \Sigma^k}\right)$ that achieves the optimal LP value $\lpval_{\calI} = c$.  
Let $N$ be the least common denominator of all $x^{*}_{\sfv,\sigma}$ and $z^*_{i,b}$. We then scale the values by a factor of $N$:
\[
    x'_{\sfv,\sigma} := N x^*_{\sfv,\sigma}, \quad z'_{i,b} := N z^*_{i,b}.
\]
These are integers by construction. Also define:
\[
    p_i^* := \sum_{b \in \Sigma^k} f_i(b) z^*_{i,b},
\]
which represents the contribution of constraint $C_i$ to the objective under the solution $(x^*, z^*)$.

\item Fix a total order $\prec$ on $\Sigma$. For each vertex $\sfv \in \calV$, define a map $q_{\sfv} : \bZ_N \to \Sigma$ such that for each $i \in \{0,1,\dots,N-1\}$,
\[
    q_{\sfv}(i) = \sigma \quad \text{if} \quad 
    \sum_{\sigma' \prec \sigma} x'_{\sfv,\sigma'} \le i < \sum_{\sigma' \preceq \sigma} x'_{\sfv,\sigma'}.
\]
This is well-defined since the total sum of $x'_{\sfv,\sigma}$ over $\sigma \in \Sigma$ is $N$. For each $i \in [m]$ now we obtain distribution $\mu_{\sfe_i}$ over $\bZ_N^k$ from the following process:
\begin{itemize}
    \item Sample $b = (b_1, \dots, b_k) \in \Sigma^k$ with probability $z^*_{i,b}$;
    \item Then uniformly sample $w \in \bZ_N^k$ from the Cartesian product
    \(q_{\sfv_1}^{-1}(b_1) \times \cdots \times q_{\sfv_k}^{-1}(b_k).
    \)
\end{itemize}
\end{enumerate}
    Recall from \Cref{def:BasicLP} that the second set of constraints in \textsc{BasicLP} ensures \[\sum_{b\in \Sigma^{k}}\ind{b_{j}=\sigma}\cdot z^*_{i,b}=x^*_{\sfv,\sigma}\] if $\sfv$ is the $j$-th variable of $C_{i}$. Therefore, a sample $w$ from $\mu_{\sfe_{i}}$ has probability exactly $x^*_{\sfv,\sigma}$ of falling in the set $\{w\in\ZNk:w_{j}\in q_{\sfv}^{-1}(\sigma)\}$. Since each pre-image $q^{-1}_{\sfv}(\sigma)$ has cardinality exactly $x'_{\sfv,\sigma}=Nx^*_{\sfv,\sigma}$, it follows that $\mu_{\sfe_{i}}$ is one-wise independent. 
    
    The distribution-labeled $k$-graph $G$ defined above, together with fixed constants satisfying \eqref{eq:reduction_alpha_K}, specifies a communication game $\DIHP(G,n,\alpha ,K)$. The following three steps together give a reduction from $\DIHP(G,n,\alpha ,K)$ to $\mathsf{MaxCSP}(\calF)[c-\varepsilon,s+\varepsilon]$ (or $\mathsf{MaxCSP}(\calF)[c,s+\varepsilon]$ when $c=1$). 
    
    \paragraph{Step 2: the reduction map.} Recall that in the communication game $\DIHP(G,n,\alpha ,K)$, each player $(\sfe,j) \in \calE \times [K]$ has a labeled matching $\bfy^{(\sfe,j)}\in \Omega^{\calU_{\sfe},\alpha n}$ in hand. As promised in the beginning of the proof, in this step we show how to map a joint input $\bfY=(\bfy^{(\sfe,j)})_{(\sfe,j)\in\calE\times[K]}$ in the communication game to a $\cspF$ instance $\calI_{\bfY}$. The construction of $\calI_{\bfY}$ is as follows:

    \begin{enumerate}[label=(\arabic*)]
    \item The variable set of $\calI_{\bfY}$ is $\calV\times [n]$, which is also the ground set in the game $\DIHP(G,n,\alpha, K)$.

    \item Recall that the edge set of $G$ is $\calE=\{\sfe_{1},\dots,\sfe_{m}\}$, where each $\sfe_{i}$ corresponds to a constraint $(\sfe_{i},f_{i})$ in the starting instance $\calI$. For each $i\in [m]$ and $j\in [K]$, the player $(\sfe_{i},j)$ gets a labeled matching $\bfy^{(\sfe_{i},j)}\in \Omega^{\calU_{\sfe_{i}},\alpha n}$. Let $M_{i,j}$ be the sub-matching of $\supp(\bfy^{(\sfe_{i},j)})$ consisting of all edges $e\in \supp(\bfy^{(\sfe_{i},j)})$ such that $\bfy^{(\sfe_{i},j)}(e)=\mathbf{0}$, where $\mathbf{0}$ is the identity element of the Abelian group $\ZNk$. We let \(\calC^{(i,j)}\) be the collection of constraints $(e,f_{i})$ where $e$ ranges in the matching $M_{i,j}$. 

    \item Finally, the constraint sequence of $\calI_{\bfY}$, denoted by $\calC_{\bfY}$, is defined to be the concatenation of all constraint sequences $\calC^{(i,j)}$ for $i\in [m]$ and $j\in [K]$. 

    \item Note that as $\calI_{\bfY}$ is meant to be fed to a hypothetical streaming algorithm, we also need to specify the order in which the constraints in $\calC_{\bfY}$ appear in the stream. This is straightforward: we fix an arbitrary total order on the index set $[m]\times [K]$, and concatenate the constraint sequences $\calC^{(i,j)}$ with respect to that order. Within each segment $\calC^{(i,j)}$, the individual constraints can be ordered arbitrarily.
    \end{enumerate}
    
    This completes the definition of the reduction map. 
    
    \paragraph{Step 3: reduction justification.} It is easy to see that a multi-pass streaming algorithm taking input $\calI_{\bfY}$ can be translated back to a communication protocol for $\DIHP(G,n,\alpha,K)$: in any pass whenever the streaming algorithm finishes processing a segment $\calC^{(i,j)}$, the player $(\sfe_{i},j)$ in the communication game correspondingly broadcasts the memory state. In this way, any $p$-pass streaming algorithm $\calA$ that achieves
    \begin{equation}\label{eq:streaming_achieves}
    \left|\Pru{\bfY \sim \calD_{\yes}}{\calA(\calI_{\bfY})=1}-\Pru{\bfY \sim \calD_{\no}}{\calA(\calI_{\bfY})=1}\right| \geq 0.1
    \end{equation}
    using $S$ bits of memory implies a communication protocol $\Pi$ for $\DIHP(G,n,\alpha,K)$ with $\mathrm{adv}(\Pi)\geq 0.1$ using $p\cdot mK\cdot S$ total bits of communication. We thus conclude that any $p$-pass streaming algorithm that achieves \eqref{eq:streaming_achieves} must use at least $(pmK)^{-1}\cdot \CC(G,n,\alpha,K)$ bits of memory.

    The next step is to show the completeness and soundness of the reduction: it remains to prove
    \begin{align}
    \Pru{\bfY\sim\calD_{\yes}}{\val_{\calI_{\bfY}}\geq c-\varepsilon}&\geq 1-o_{n}(1), \label{eq:reduction_completeness}\\
    \Pru{\bfY\sim\calD_{\no}}{\val_{\calI_{\bfY}}\leq s+\varepsilon}&\geq 1-o_{n}(1), \label{eq:reduction_soundness}
    \end{align}
    and if $c=1$ we will show 
    \begin{equation}\label{eq:reduction_perfect_completeness}
    \Pru{\bfY\sim\calD_{\yes}}{\val_{\calI_{\bfY}}=1}=1.
    \end{equation}
    
    For the $c<1$ case, the combination of \eqref{eq:reduction_completeness} and \eqref{eq:reduction_soundness} imply that any $p$-pass streaming algorithm for $\McspF{c-\varepsilon}{s+\varepsilon}$ (with error probability at most $1/3$, as per \Cref{def:McspF}) must satisfy \eqref{eq:streaming_achieves}, and thus have memory size at least $(pmK)^{-1}\cdot \CC(G,n,\alpha,K)$ on input instances with $|\calV| \cdot n$ variables (note that $\calI_{\bfY}$ always have $|\calV|\cdot n$ variables).
    
    For the $c=1$ case, the combination of \eqref{eq:reduction_perfect_completeness} and \eqref{eq:reduction_soundness} imply that any $p$-pass streaming algorithm for $\McspF{1}{s+\varepsilon}$ must satisfy \eqref{eq:streaming_achieves}, and thus have memory size at least $(pmK)^{-1}\cdot \CC(G,n,\alpha,K)$  on input instances with $|\calV| \cdot n$ variables. 
    
    In Step 4 below we prove \eqref{eq:reduction_completeness} and \eqref{eq:reduction_perfect_completeness}, while \eqref{eq:reduction_soundness} is proven in Step 5.

    \paragraph{Step 4: completeness.} Recall from \Cref{def:communication_game} that in the process of drawing a sample $\mathbf{Y} \sim \mathcal{D}_{\yes}$, the first step is to sample a random vector $\widetilde{x} \in \ZmodN^{\mathcal{V} \times [n]}$. For each such vector $\widetilde{x}$, we define an assignment
\[
\widetilde{\tau_{x}} : \mathcal{V} \times [n] \rightarrow \Sigma
\quad \text{by letting} \quad
\widetilde{\tau_{x}}((\sfv, \ell)) = q_{\sfv}(\widetilde{x}_{(\sfv,\ell)})
\]
for each variable $(\sfv, \ell) \in \mathcal{V} \times [n]$.

Fix a player $(\sfe_i=(\sfv_{i,1},\dots,\sfv_{i,k}), j) \in \mathcal{E} \times [K]$ in the communication game $\DIHP(G, n, \alpha, K)$ as defined in Step 1. Recall from \Cref{def:communication_game} that the YES-case input $\mathbf{y}^{(\sfe_i, j)}$ given to this player is determined by: 
\begin{enumerate}[label=(\arabic*)]
\item the sampled random vector $\widetilde{x}$,
\item a random matching $\supp(\mathbf{y}^{(\sfe_i, j)})$, and
\item labels on the edges in $\supp(\mathbf{y}^{(\sfe_i, j)})$ determined by drawing $w_e \sim \mu_{\sfe_i}$ independently for each $e \in \supp(\mathbf{y}^{(\sfe_i, j)})$.
\end{enumerate}
According to the reduction map in Step 2, a constraint $(e, f_i)$ is placed if and only if
\[
\mathbf{y}^{(\sfe_i, j)}(e) \defeq \widetilde{x}_{|e} - w_e = \mathbf{0}.
\]
Furthermore, such a constraint is satisfied by the assignment $\widetilde{\tau_{x}}$ if and only if
\[
\widetilde{x}_{|e} \in \bigcup_{b \in \Sigma^k,\, f_i(b) = 1}
q_{\sfv_{i,1}}^{-1}(b_1) \times \cdots \times q_{\sfv_{i,k}}^{-1}(b_k).
\]
Therefore, conditioned on $\supp(\bfy^{(\sfe_{i},j)})$, for each $e\in \supp(\bfy^{(\sfe_{i},j)})$ we have
\begin{equation}\label{eq:prob_constraint_placed}
\Pru{\widetilde{x}\in \ZmodN^{
\calV\times [n]
},\,w_{e}\sim\mu_{\sfe_{i}}}{(e,f_{i})\text{ is placed in }\calI_{\bfY}\,\middle|\, \supp(\bfy^{(\sfe_{i},j)})}=\Pru{\widetilde{x}_{|e}\in\ZNk,\,w_{e}\sim\mu_{\sfe_{i}}}{\widetilde{x}_{e}=w_{e}}=\frac{1}{N^{k}}
\end{equation}
and
\begin{align}
&\quad\Pru{\widetilde{x}\in \ZmodN^{
\calV\times [m]
},\,w_{e}\sim\mu_{\sfe_{i}}}{(e,f_{i})\text{ is placed in }\calI_{\bfY}\text{ and satisfied by }\widetilde{\tau_{x}}\,\middle|\, \supp(\bfy^{(\sfe_{i},j)})}\nonumber\\
&=\Pru{\widetilde{x}_{|e}\in\ZNk,\,w_{e}\sim\mu_{\sfe_{i}}}{\widetilde{x}_{e}=w_{e}\quad\text{and}\quad\widetilde{x}_{|e}\in \bigcup_{b \in \Sigma^k,\, f_i(b) = 1}
q_{\sfv_{i,1}}^{-1}(b_1) \times \cdots \times q_{\sfv_{i,k}}^{-1}(b_k)}\nonumber\\
&=\frac{1}{N^{k}}\sum_{b\in \Sigma^{k},\, f_{i}(b)=1}\mu_{\sfe_{i}}\left(q_{\sfv_{i,1}}^{-1}(b_1) \times \cdots \times q_{\sfv_{i,k}}^{-1}(b_k)\right)=\frac{1}{N^{k}}\sum_{b\in\Sigma^{k}}f_{i}(b)z^{*}_{i,b}.\label{eq:prob_constraint_satisfied}
\end{align}

\textbf{The $c=1$ case:} in this case $\lpval_{\calI}\defeq c=1$, which means $\frac{1}{m}\sum_{i=1}^{m}\sum_{b\in \Sigma^{k}}f_{i}(b)z^*_{i,b}=1$. By \Cref{obs:BasicLP}, we must have $\sum_{b\in \Sigma^{k}}f_{i}(b)z^*_{i,b}=1$ for all $i\in [m]$. From \eqref{eq:prob_constraint_placed} and \eqref{eq:prob_constraint_satisfied} we know that this implies all constraints that are placed in $\calI_{\bfY}$ by any player are satisfied by $\widetilde{\tau_{x}}$ with probability 1. This proves \eqref{eq:reduction_perfect_completeness}.

\textbf{The $c<1$ case:} we let $X^{(i,j)}$ be the number of constraints placed by the player $(\sfe_{i},j)$ (in other words, the length of the sequence $\calC^{(i,j)}$) and let $Z^{(i,j)}$ be the number of those constraints satisfied by $\widetilde{\tau_{x}}$. Due to independence among edges $e\in \supp(\bfy^{(\sfe_{i},j)})$ and Hoeffding's inequality (\Cref{prop:Hoeffding}), we have
\[
\Pru{\bfY\sim\calD_{\yes}}{X^{(i,j)}\geq (1+\varepsilon/2)N^{-k}\alpha n\,\middle|\,\supp(\bfy^{(\sfe_{i},j)})}\leq \exp\left(-\alpha n\cdot \varepsilon^{2}N^{-2k}/16\right)
\]
by \eqref{eq:prob_constraint_placed} and
\[
\Pru{\bfY\sim\calD_{\yes}}{Z^{(i,j)}\leq \left(\big(\textstyle\sum_{b\in \Sigma^{k}}f_{i}(b)z^*_{i,b}\big)-\varepsilon/2\right)N^{-k}\alpha n\,\middle|\,\supp(\bfy^{(\sfe_{i},j)})}\leq \exp\left(-\alpha n\cdot \varepsilon^{2}N^{-2k}/16\right)
\]
by \eqref{eq:prob_constraint_satisfied}.
Now taking expectation over $\supp(\bfy^{(\sfe_{i},j)})$ and taking union bound over all players $(\sfe_{i},j)\in \calE\times [K]$, it follows that with probability at least 
\[1-2mK\exp\left(-\alpha n\cdot \varepsilon^{2}N^{-2k}/16\right)=1-o_{n}(1)\] 
over $\bfY\sim\calD_{\yes}$, we have both
\[\sum_{(i,j)\in [m]\times[K]}X^{(i,j)}\leq (1+\varepsilon
/2)mK\cdot N^{-k}\alpha n\] and
\[
\sum_{(i,j)\in [m]\times [K]}Z^{(i,j)}\geq \left(\frac{1}{m}\big(\textstyle\sum_{i=1}^{m}\textstyle\sum_{b\in \Sigma^{k}}f_{i}(b)z^*_{i,b}\big)-\varepsilon/2\right)mK\cdot N^{-k}\alpha n = (c-\varepsilon/2)mK\cdot N^{-k}\alpha n
\]
and hence 
\[
\val_{\calI_{\bfY}}(\widetilde{\tau_{x}})=\frac{\sum_{(i,j)\in [m]\times [K]}Z^{(i,j)}}{\sum_{(i,j)\in [m]\times [K]}X^{(i,j)}}\geq \frac{c-\varepsilon/2}{1+\varepsilon/2}\geq c-\varepsilon.
\]
This proves \eqref{eq:reduction_completeness}.
    
    \paragraph{Step 5: soundness.} In order to upper bound $\val_{\calI_{\bfY}}$ with high probability, we upper bound the value of any fixed assignment $\widetilde{\tau}: \calV\times [n] \rightarrow \Sigma$ under $\calI_{\bfY}$ with high probability. Similarly to the proof of \Cref{lem:reduction_bounded_degree}, we define an associated \emph{random} assignment $\tau: \calV \rightarrow \Sigma$ by assigning value $\sigma\in \Sigma$ to $\sfv \in \calV$ with probability 
    \begin{align*}
        \frac{\Big|\big\{j\in [n]\,\big|\, \widetilde{\tau}((\sfv,j))=\sigma\big\}\Big|}{n },
    \end{align*}
    independently for each $v\in \calV$. 
    
    We define similar random variables as in Step 4: let $X^{(i,j)}$ be the number of constraints placed by the player $(\sfe_{i},j)$ into $\calI_{\bfY}$, and let $Z^{(i,j)}$ be the number of those constraints satisfied by $\widetilde{\tau}$. Note that unlike in Step 4, the assignment $\widetilde{\tau}$ is fixed, and all randomness lies in $\calI_{\bfY}$. According to \Cref{def:communication_game}, in the NO case, and conditioned on the support of $\bfy^{(\sfe_i,j)}$, each edge $e \in \supp(\bfy^{(\sfe_i,j)})$ contributes a constraint to $\calI_{\bfY}$ independently with probability $1/N^k$. Therefore, due to the further independence among players and Hoeffding's inequality (\Cref{prop:Hoeffding}), we have 
    \[\Pru{\bfY\sim\calD_{\no}}{\sum_{j=1}^{K}X^{(i,j)}\leq \left(1-\frac{\varepsilon}{4}\right)KN^{-k}\alpha n }\leq \exp\left(-\frac{K\alpha n}{64}\varepsilon^{2}N^{-2k}\right).\]
    
    To obtain a high-probability lower bound for $\sum_{j=1}^{K}Z^{(i,j)}$, we use \Cref{prop:concentration} in the same way as in the proof of \Cref{lem:reduction_bounded_degree}. We can think of each random matching $\supp(\bfy^{(\sfe_{i},j)})$ as the result of a sequential random selection of edges in the hypergraph $(\tcup\calU_{\sfe_{i}},\tprod\calU_{\sfe_{i}})$, without replacement of vertices. For a perfectly random edge $e\in \tprod\calU_{\sfe_{i}}$, we have
    \[
    \Pru{e\in \prod\calU_{\sfe_{i}}}{(e,f_{i})\text{ is satisfied by }\widetilde{\tau}}=\Pru{\tau}{C_{i}\text{ is satisfied by }\tau}.
    \]
    In the sequential random selection process, the number of available edges at any selection step is at least $(1-\alpha)^{k}n^{k}$. As in the proof of \Cref{lem:reduction_bounded_degree}, it follows that for any $t\in [\alpha n]$, the probability that the $t$-th selected edge in $\supp(\bfy^{(\sfe_{i},j)})$ contributes a constraint \emph{and} $\widetilde{\tau}$ satisfies it is at most 
    \[
    N^{-k}\cdot (1-\alpha)^{-k}\cdot\Pru{\tau}{C_{i}\text{ is satisfied by }\tau}.
    \]
    Note that the edge selection processes of all players in $\calE\times[K]$ are independent, and we may apply \Cref{prop:concentration} to processes of multiple players combined as a whole. Combining the processes of players $(\sfe_{i},j)$, where $j$ ranges in $[K]$, we conclude that
    \[
    \Pru{\bfY\sim\calD_{\no}}{\sum_{j=1}^{K}Z^{(i,j)}\geq \left((1-\alpha)^{-k}\cdot \Pru{\tau}{C_{i}\text{ is satisfied by }\tau}+\frac{\varepsilon}{4}\right)KN^{-k}\alpha n}\leq \exp\left(-\frac{K\alpha n}{64}\varepsilon^{2}N^{-2k}\right).
    \]
    Taking union bound over $i\in [m]$, it follows that with probability at least
    \[
    1-2m\exp\left(-\alpha n\cdot \varepsilon^{2}N^{-2k}/64\right)
    \]
    over $\bfY\sim\calD_{\no}$, we have both
    \begin{equation}\label{eq:soundness_X_lower}
    \sum_{(i,j)\in [m]\times [K]}X^{(i,j)}\geq (1-\varepsilon/4)m\cdot KN^{-k}\alpha n
    \end{equation}
    and
    \begin{align}
    \sum_{(i,j)\in [m]\times [K]}Z^{(i,j)}&\leq \left(\frac{1}{m}\left(\sum_{i=1}^{m}(1-\alpha)^{-k}\cdot\Prs{\tau}{C_{i}\text{ is satisfied by }\tau}\right)+\frac{\varepsilon}{4}\right)m\cdot KN^{-k}\alpha n\nonumber\\
    &=\left((1-\alpha)^{-k}\cdot \Exu{\tau}{\val_{\calI}(\tau)}+\frac{\varepsilon}{4}\right)m\cdot KN^{-k}\alpha n\nonumber\\
    &\leq \left((1+\varepsilon/4)\cdot s+\varepsilon/4\right)m\cdot KN^{-k}\alpha n.\label{eq:soundness_Z_upper}
    \end{align}
    In the last transition of \eqref{eq:soundness_Z_upper}, we used the definition $\Exs{\tau}{\val_{\calI}(\tau)}\leq \val_{\calI}\defeq s$ and the bound $(1-\alpha)^{-k}\leq e^{2k/\alpha}\leq 1+\varepsilon/4$ due to the choice $\alpha\leq  10^{-8}\varepsilon k^{-3}$. Whenever both \eqref{eq:soundness_X_lower} and \eqref{eq:soundness_Z_upper} hold, we have
    \[
    \val_{\calI_{\bfY}}(\widetilde{\tau})= \frac{\sum_{(i,j)\in [m]\times [K]}Z^{(i,j)}}{\sum_{(i,j)\in [m]\times [K]}X^{(i,j)}}\leq \frac{(1+\varepsilon/4)s+\varepsilon/4}{1-\varepsilon/4}\leq s+\varepsilon.
    \]
    Finally, taking a union bound over all $\widetilde{\tau}:\calV\times [n]\rightarrow\Sigma$, we conclude that $\val_{\calI_{\bfY}}\leq s+\varepsilon$ with probability at least
    \[
    1-|\Sigma|^{n|\calV|}\cdot 2m\exp\left(-\frac{K\alpha n}{64}\varepsilon^{2}N^{-2k}\right)\geq 1-o_{n}(1),
    \]
    due to the choice $K\geq 100\alpha^{-1}\varepsilon^{-2}N^{2k}\cdot|\calV|\log|\Sigma|$. This completes the proof of \eqref{eq:reduction_soundness} and the proof of the lemma.
\end{proof}

\section{Communication Lower Bound for DIHP}\label{sec:communication_lower_bound}
This section is devoted to the proof of~\Cref{thm:dihp_lowerbound}, which establishes the communication lower bound for the game $\DIHP(G, n, \alpha, K)$. As in~\cite{FMW25}, the argument follows the standard structure-vs.-randomness framework in communication complexity (see e.g. \cite{raz1997separation,goos2017query}), and consists of two main steps:

\begin{enumerate}
\item Given a communication protocol $\Pi$ with $|\Pi| \lesssim n^{1/3}$, we decompose the rectangles induced by $\Pi$ into smaller subrectangles. We show that, after decomposition, most subrectangles are ``good'' --- that is, each carries a well-structured piece of information combined with a controlled form of pseudorandom noise. This is done in the ``decomposition lemma'', \Cref{lem:regularity_decomposition}. 

\item For each ``good'' rectangle $R$, we establish a discrepancy bound of the form
\[
|\calD_\no(R) - \calD_\yes(R)| \leq 0.001 \cdot \calD_\no(R).
\]
This is done in the ``discrepancy lemma'',~\Cref{lem:discrepancy_bound}.

\end{enumerate}
These two steps are then combined to complete the proof of the communication lower bound; see~\Cref{lem:discrepancy_to_lower_bound}.

The decomposition lemma (\Cref{lem:regularity_decomposition}) closely follows its counterpart in~\cite{FMW25}, and its proof is deferred to~\Cref{app:regularity_decomposition}. In contrast, the discrepancy lemma (\Cref{lem:discrepancy_bound}) requires a different treatment than~\cite[Lemma 2.11]{FMW25}, and is developed in~\Cref{sec:global_rectangle,sec:Fourier_decay}. This (relatively short) section is devoted to laying out the overarching framework that connects these components. In particular, we formalize the notion of ``good'' rectangles in \Cref{subsec:pseudorandomness,subsec:good_rectangle}. Then, in \Cref{subsec:main_lemmas}, we lay out the main lemmas, from which we derive the desired communication lower bound in \Cref{subsec:communication_lower_bound}.

\subsection{Pseudorandomness Notions}\label{subsec:pseudorandomness}

A ``good'' rectangle is one in which the structural information and the pseudorandom noise are cleanly separated. In this subsection, we formalize the notions of pseudorandomness for sets of labeled matchings. This will allow us, in~\Cref{subsec:good_rectangle}, to control the pseudorandomness in rectangles. 

Throughout this subsection, we fix a $k$-universe $\calU=(U_{1},\dots,U_{k})$ and a positive integer $m\leq |\calU|$. We will consider pseudorandomness notions for the space of labeled matchings $\Omega^{\calU,m}$.
Our notion is based on the following type of restriction on the space $\Omega^{\calU,m}$. 

\begin{definition}\label{def:restrictions}
We define the set of \emph{restrictions} to be $\Omega^{\calU,\leq m}:=\bigcup_{0\leq d\leq m}\Omega^{\calU,d}$, i.e., the subset of $\Map{\tprod \calU}{\ZNk\cup\{\nil\}}$ that consists of all labeled matchings with at most $m$ edges. For each such labeled matching $\bfz\in\Omega^{\calU,\leq m}$, we let $\Omega^{\calU,m}_{\bfz}\subseteq \Omega^{\calU,m}$ be the restricted domain defined by
\[
\Omega^{\calU,m}_{\bfz}:=\left\{\bfy\in \Omega^{\calU,m}:\bfy(e)=\bfz(e)\text{ for all }e\in\supp(\bfz)\right\}.
\]
\end{definition}

In \Cref{def:restrictions}, restrictions are placed on the space of labeled matchings $\Omega^{\calU,m}$. Alternatively, we may view restrictions as directly acting on the universe $\calU$, as made precise by the following notation.

\begin{notation}
For a matching $M\in \calM_{\calU,\leq m}$, we denote by $\calU_{\setminus M}$ the $k$-universe $(U_{1}',U_{2}',\dots,U_{k}')$ defined by setting for each $i\in [k]$
$$U_{i}':=U_{i}\setminus\left\{u:\text{some edge of $M$ has $u$ as its $i$-th vertex}\right\}.$$
\end{notation}
\begin{remark}\label{rem:restriction_isomorphic}
A key observation is that the restricted domain \(\Omega^{\calU,m}_{\bfz}\) is naturally ``isomorphic'' to the unrestricted domain \(\Omega^{\calU_{\setminus M},\,m - |M|}\) associated to the smaller \(k\)-universe  \(\calU_{\setminus M}\), where \(M := \supp(\bfz)\).
\end{remark}
\begin{notation}
Given a matching $M\in\calM_{\calU,\leq m}$, we will use shorthand $\Omega^{\calU,m}_{\setminus M}$ to denote the space $\Omega^{\calU_{\setminus M},\,m-|M|}$.
\end{notation}

Before formalizing the main notion of pseudorandomness, we define the following convenient concept of subsumption of restrictions.

\begin{definition}
    For two restrictions $\bfz,\bfz'\in \Omega^{\calU,\leq m}$, we say $\bfz'$ \emph{subsumes} $\bfz$ if $\supp(\bfz)\subseteq \supp(\bfz')$ and for all $e\in \supp(\bfz)$ we have $\bfz(e)=\bfz'(e)$. 
\end{definition}

We are now ready to define pseudorandomness for sets of labeled matchings:

\begin{definition}\label{def:global_set}
A subset $A\subseteq \Omega^{\calU,m}$ is said to be \emph{$\bfz$-global} if $A\subseteq \Omega^{\calU,m}_{\bfz}$, and for all restrictions $\bfz'$ that subsume $\bfz$ we have 
\[
\frac{\left|A\cap \Omega^{\calU,m}_{\bfz'}\right|}{\left|\Omega^{\calU,m}_{\bfz'}\right|}\leq 2^{|\supp(\bfz')|-|\supp(\bfz)|}\cdot \frac{\left|A\cap \Omega^{\calU,m}_{\bfz}\right|}{\left|\Omega^{\calU,m}_{\bfz}\right|}.
\]
When $\bfz = \boldsymbol{0}$ is the trivial restriction, we simply say 
that $A$ is global (omitting the $\bfz$).
\end{definition}

In words, for a set $A$ and a restriction $\bfz$, we say that $A$ is 
$\bfz$-global if any further restrictions $\bfz'$ that subsumes $\bfz$ increases
the relative density of $A$ by factor at most $2^{|\supp(\bfz')|-|\supp(\bfz)|}$.

\begin{remark}\label{rem:restriction_isomorphic_2}
Continuing from \Cref{rem:restriction_isomorphic}, suppose \( A \subseteq \Omega^{\calU,m} \) is a \( \bfz \)-global subset. Then, under the natural identification between the restricted domain \( \Omega^{\calU,m}_{\bfz} \) and the unrestricted domain 
\[ \Omega^{\calU_{\setminus \supp(\bfz)},\,m - |\supp(\bfz)|}=\Omega^{\calU,m}_{\setminus\supp(\bfz)},\]
the set \( A \) corresponds to a subset \( A_\rem \subseteq \Omega^{\calU,m}_{\setminus\supp(\bfz)} \) that is \( \boldsymbol{0} \)-global. This correspondence follows directly from the definition of globalness and will play an important role in \Cref{sec:global_rectangle}. In particular, when a set \( A \subseteq \Omega^{\calU,m}_{\bfz} \) arises and the restriction \( \bfz \) is clear from context, we will use the same notation \( A_\rem \) to denote the corresponding subset of the domain \( \Omega^{\calU,m}_{\setminus \supp(\bfz)} \).
\end{remark}

\subsection{``Good'' Rectangles}\label{subsec:good_rectangle}

Now, we turn to pseudorandomness notions for \textit{rectangles}. In this subsection, we fix a distribution-labeled $k$-graph $G=(\calV,\calE,N,(\mu_{\sfe})_{\sfe\in \calE]})$ and a communication game $\DIHP(G,n,\alpha,K)$.

Recall from \Cref{def:communication_game} that in the communication game \(\DIHP(G, n, \alpha, K)\), the joint input to the $|\calE| \cdot K$ players is an element $\bfY$ in the product space \(\prod_{(\sfe, j) \in \calE \times [K]} \Omega^{\calU_{\sfe}, \alpha n}\). As is standard in communication complexity, a subset of this product space that is a Cartesian product is referred to as a \emph{rectangle}, formally defined below.

\begin{definition}
A subset \(R \subseteq \prod_{(\sfe, j) \in \calE \times [K]} \Omega^{\calU_{\sfe}, \alpha n}\) is called a \emph{rectangle} if it is a Cartesian product of sets \(A^{(\sfe, j)} \subseteq \Omega^{\calU_{\sfe}, \alpha n}\), one for each \((\sfe, j) \in \calE \times [K]\); that is,
\[
R = \prod_{(\sfe, j) \in \calE \times [K]} A^{(\sfe, j)}.
\]
\end{definition}
Then, it is natural to extend our definitions of global sets to rectangles, which requires each component $A^{(\sfe,j)}$ to be a global set. 
\begin{definition}\label{def:structured_rectangle}
Let \(\bdzeta = \left(\bfz^{(\sfe, j)}\right)_{(\sfe, j) \in \calE \times [K]}\) be a sequence where each \(\bfz^{(\sfe, j)}\) is a restriction on the space \(\Omega^{\calU_{\sfe}, \alpha n}\). A rectangle \(R = \prod_{(\sfe, j) \in \calE \times [K]} A^{(\sfe, j)}\) is called \emph{\(\bdzeta\)-global} if each set \(A^{(\sfe, j)}\) is \(\bfz^{(\sfe, j)}\)-global. When a rectangle $R$ is $\bdzeta$-global, we also say that the pair $(\bdzeta,R)$ is a \emph{structured rectangle}.
\end{definition}
We are now ready to give the formal definition of ``good'' rectangles. The rationale behind the three technical requirements in the following definition will become clear in \Cref{sec:global_rectangle}.

\begin{definition}\label{def:good_rec}
    Let $W$ be a positive real number. We say a structured rectangle $(\boldsymbol{\zeta},R)$, where \(R=\prod_{(\sfe, j) \in \calE \times [K]} A^{(\sfe, j)}\) and \(\bdzeta = \left(\bfz^{(\sfe, j)}\right)_{(\sfe, j) \in \calE \times [K]}\), is $W$-good if the following conditions hold:
    \begin{enumerate}[label=(\arabic*)]
        \item The hyperedge sets \(\left(\supp(\bfz^{(\mathsf{e},j)})\right)_{(\sfe,j)\in \calE\times [K]}\) are pairwise disjoint, and their union does not contain any cycle (for the definition of cycle-freeness in hypergraphs, see \Cref{subsec:general_notations}).
        \item $\sum_{(\mathsf{e},j)}\left|\supp(\bfz^{(\mathsf{e},j)} )\right| \leq W$.
        \item  $\left|A^{(\sfe,j)}\right| / \left|\Omega^{\calU_{\sfe},\alpha n}_{\bfz^{(\sfe,j)}}\right|\geq 2^{-W}$ for all $(\sfe,j)\in \calE\times [K]$. 
    \end{enumerate}
\end{definition}

\subsection{Two Main Lemmas}\label{subsec:main_lemmas}

We now present the two main lemmas as promised in the introductory text of \Cref{sec:communication_lower_bound}.

The \emph{decomposition lemma} captures the following fact: a protocol $\Pi$ induces at most $2^{|\Pi|}$ of rectangles, most rectangles $R$ of measure $\calD_\no (R) \gtrsim 2^{-|\Pi|}$; then, if $|\Pi|\lesssim \sqrt{n}$ holds, one can decompose those rectangles into structured rectangles while most of the structured rectangles are $\Theta(|\Pi|)$-good. 

\begin{restatable}[Decomposition lemma]{lemma}{DecompositionLemma}\label{lem:regularity_decomposition}
    Fix a distribution-labeled $k$-graph $G=(\calV,\calE,N,(\mu_{\sfe})_{\sfe\in \calE})$, an integer $K>0$ and a parameter $\alpha>0$. There exists a constant $\eta>0$ such that given any communication protocol $\Pi$ for $\DIHP(G,n,\alpha,K)$ with $|\Pi|\leq \eta \sqrt{n}$, there exists a collection $\calR$ of pairwise-disjoint structured rectangles $(\boldsymbol{\zeta},R)$ in the space \(\prod_{(\sfe, j) \in \calE \times [K]} \Omega^{\calU_{\sfe}, \alpha n}\) such that the following conditions hold: 
    \begin{enumerate}[label=(\arabic*)]
        \item \(\calD_{\no}\left(\bigcup_{(\boldsymbol{\zeta},R)\in \calR}R\right)\geq 0.99\).
        \item Each $(\boldsymbol{\zeta},R)\in\calR$ is $\left(10^5 \cdot|\Pi|\right)$-good. 
        \item For each $(\boldsymbol{\zeta},R)\in \calR $, there exists $a_R \in \{0,1\}$ such that $\Pi(\bfY) = a_R$ for every $\bfY \in R$. 
    \end{enumerate}
\end{restatable}
The proof of Lemma \ref{lem:regularity_decomposition} is included in Appendix \ref{app:regularity_decomposition}. Furthermore, for the good rectangles, we have the following discrepancy bound:
\begin{restatable}[Discrepancy lemma]{lemma}{discrepancybound}\label{lem:discrepancy_bound}
    Fix a distribution-labeled $k$-graph $G=(\calV,\calE,N,(\mu_{\sfe})_{\sfe\in \calE})$, an integer $K>0$ and a parameter $\alpha\in \big(0,10^{-8}k^{-3}\big]$. There exists a constant $\gamma=\gamma(G,\alpha,K)>0$ such that for any $(\gamma n^{1/3})$-good structured rectangle $(\boldsymbol{\zeta},R)$, we have
    \begin{align*}
        |\calD_{\no}(R) - \calD_{\yes}(R)|\leq 0.001\cdot \calD_{\no}(R).
    \end{align*}
\end{restatable}

The proof of \Cref{lem:discrepancy_bound} will take up \Cref{sec:global_rectangle,sec:Fourier_decay}.

\subsection{The Communication Lower Bound}\label{subsec:communication_lower_bound}
We note that~\Cref{lem:regularity_decomposition} and~\Cref{lem:discrepancy_bound} differ in their tolerance with respect to the parameter $n$. In particular, the goodness parameter in~\Cref{lem:discrepancy_bound} scales as $n^{1/3}$, which is the primary reason why our space lower bound in~\Cref{thm:main} is limited to $\Omega(n^{1/3})$. If, instead, we were able to establish the discrepancy bound for rectangles with goodness parameter as large as $\Theta(n^{1/2})$, the lower bound on space would improve to $\Omega(n^{1/2})$, as formalized in the following lemma.

\begin{lemma}\label{lem:discrepancy_to_lower_bound}
     Fix a distribution-labeled $k$-graph $G=(\calV,\calE,N,(\mu_{\sfe})_{\sfe\in \calE})$, an integer $K>0$ and a parameter $\alpha>0$. There exists a constant $\eta >0$ such that for every $W\leq \eta \sqrt{n}$, if $|\calD_{\no}(R) - \calD_{\yes}(R)|\leq 0.001\cdot \calD_{\no}(R)$ holds for every $W$-good structured rectangle $(\bdzeta,R)$, then we have $\mathsf{CC}(G,n,\alpha ,K ) \geq 10^{-5}\cdot W$.  
\end{lemma}
\begin{proof} 
    Take $\eta$ to be the constant obtained from Lemma \ref{lem:regularity_decomposition}. We fix a communication protocol $\Pi$ with $|\Pi|\leq 10^{-5}\cdot W$, and we show that $\adv(\Pi)<0.1$. 
    
    Since we have $|\Pi |\leq 10^{-5}\cdot W \leq \eta \sqrt{n}$, we may apply Lemma~\ref{lem:regularity_decomposition} and obtain a collection $\calR$ of structured rectangles. We know that each pair $(\boldsymbol{\zeta},R)\in \calR$ is $\left(10^5\cdot |\Pi|\right)$-good, which is also $W$-good since $10^{5}\cdot |\Pi|\leq W$. By the assumption in the statement, we then have 
    \begin{align*}
        |\calD_{\no}(R) - \calD_{\yes}(R)|\leq 0.001\cdot \calD_{\no}(R),
    \end{align*}
    holds for each $(\boldsymbol{\zeta},R)\in \calR$. 
    
    Note that for every $(\boldsymbol{\zeta},R)\in \calR$, the output of $\Pi$ is constant on $R$. By \Cref{def:advantage}, we have
    \begin{align*}
        \mathrm{adv}(\Pi) &= \left|\Pru{\bfY \sim \cal{D}_\yes}{\Pi(\bfY) =1} - \Pru{\bfY\sim \calD_\no}{\Pi(\bfY) = 1}\right| \\
        &\leq \Pru{\bfY\sim\calD_{\yes}}{\bfY\not\in\bigcup_{(\bdzeta,R)\in \calR}R}+\Pru{\bfY\sim\calD_{\no}}{\bfY\not\in\bigcup_{(\bdzeta,R)\in \calR}R}+ \sum_{(\bdzeta,R)\in \mathcal{R}} \left|\calD_\yes (R) - \calD_\no (R)\right| \\
        &\leq 2\cdot\Pru{\bfY\sim\calD_{\no}}{\bfY\not\in\bigcup_{(\bdzeta,R)\in \calR}R} + 2\cdot \sum_{(\bdzeta,R)\in \calR} \left| \calD_\yes (R) - \calD_\no (R)\right| \\
        &\leq 2(1-0.99) + 2\sum_{(\bdzeta,R)\in \calR} 0.001\cdot \calD_\no (R) < 0.1,
    \end{align*}
    as desired. 
\end{proof}

\Cref{thm:dihp_lowerbound} then follows easily from \Cref{lem:discrepancy_bound}.

\begin{proof}[Proof of \Cref{thm:dihp_lowerbound}]
Let $\eta$ be the constant obtained from \Cref{lem:discrepancy_to_lower_bound}. Lemma \ref{lem:discrepancy_bound} tells us that there exists a constant $\gamma>0$ such that $|\calD_{\no}(R) - \calD_{\yes}(R)|\leq 0.001\cdot \calD_{\no}(R)$ holds for all $\left(\gamma n^{1/3}\right)$-good rectangles, and it is clear that $\gamma n^{1/3}\leq \eta \sqrt{n}$ for large enough $n$. So from \Cref{lem:discrepancy_to_lower_bound} we conclude that $\mathsf{CC}(G,n,\alpha ,K ) = \Omega(\gamma n^{1/3}) = \Omega(n^{1/3})$. 
\end{proof}
\section{Bounding the Discrepancy of Good Rectangles}\label{sec:global_rectangle}
The goal of this section is to prove \Cref{lem:discrepancy_bound}, modulo a Fourier analytic lemma that we prove in \Cref{sec:Fourier_decay}. We begin with a high-level overview of the proof strategy.

Recall that in \Cref{lem:discrepancy_bound}, we are given a restriction sequence \(\bdzeta = \left(\bfz^{(\sfe, j)}\right)_{(\sfe, j) \in \calE \times [K]}\) and a $\bdzeta$-global rectangle $R = \prod_{(\sfe,j)\in \calE\times [K]} A^{(\sfe,j)}$. We will define a probability density function $f^{(\sfe,j)}$ on $\ZmodN^{\calV\times [n]}$, induced by the set $A^{(\sfe,j)}$. It turns out that $\calD_{\yes}(R)$ and $\calD_{\no} (R)$ can then be related by the identity (see \Cref{lem:relating_yes_no})
\begin{align}\label{eq:relating_yes_no}
\calD_{\yes}(R)=\calD_{\no}(R)\cdot \Exu{\widetilde{x}\in \ZmodN^{\calV\times [n]}}{\prod_{(\sfe,j)\in \calE\times [K]}f^{(\sfe,j)}(\widetilde{x})}.
\end{align}
Thus, it suffices to show that the expectation of the product \(\prod_{(\sfe,j)} f^{(\sfe,j)}(\widetilde{x})\) is close to 1.

Now the main difference from the setting of \cite{FMW25} arises. In \cite{FMW25}, each of the probability density functions in the product is supported on and approximately uniform over an affine subspace of the underlying vector space. This structure allows the analysis to proceed by restricting to the intersection of these affine subspaces and examining the product of the functions on this intersection.

In our setting, the functions \(f^{(\sfe,j)}\) are not close to uniform on their support. Instead, we decompose each \(f^{(\sfe,j)}\) as a product \(g^{(\sfe,j)} \cdot h^{(\sfe,j)}\), where \(g^{(\sfe,j)}\) captures the structured component of \(A^{(\sfe,j)}\), and \(h^{(\sfe,j)}\) models the pseudorandom noise. The structured-only product \(\prod_{(\sfe,j)} g^{(\sfe,j)}\) will play a similar role to the intersection of affine subspaces in~\cite{FMW25}, while each \(h^{(\sfe,j)}\) is expected to be close to uniform over the entire space \(\ZmodN^{\calV \times [n]}\).

When we combine the structured-only product and the pseudorandom parts $h^{(\sfe,j)}$ together, we will analyze the overall product using what we call a ``hybrid method'' (see \Cref{subsec:hybrid}). This method is inspired by~\cite[Lemma 3.12]{FMW25}, but also draws from classical hybrid arguments in streaming lower bounds (e.g.,~\cite{KKS14,CGSV24}). 

\subsection{Relating YES and NO Distributions}\label{subsec:calculating_discrepancy}
The goal of this section is to give an explicit formula of the ratio $\calD_\yes(R)/\calD_\no(R)$. Since the YES distribution $\calD_{\yes}$ is generated by the Markov kernel in \Cref{def:Markov_kernel}, the main task is to analyze this Markov kernel. 

A Markov kernel from $\ZmodN^{\calV \times [n]}$ to $\Omega^{\calU,m}$ pushes forward a probability distribution from the former space to a distribution on the latter. At the same time, it also induces a pull-back operation, mapping functions defined on $\Omega^{\calU,m}$ to functions on $\ZmodN^{\calV \times [n]}$. For the Markov kernel defined in~\Cref{def:Markov_kernel}, it turns out that the pull-back perspective is the more convenient one for analysis. We denote this pull-back operator by the italic bold symbol $\bdP^{\calU,m}_{\mu}[\cdot]$, distinguishing it from the matrix expression $\bfP^{\calU,m}_{\mu}(\cdot, \cdot)$ to reflect that, while formally distinct, the two represent the same underlying Markov transition.

\begin{notation}\label{notation:linear_operator}
Fix a $k$-universe $\calU$, a nonnegative integer $m\leq |\calU|$, and a one-wise independent distribution $\mu$ over $\ZNk$. The (right) stochastic matrix \(\bfP^{\calU,m}_{\mu}: \mathbb{Z}_N^{\bigcup\calU} \times \Omega^{\calU,m} \to \mathbb{R}\), defined in \Cref{def:Markov_kernel}, can be viewed as a linear operator
\[
\bdP^{\calU,m}_{\mu} : L^2(\Omega^{\calU, m}) \to L^2\big(\mathbb{Z}_N^{\bigcup\calU}\big)
\]
given by
\[
\bdP^{\calU,m}_{\mu}[f](x) = \sum_{\bfy \in \Omega^{\calU_{\sfe}, \alpha n}} \bfP^{\calU,m}_{\mu}(x, \bfy) f(\bfy),
\]
for all \(x \in \mathbb{Z}_N^{\bigcup\calU}\) and \(f \in L^2(\Omega^{\calU,m})\).
\end{notation}
The operator $\bdP_\mu^{\calU,m}$ satisfies the following two basic properties. 
\begin{proposition}\label{prop:infinity_norm_contraction}
For any $f\in L^{2}(\Omega^{\calU,m})$, we always have $\left\|\bdP^{\calU,m}_{\mu}[f]\right\|_{\infty}\leq \|f\|_{\infty}$.
\end{proposition}
\begin{proof}
This is obvious since the value of $\bdP^{\calU,m}_{\mu}[f]$ at any input $x$ is a convex combination of function values of $f$.
\end{proof}

\begin{proposition}
The operator $\bdP^{\calU,m}_{\mu}$ maps a density function on $\Omega^{\calU,m}$ to a density function on $\ZmodN^{\bigcup\calU}$.
\end{proposition}
\begin{proof}
Since the matrix $\bfP^{\calU,m}_{\mu}(\cdot,\cdot)$ has only nonnegative entries, the operator $\bdP^{\calU,m}_{\mu}$ preserves nonnegativity. It suffices to check that $\bdP^{\calU,m}_{\mu}$ also preserves expected values. 

Now we need to revisit the definition of the matrix $\bfP^{\calU,m}_{\mu}(\cdot,\cdot)$ in \Cref{def:Markov_kernel}. In the Markov process given in \Cref{def:Markov_kernel}, it is clear that if $x\in \ZmodN^{\bigcup\calU}$ is chosen uniformly at random, then the output $\bfz$ of the process is also uniformly distributed in $\Omega^{\calU,m}$. This means that for any $\bfy\in\Omega^{\calU,m}$, the expected value $\Exu{x\in\ZmodN^{\bigcup\calU}}{\bdP^{\calU,m}_{\mu}(x,\bfy)}$ is equal to $1/\left|\Omega^{\calU,m}\right|$. Therefore, for any \(f\in L^{2}(\Omega^{\calU,m})\), we have
\[
\Exu{x\in \mathbb{Z}_N^{\bigcup\calU}}{\bdP^{\calU,m}_{\mu}[f](x)}=\sum_{\bfy\in\Omega^{\calU,m}}\left(f(\bfy)\Exu{x\in\ZmodN^{\bigcup\calU}}{\bdP^{\calU,m}_{\mu}(x,\bfy)}\right)=\Exu{\bfy\in\Omega^{\calU,m}}{f(\bfy)},
\]
and we conclude that the operator $\bdP^{\calU,m}_{\mu}$ preserves expected values.
\end{proof}

Armed with the operator formalism, we are now ready to prove the main lemma of this subsection that relates the YES and NO distributions (\Cref{lem:relating_yes_no}). The following notation is useful for stating the lemma as well as later throughout this section.
\begin{notation}\label{notation:density_function}
Given a $k$-universe $\calU$, a nonnegative integer $m\leq |\calU|$, and a nonempty set $A\subseteq \Omega^{\calU,m}$, the density function of the uniform distribution on $A$ is denoted by $\phi_{A}:\Omega^{\calU,m}\rightarrow[0,\infty)$, specifically defined as 
\[\phi_{A}(\bfy):=\begin{cases}
\left|\Omega^{\calU,m}\right|/|A|, &\text{if }\bfy\in A,\\
0, &\text{if }\bfy\not\in A.
\end{cases}\]
\end{notation}

\begin{lemma}\label{lem:relating_yes_no}
Fix a $\DIHP(G,n,\alpha,K)$ communication game, where $G=(\calV,\calE,N,(\mu_{\sfe})_{\sfe\in\calE})$. Given a rectangle $R=\prod_{(\sfe,j)\in\calE\times [K]}A^{(\sfe,j)}$, where $A^{(\sfe,j)}\subseteq \Omega^{\calU_{\sfe},\alpha n}$, we have
\[
\calD_{\yes}(R)=\calD_{\no}(R)\cdot \Exu{\widetilde{x}\in \ZmodN^{\calV\times [n]}}{\prod_{(\sfe,j)\in \calE\times [K]}\bdP^{\calU_{\sfe},\alpha n}_{\mu_{\sfe}}\Big[\phi_{A^{(\sfe,j)}}\Big]\circ\proj_{\sfe}(\widetilde{x})}.
\]
\end{lemma}

\begin{proof}
The result follows from direct calculation:
\begin{align*}
\calD_{\yes}(R)&=\Exu{\widetilde{x}\in \ZmodN^{\calV\times [n]}}{\prod_{(\sfe,j)\in \calE\times[K]}\left(\sum_{\bfy\in A^{(\sfe,j)}}\bfP^{\calU_{\sfe},\alpha n}_{\mu_{\sfe}}\Big(\proj_{\sfe}(\widetilde{x}),\bfy\Big)\right)}\\
&=\Exu{\widetilde{x}\in \ZmodN^{\calV\times [n]}}{\prod_{(\sfe,j)\in \calE\times[K]}\left(\sum_{\bfy\in \Omega^{\calU_{\sfe},\alpha n}}\bfP^{\calU_{\sfe},\alpha n}_{\mu_{\sfe}}\Big(\proj_{\sfe}(\widetilde{x}),\bfy\Big)\phi_{A^{(\sfe,j)}}(\bfy)\right)}\cdot\prod_{(\sfe,j)\in\calE\times [K]}\frac{\left|A^{(\sfe,j)}\right|}{\left|\Omega^{\calU_{\sfe},\alpha n}\right|}\\
&= \Exu{\widetilde{x}\in \ZmodN^{\calV\times [n]}}{\prod_{(\sfe,j)\in \calE\times [K]}\bdP_{\mu_\sfe}^{\calU_\sfe,\alpha n}\Big[\phi_{A^{(\sfe,j)}}\Big]\Big(\proj_{\sfe}(\widetilde{x})\Big)}\cdot\calD_{\no}(R).
\end{align*}
The definitions of $\calD_{\yes}$ and $\calD_{\no}$ in \Cref{def:communication_game} are used in the first and the third transitions above, respectively.
\end{proof}

\subsection{Separating Structured and Pseudorandom Parts}\label{sec:separating_structure_random}
Given a $\bfz$-global set $A\subseteq \Omega^{\calU,m}$, the goal of this subsection is to express the function $\bdP_{\mu}^{\calU,m}[\phi_{A}]$ as the product of two functions: the structured part and the pseudorandom part. Now, we first give a formal description of what the structured part looks like. 
\begin{definition}
Fix a $k$-universe $\calU$, a nonnegative integer $m$, and a one-wise independent distribution $\mu$ over $\ZNk$. Let $\bfz$ be a restriction on the space $\Omega^{\calU,m}$. We define a density function $g_{\bfz}:\ZmodN^{\bigcup \calU}\rightarrow [0,\infty)$ by
\[
g_{\bfz}(x):=\prod_{e\in \supp(\bfz)}N^{k}\mu\Big(x_{|e}-\bfz(e)\Big).
\]
\end{definition}
Next, we define the pseudorandom part. For that purpose, we introduce the following two notations. 
\begin{notation}
Fix a $k$-universe $\calU$, a nonnegative integer $m\leq |\calU|$, and a matching $M\in \calM_{\calU,\leq m}$. The canonical projection from $\ZmodN^{\bigcup\calU}$ to $\ZmodN^{\bigcup\calU_{\setminus M}}$ is denote by $\proj_{\setminus M}$.
\end{notation}
\begin{notation}
Suppose $\bfz$ is a restriction on a labeled matching space $\Omega^{\calU,m}$, and let $M:=\supp(\bfz)$. For an element $\bfy\in \Omega^{\calU,m}_{\bfz}$, we define $\bfy_{\setminus M}$ to be the restriction of the map $\bfy:\tprod \calU\rightarrow \ZNk\cup\{\nil\}$ to the set $\tprod\calU_{\setminus M}\subseteq \tprod \calU$. Therefore, $\bfy_{\setminus M}$ is an element of $\Omega^{\calU,m}_{\setminus M}$.
\end{notation}

Suppose $\bfz$ is a restriction on a labeled matching space $\Omega^{\calU,m}$. Recall from \Cref{rem:restriction_isomorphic_2} that a subset $A\subseteq \Omega^{\calU,m}_{\bfz}$ is identified with a subset $A_{\rem}\subseteq \Omega^{\calU_{M},\,m-|M|}=\Omega^{\calU,m}_{\setminus M}$, where $M:=\supp(z)$. Therefore, in addition to the density function $\phi_{A}$ defined on $\Omega^{\calU,m}$, we have another density function $\phi_{A_\rem}$ defined on $\Omega^{\calU,m}_{\setminus M}$. The following lemma establishes a relation between the pull-backs of $\phi_{A}$ and $\phi_{A_\rem}$ under the Markov operators.

\begin{lemma}\label{lem:separating_structure_pseudorandom}
Fix a $k$-universe $\calU$, a nonnegative integer $m\leq |\calU|$, and a one-wise independent distribution $\mu$ over $\ZNk$. Let $\bfz$ be a restriction on the space $\Omega^{\calU,m}$ with support $M:=\supp(\bfz)$, and let $A\subseteq \Omega_{\bfz}^{\calU,m}$. Then for every $x\in \ZmodN^{\bigcup \calU}$, we have
\begin{equation}\label{eq:separating_structure_from_randomness}
\bdP^{\calU,m}_{\mu}[\phi_{A}](x)=g_{\bfz}(x)\cdot \bdP_{\mu}^{\calU_{\setminus M},\,m-|M|}\Big[\phi_{A_\rem}\Big]\Big(\proj_{\setminus M}(x)\Big).
\end{equation}
\end{lemma}

\begin{proof}
For elements $x\in \ZmodN^{\bigcup\calU}$ and $\bfy\in \Omega^{\calU,m}_{\bfz}$, consider the value of $\bfP^{\calU,m}_{\mu}(x,\bfy)$, which is the probability of obtaining $\bfy$ in the sampling process of \Cref{def:Markov_kernel}. Since $\bfy(e) = \bfz(e)\neq \nil$ for all $e \in M$, the first requirement for producing $\bfy$ is that the uniformly random matching sampled in Step~1 of \Cref{def:Markov_kernel} contains $M$. This occurs with probability
\[
\frac{N^{k|M|} \cdot \left|\Omega^{\calU,m}_{\setminus M}\right|}{\left|\Omega^{\calU,m}\right|}.
\]
Conditioned on this event, the second requirement is that the labels on each $e \in M$ sampled in Step~3 of \Cref{def:Markov_kernel} match $\bfz(e)$. This occurs with probability
\[
\prod_{e \in M} \mu\big(x_{|e} - \bfz(e)\big).
\]
Finally, conditioned on the first two requirements, the third requirement is that the remainder of the labeled matching coincides with $\bfy_{\setminus M}$. This occurs with probability
\[
\bfP^{\calU_{\setminus M},\,m-|M|}_{\mu}\big(\proj_{\setminus M}(x),\,\bfy_{\setminus M}\big).
\]
Combining these factors, we have
\[
\bfP^{\calU,m}_{\mu}(x,\bfy)
= \frac{N^{k|M|} \cdot \left|\Omega^{\calU,m}_{\setminus M}\right|}{\left|\Omega^{\calU,m}\right|}
\cdot \prod_{e \in M} \mu\big(x_{|e} - \bfz(e)\big)
\cdot \bfP^{\calU_{\setminus M},\,m-|M|}_{\mu}\big(\proj_{\setminus M}(x),\,\bfy_{\setminus M}\big).
\]

It then follows from direct calculation that
\begin{align*}
\bdP^{\calU,m}_{\mu}[\phi_{A}](x)&=\frac{\left|\Omega^{\calU,m}\right|}{|A|}\sum_{\bfy\in A}\bfP^{\calU,m}_{\mu}(x,\bfy)\\
&=\frac{N^{k|M|}\cdot\left|\Omega^{\calU,m}_{\setminus M}\right|}{|A|}\left(\sum_{\bfy\in A}\bfP^{\calU_{\setminus M},\,m-|M|}_{\mu}\Big(\proj_{\setminus M}(x),\bfy_{\setminus M}\Big)\right)\prod_{e\in M}\mu(x_{|e}-\bfz(e))\\
&= \left(\prod_{e\in M}N^{k}\mu(x_{|e}-\bfz(e))\right)\cdot\frac{\left|\Omega^{\calU,m}_{\setminus M}\right|}{|A|}\sum_{\bfy\in A_\rem}\bfP^{\calU_{\setminus M},\,m-|M|}_{\mu}\Big(\proj_{\setminus M}(x),\bfy\Big)\\
&=g_{\bfz}(x)\cdot \bdP_{\mu}^{\calU_{\setminus M},\,m-|M|}\Big[\phi_{A_\rem}\Big]\Big(\proj_{\setminus M}(x)\Big). \qedhere
\end{align*}
\end{proof}
We remark that on the right hand side of \eqref{eq:separating_structure_from_randomness}, the first factor $g_{\bfz}(x)$ is the structured part of the function $\bdP^{\calU,m}_{\mu}[\phi_{A}]$, while the second factor is the pseudorandom part.

\subsection{Analyzing the Structured Part}
As promised in the introductory text of \Cref{sec:global_rectangle}, given a sequence of restrictions $\boldsymbol{\zeta} = \left(\bfz^{(\sfe,j)} \right)_{(\sfe,j)\in \calE\times [K]}$, we need to analyze the ``structured-only product''
\begin{equation}\label{eq:structured_only_product}
\prod_{(\sfe,j)\in \calE\times [K]}g_{\bfz^{(\sfe,j)}}\circ \proj_{\sfe}\, ,
\end{equation}
which is a nonnegative-valued function on $\ZmodN^{\calV\times [n]}$. A natural question is whether the function is still a density function, i.e., whether its expected value is 1.

It turns out that this is not always the case, as suggested by the following simple counterexample. Suppose there are two distinct players $(\sfe_1,j_2),(\sfe_2,j_2)$ such that $$\supp\left(\bfz^{(\sfe_1,j_1)}\right) \bigcap \supp\left(\bfz^{(\sfe_2,j_2)}\right)\neq \emptyset,$$ and $\mu_{\sfe_1}$ and $\mu_{\sfe_2}$ are two one-wise independent distributions with disjoint supports. Then, it is easy to see that the product of the two functions
\[\left(g_{\bfz^{(\sfe_{1},j_{1})}}\circ \proj_{\sfe_{1}}\right)\cdot \left(g_{\bfz^{(\sfe_{2},j_{2})}}\circ \proj_{\sfe_{2}}\right)\]
already collapses to 0, and hence the expected value of \eqref{eq:structured_only_product} is 0 in this case. 

Counterexamples of this type can be fixed by imposing the requirement that the supports in the sequence \(\left(\supp(\bfz^{(\mathsf{e},j)})\right)_{(\sfe,j)\in \calE\times [K]}\) are pairwise disjoint and their union does not contain any cycle. These conditions enable us to make use of the one-wise independent nature of the distributions, and the expected value of \eqref{eq:structured_only_product} is indeed 1 under these conditions. The formal statement and proof follow. 
\begin{lemma}\label{lem:structure_part}
Fix a $\DIHP(G,n,\alpha,K)$ communication game, where $G=(\calV,\calE,N,(\mu_{\sfe})_{\sfe\in \calE})$. Let $\bdzeta=\left(\bfz^{(\sfe, j)}\right)_{(\sfe, j) \in \calE \times [K]}$ be a sequence of restrictions, where each restriction $\bfz^{(\sfe, j)}$ is on $\Omega^{\calU_{\sfe},\alpha n}$. If the hyperedge sets \(\left\{\supp(\bfz^{(\mathsf{e},j)})\right\}_{(\sfe,j)\in \calE\times [K]}\) are pairwise disjoint, and their union does not contain any cycle, then we have
\[
\Exu{\widetilde{x}\in \ZmodN^{\calV\times [n]}}{\prod_{(\sfe,j)\in \calE\times [K]}g_{\bfz^{(\sfe,j)}}\circ \proj_{\sfe}(\widetilde{x})}=1.
\]
\end{lemma}
\begin{proof}
Let 
\[
E = \bigcup_{(\sfe, j) \in \calE \times [K]} \supp\left(\bfz^{(\sfe,j)}\right)
\]
denote the set of all hyperedges appearing in the restrictions. Let $V(E)$ be the set of vertices in $\calV\times [n]$ covered by $E$. We run a breadth-first search (BFS) on the hypergraph $(V(E),E)$ and rank all edges in $E$ by the time they are discovered in the BFS. This yields a total order $\prec$ on $E$ such that each hyperedge $e\in E$ is incident to at most one vertex that is covered by some hyperedge preceding $e$ in the order. Indeed, if some hyperedge $e$ violates this condition, then by the time the BFS discovers $e$, it has also found a cycle of distinct hyperedges $e=e_{0},e_{1},e_{2},\dots,e_{\ell-1}\in E$ and distinct vertices $v_{0},v_{1},v_{2},\dots,v_{\ell-1}\in V(E)$ such that $e_{i}$ is incident to both $v_{i}$ and $v_{i+1\pmod \ell}$, for each $i\in \{0,1,\dots,\ell-1\}$. The $\ell$ hyperedges $e_{0},\dots,e_{\ell-1}$ thus together cover at most $\ell(k-1)$ vertices, violating the cycle-free assumption on $E$ (see \Cref{subsec:general_notations} for the definition of cycle-freeness in hypergraphs).

Since the edge sets $\supp(\bfz^{(\sfe,j)})$ are pairwise disjoint as $(\sfe,j)$ ranges in $\calE\times [K]$, we can define a map $\widetilde{\bfz}:E\rightarrow\ZNk$ by letting $\widetilde{\bfz}(e)=\bfz^{(\sfe,j)}(e)$ for each $(\sfe,j)\in \calE\times [K]$ and each $e\in \supp(\bfz^{(\sfe,j)})$. For each $e \in E$, let $\langle e \rangle$ denote the original hyperedge $\sfe \in \calE$ such that $e \in \prod \calU_{\sfe}$. Then, we have the identity:
\begin{equation} \label{eq:structure_part_rewrite}
\prod_{(\sfe, j) \in \calE \times [K]} g_{\bfz^{(\sfe, j)}} \circ \proj_{\sfe}(\widetilde{x}) = \prod_{e \in E} N^k \mu_{\langle e \rangle} \Big( \widetilde{x}_{|e} - \widetilde{\bfz}(e) \Big).
\end{equation}
Using the total ordering $\prec$ on $E$, we now analyze the expected value of the right-hand side of~\eqref{eq:structure_part_rewrite}. For any $e \in E$, we have:
\begin{align*}
&\quad \Exu{\widetilde{x} \in \ZmodN^{\calV \times [n]}} { \prod_{e' \preceq e} N^k \mu_{\langle e' \rangle} \Big( \widetilde{x}_{|e'} - \widetilde{\bfz}(e') \Big)  }\\
&= \Exu{\widetilde{x}} { \prod_{e' \prec e} N^k \mu_{\langle e' \rangle} \Big( \widetilde{x}_{|e'} - \widetilde{\bfz}(e') \Big) \cdot 
\Exu{\widetilde{x}}{ N^k \mu_{\langle e \rangle} \Big( \widetilde{x}_{|e} - \widetilde{\bfz}(e) \Big) \,\middle|\, \left( \widetilde{x}_{|e'} \right)_{e' \prec e}}}.
\end{align*}
Note that by our choice of the acyclic ordering $\prec$, we know that conditioning on $(\widetilde{x}_{\mid e'})_{e'\prec e}$ only fixes at most one coordinate of the coordinates in $e$. If no coordinate is fixed, it is easy to see that the inner conditional expectation evaluates to $1$. Otherwise, suppose the $i$-th vertex $v$ of $e$ is fixed to $b\in \ZmodN$ by the conditioning. In this case, the inner conditional expectation equals
\begin{align*}\Exu{\widetilde{x}}{ N^k \mu_{\langle e \rangle} \Big( \widetilde{x}_{|e} - \widetilde{\bfz}(e) \Big) \,\middle|\, \left( \widetilde{x}_{|e'} \right)_{e' \prec e}} = N\cdot \sum_{z\in \bZ_N^k,\, z_{i} = b} \mu_{\langle e\rangle}\left(z - \widetilde{\bfz}(e)\right) = N\cdot \frac{1}{N} =1 .
\end{align*}
due to the one-wise independence of $\mu_{\langle e \rangle}$.
Therefore, the overall expectation remains unchanged when we remove the term associated with $e$.

By applying this argument recursively --- removing the maximal element under $\prec$ at each step --- we can eliminate all hyperedges in $E$ without affecting the expectation. Consequently, the expected value of~\eqref{eq:structure_part_rewrite} is equal to 1, as claimed.
\end{proof}

\subsection{Analyzing the Pseudorandom Part}
For the pseudo-random part $\bdP_{\mu}^{\calU_{\setminus M},\,m-|M|}\Big[\phi_{A_\rem}\Big]$ in the decomposition \eqref{eq:separating_structure_from_randomness}, we will show that it has a good ``Fourier-decay'' property, defined as follows. 
\begin{definition}\label{def:Fourier_decay}
Suppose $\Lambda$ is a ground set, $n$ is a positive integer at most $|\Lambda|$, and $w$ is a real number in the range $(0,|\Lambda|)$. We say a density function $f:\bZ_{N}^{\Lambda}\rightarrow[0,\infty)$ is $(n,w)$-decaying if for every nonnegative integer $\ell\leq |\Lambda|$ we have $\left\|f^{=\ell}\right\|_{2}^{2}\leq F(n,\ell,w)$, where $F(n,\ell,w)$ is defined by
$$F(n,\ell,w)= 
\begin{cases}
\left(\frac{w}{n}\right)^{\ell/2}, &\text{if }0\leq \ell\leq w,\\
\left(\frac{\ell}{8n}\right)^{\ell/2}\cdot 2^{2w}, &\text{if }w<\ell \leq n,\\
0, &\text{if }\ell >n.
\end{cases}
$$
\end{definition}
In particular, it is not hard to see that a probability density function that is, say, $(n,n/2)$-decaying, must be close to uniform. 

The following simple observation ensures that the Fourier weight bound $F(n,\ell,w)$ is convenient to work with.

\begin{proposition}\label{prop:monotonicity_F}
For fixed $n$ and $\ell$, the bound $F(n,\ell,w)$ is monotone increasing in $w$.
\end{proposition}
\begin{proof}
Note that the function $F(n,\ell,w)$ is continuous and piecewise differentiable in $w$, in the range $w\in (0,\infty)$. The conclusion then follows by verifying that the partial derivative of $F(n,\ell,w)$ in $w$ is always positive.
\end{proof}
We now present the desired lemma that proves Fourier decay properties for pull-backs of density functions of the form $\phi_{A}$, where $A$ is a global set. 
\begin{restatable}{lemma}{globalimpliesdecay}\label{lem:global-decay}
Fix a $k$-universe $\calU$, an integer $m$ and a real number $w>0$ such that $|\calU|\geq 10^{8}k^{3}m$ and $m\geq 2(w+1)$. Let $A\subseteq \Omega^{\calU,m}$ be a global set with $|A|= 2^{-w}\cdot \left|\Omega^{\calU,m}\right|$. Then the density function $\bdP^{\calU,m}_{\mu}[\phi_{A}]$ is $(|\calU|,w)$-decaying, for any one-wise independent distribution $\mu$ over $\ZNk$.
\end{restatable}
\begin{proof}
The proof is deferred to \Cref{sec:Fourier_decay}.
\end{proof}
\subsection{The Hybrid Method}\label{subsec:hybrid}
In order to combine the structured-part result \Cref{lem:structure_part} and the pseudorandom-part result \Cref{lem:global-decay}, we need the following important lemma. The proof of this lemma somewhat resembles the hybrid arguments used in previous works such as \cite{KKS14,CGSV24} to extract two-player communication games from multi-player ones.
\begin{lemma}
    \label{lem:product_mixing}
    For any nonnegative integer $r$, there exists a constant $\delta=\delta(r)>0$ such that the following holds. Suppose $n$ is a sufficiently large integer and $\Lambda$ is a ground set with $|\Lambda|\geq n$. For any density functions $h_{0},h_{1},\dots,h_{r}: \mathbb{Z}_N^{\Lambda}\rightarrow [0,\infty)$ such that 
    \begin{enumerate}[label=(\arabic*)]
        \item $\|h_{i}\|_{\infty}\leq 2^{\delta n^{1/3}}$ for all $i\in \{0,1,\dots,r\}$, and
        \item $h_{i}$ is $\left(n/2,\delta n^{1/3}\right)$-decaying for all $i\in \{1,2,\dots,r\}$,
    \end{enumerate}
    we have 
    \begin{align*}
        \left|\Exu{x\in \mathbb{Z}_N^{\Lambda}}{\prod_{i=0}^r h_i(x)} -1 \right| \leq  0.001. 
    \end{align*}
\end{lemma}

\begin{proof} 
We let $\delta(r):=10^{-6}N^{-1}r^{-3}$. Since $\Exs{x\in \ZmodN^{\Lambda}}{h_{0}(x)}=1$, the statement clearly holds for $r=0$. We proceed by induction on $r$. In the following, assume $r\geq 1$ and the result holds for all smaller values of $r$.
Writing
\begin{equation}\label{eq:hybrid_decomposition}
\Exu{x\in \ZmodN^{\Lambda}}{\prod_{i=0}^{r}h_{i}(x)}-1=\Exu{x\in \ZmodN^{\Lambda}}{\prod_{i=0}^{r}h_{i}(x)}-\Exu{x\in\ZmodN^{\Lambda}}{h_{0}(x)}=
\sum_{j=1}^{r}\Exu{x\in \ZmodN^{\Lambda}}{(h_{j}(x)-1)\prod_{i=0}^{j-1}h_{i}(x)},
\end{equation}
it suffices to upper bound the absolute value of each summand in the sum above. 

Using the level decomposition, for each $j\in [r]$ we have 
\begin{align}
    \left|\Exu{x\in \ZmodN^{\Lambda}}{(h_{j}(x)-1)\prod_{i=0}^{j-1}h_{i}(x)}\right|
    &= \left|\sum_{\ell=0}^{|\Lambda|} \left\langle \left(h_j  -1\right)^{=\ell},\left(\prod_{i=0} ^{j-1} h_i\right)^{=\ell}\right\rangle\right| \nonumber\\
    &\leq  \sum_{\ell=1}^{|\Lambda|}\left\lVert \left(h_j -1\right)^{=\ell}\right\rVert_2 \cdot  \left\lVert  \left(\prod_{i=0} ^{j-1} h_i\right)^{=\ell}\right \rVert_2, \label{eq:hybrid_degree_decomposition}
\end{align}
where we use the Cauchy-Schwarz and the fact that $(h_j-1 )^{=0} \equiv 0$ to deduce the inequality. 

We next provide upper bounds on the level-$\ell$ Fourier weights of the two functions $h_j -1$ and $\prod_{i=0} ^{j-1} h_j$ separately. We know that $h_j $ is $\left(n/2,\delta n^{1/3}\right)$-decaying. Plugging in \Cref{def:Fourier_decay}, we obtain
\begin{equation}\label{eq:hybrid_decay}
\left\| (h_j-1)^{=\ell} \right\|_2^2\leq  
\begin{cases}
\left(2\delta n^{-2/3}\right)^{\ell/2}, &\text{if }1\leq \ell\leq \delta n^{1/3},\\
\left(\ell/(4n)\right)^{\ell/2}\cdot 2^{2\delta n^{1/3}}, &\text{if }\delta n^{1/3}<\ell\leq n,\\
0, &\text{if }\ell >n.
\end{cases}
\end{equation}

Since $\delta=\delta(r)\leq \delta(j-1)$ and $n$ is sufficiently large, we may use the induction hypothesis on $h_{0},\dots,h_{j-1}$ and obtain
\begin{align*}
    \left\lVert \prod_{i=0}^{j-1} h_i\right \rVert_1=1+\left(\Exu{x\in \ZmodN^{\Lambda}}{\prod_{i=0}^{j-1}h_{i}(x)}-1\right)\in \left[\frac{1}{2},\frac{3}{2}\right].
\end{align*}
Furthermore, since $\|h_{i}\|_{\infty}\leq 2^{\delta n^{1/3}}$ for all $i\in \{0,1,\dots,j-1\}$, the infinity norm (and hence the 2-norm) of $\prod_{i=0}^{j-1}h_{i}$ is at most $2^{\delta r n^{1/3}}$. We may apply \Cref{prop:level_d_classical} and get
\begin{equation}\label{eq:hybrid_classical_hypercontracitivity}
    \left\lVert \left(\prod_{i=0}^{j-1} h_i\right)^{=\ell}\right\rVert_2^2 \leq\left(\frac{3}{2}\right)^{2}\cdot \left(12 N\cdot (\delta r n^{1/3}+2)\right)^\ell \leq \left(50 N (\delta r n^{1/3}+2)\right)^\ell.
\end{equation}
Plugging \eqref{eq:hybrid_decay} and \eqref{eq:hybrid_classical_hypercontracitivity} into \eqref{eq:hybrid_degree_decomposition}, we get
\begin{align*}
     &\quad\left|\Exu{x\in \ZmodN^{\Lambda}}{(h_{i}(x)-1)\prod_{i=0}^{j-1}h_{i}(x)}\right|\\
    &\leq \sum_{\ell=1}^{\left\lfloor\delta n^{1/3}\right\rfloor} \left(2\delta n^{-2/3}\right)^{\ell/4}\cdot \left(50 N (\delta r n^{1/3}+2)\right)^{\ell/2} +\sum_{\ell=\left\lfloor\delta n^{1/3}\right\rfloor+1}^{n} \left(\frac{\ell}{4n}\right)^{\ell/4}\cdot 2^{\delta n^{1/3}} \cdot 2^{\delta r n^{1/3}} \\
     &\leq \sum_{\ell=1}^{\left\lfloor\delta n^{1/3}\right\rfloor} (100N\delta r)^{\ell/2} + \sum_{\ell=\left\lfloor\delta n^{1/3}\right\rfloor+1}^{\left\lfloor 2^{-4r-4}n\right\rfloor} \left(\frac{\ell }{4n}\cdot 2^{4r+4}\right)^{\delta n^{1/3}/4}+\sum_{\ell=\left\lfloor 2^{-4r-4}n\right\rfloor +1}^{n}4^{-\ell/4}\cdot 2^{(r+1)\delta n^{1/3}}\\
     &\leq \frac{0.001}{r}. 
\end{align*}
for sufficiently large $n$ and our choice of $\delta$. The conclusion then follows by \eqref{eq:hybrid_decomposition}.
\end{proof}

\subsection{Proof of the Discrepancy Bound}
Now, we have all the ingredients needed to prove  Lemma \ref{lem:discrepancy_bound}, restated below. 
\discrepancybound*
\begin{proof}
    Let $\bdzeta=\left(\bfz^{(\sfe,j)}\right)_{(\sfe,j)\in\calE\times [K]}$ and let $R=\prod_{(\sfe,j)\in \calE\times [K]}A^{(\sfe,j)}$. Let $h_{0}:\ZmodN^{\calV\times [n]}\rightarrow [0,\infty)$ be defined by
    \[
    h_{0}(\widetilde{x}):=\prod_{(\sfe,j)\in \calE\times [K]}g_{\bfz^{(\sfe,j)}}\circ \proj_{\sfe}(\widetilde{x}).
    \]
    For each $(\sfe,j)\in \calE\times [K]$, we have a function $h^{(\sfe,j)}:\ZmodN^{\calV\times [n]}\rightarrow[0,\infty)$ defined by
    \[
    h^{(\sfe,j)}:=\bdP_{\mu_\sfe}^{(\calU_{\sfe})_{\setminus M},\,\alpha n-|M|}\left[\phi_{A_\rem^{(\sfe,j)}}\right]\circ \proj_{\setminus M}\circ\proj_{\sfe},
    \]
    where $M$ stands for $\supp(\bfz^{(\sfe,j)})$. 
    
    By \Cref{lem:relating_yes_no,lem:separating_structure_pseudorandom}, we can now write
    \begin{align}
        \frac{\left|\calD_\yes(R) - \calD_\no (R)\right|}{\calD_\no(R)}  &= \left|\Exu{\widetilde{x}\in \ZmodN^{\calV\times [n]}}{\prod_{(\sfe,j)\in \calE\times [K]}\bdP^{\calU_{\sfe},\alpha n}_{\mu_{\sfe}}\Big[\phi_{A^{(\sfe,j)}}\Big]\circ\proj_{\sfe}(\widetilde{x})} -1 \right|\nonumber\\
        &= \left|\Exu{\widetilde{x}\in\ZmodN^{\calV\times [n]}}{h_{0}(\widetilde{x})\cdot\prod_{(\sfe,j)\in \calE\times[K]}h^{(\sfe,j)}(\widetilde{x})}-1\right|.\label{eq:RHS_to_bound}
    \end{align}
    It suffices to upper bound the right hand side above. 

    Note that since the infinity norm of each $g_{\bfz^{(\sfe,j)}}$ is clearly at most $N^{k\left|\supp(\bfz^{\sfe,j})\right|}$, and using the goodness assumption $\sum_{(\sfe,j)\in \calE\times[K]}\left|\supp(\bfz^{\sfe,j})\right|\leq \gamma n^{1/3}$, we have
    \begin{equation}\label{eq:h_0_infinity_norm}
    \left\|h_{0}\right\|_{\infty}\leq \prod_{(\sfe,j)\in\calE\times [K]}\left\|g_{\bfz^{(\sfe,j)}}\right\|_{\infty}\leq N^{k\cdot \gamma n^{1/3}}=2^{(k\log N)\gamma n^{1/3}}.
    \end{equation}
    The goodness assumption also implies
    \begin{align*}
        \left|A_\rem^{(\sfe,j)}\right| / \left|\Omega_{\bfz ^{(\sfe,j)}}^{\calU_\sfe,\alpha n}\right| \geq 2^{-\gamma n^{1/3}}\quad\text{and hence}\quad \left \lVert \phi_{A_\rem^{(\sfe,j)}}\right\rVert_\infty \leq 2^{\gamma n^{1/3}}.
    \end{align*}
    It then follows from \Cref{prop:infinity_norm_contraction} that
    \begin{equation}\label{eq:h_i_infinity_norm}
    \|h^{(\sfe,j)}\|_{\infty}\leq 2^{\gamma n^{1/3}}\quad \text{for each}\quad (\sfe,j)\in\calE\times [K].
    \end{equation}
    Finally, since we always have
    \[
    \left|(\calU_{\sfe})_{\setminus \supp(\bfz^{(\sfe,j)})}\right|=n-\left|\supp(\bfz^{(\sfe,j)})\right|\geq 10^{8}k^3\left(\alpha n-\left|\supp(\bfz^{(\sfe,j)})\right|\right),
    \]
    we may apply \Cref{lem:global-decay} to $A^{(\sfe,j)}_{\rem}$ and obtain that the function 
    \[
    p^{(e,j)}:=\bdP_{\mu_{\sfe}}^{\calU_{\setminus M},\,m-|M|}\left[\phi_{A_\rem^{(\sfe,j)}}\right],\quad\text{ where } M=\supp(\bfz^{(\sfe,j)})
    \]
    is $\left(\left|(\calU_{\sfe})_{\setminus\supp(\bfz^{(\sfe,j)})}\right|,\gamma n^{1/3}\right)$-decaying for all $(\sfe,j)\in \calE\times [K]$. Note that 
    $$\left|(\calU_{\sfe})_{\setminus\supp(\bfz^{(\sfe,j)})}\right|\geq n-\alpha n\geq n/2,$$
    which implies that $p^{(\sfe,j)}$ is $\left(n/2,\gamma n^{1/3}\right)$-decaying by \Cref{prop:monotonicity_F}. Since $h^{(\sfe,j)}$ has the same Fourier spectrum as $p^{(\sfe,j)}$ in the sense that
    \begin{align*}
        \widehat{h^{(\sfe,j)}}(b) = \begin{cases}
        0, &\text{if }\supp(b)\nsubseteq (\bigcup \calU_\sfe)_{\setminus \supp(\bfz^{(\sfe,j)})},\\
        \widehat{p^{(\sfe,j)}}\left(\proj_{\setminus M}\circ\proj_{\sfe}(b)\right), &\text{if }\supp(b)\subseteq (\bigcup \calU_\sfe)_{\setminus \supp(\bfz^{(\sfe,j)})}
        \end{cases}
    \end{align*}
    for any $b\in \ZmodN^{\calV\times [n]}$, we also have that
    \begin{equation}\label{eq:h_i_decaying}
    h^{(\sfe,j)}\quad\text{is}\quad \left(n/2,\gamma n^{1/3}\right)\text{-decaying, for }(\sfe,j)\in\calE\times [K].
    \end{equation}
    Due to the established properties \eqref{eq:h_0_infinity_norm}, \eqref{eq:h_i_infinity_norm} and \eqref{eq:h_i_decaying}, we may now apply \Cref{lem:product_mixing} to $h_{0}$ and $(h^{(\sfe,j)})_{(\sfe,j)\in\calE\times[K]}$. It follows that as long as $\gamma$ is chosen to be less than 
    $$\frac{1}{k\log N}\cdot\delta(K)=\frac{1}{k\log N}\cdot\frac{1}{10^{6}N(K+1)^{2}},$$
    and $n$ is sufficiently large, the right hand side of \eqref{eq:RHS_to_bound} is upper-bounded by $0.001$, as desired.
\end{proof}

\section{Fourier Decay from Global Hypercontractivity}
\label{sec:Fourier_decay}

The goal of this section is to prove~\Cref{lem:global-decay}. The high-level strategy follows the approach of~\cite[Section 4]{FMW25}. The first step is to establish a global hypercontractivity result for functions on $\Omega^{\calU, m}$, formulated as a projected level-$d$ inequality. The second step is to apply this inequality to the density function $\phi_A : \Omega^{\calU, m} \to [0, \infty)$, where $A\subseteq \Omega^{\calU,m}$, and to show how the resulting bound yields the desired Fourier decay for the function $\bdP^{\calU, m}_{\mu}[\phi_A]$.

Since the proof of the projected level-$d$ inequality\footnote{We are aware of an alternative, shorter proof of the projected level-$d$ inequality (yielding slightly weaker parameters) that avoids the machinery of \Cref{app:global_hypercontractivity}. We nevertheless include \Cref{app:global_hypercontractivity}, as it may provide additional conceptual insight.} closely follows that of~\cite[Section 4]{FMW25}, we defer the details to~\Cref{app:global_hypercontractivity}. Nonetheless, in order to formalize the inequality, we must first introduce several preliminary definitions. This preparatory material occupies~\Cref{subsec:Characters_on_Omega,subsec:derivatives,subsec:projections}. Then, in~\Cref{subsec:SVD}, we study structural properties of the operator $\bdP^{\calU, m}_{\mu}$ that enable the projected level-$d$ inequality for $\phi_A$ to imply Fourier decay of $\bdP^{\calU, m}_{\mu}[\phi_A]$. We conclude with a proof of \Cref{lem:global-decay} in \Cref{subsec:proof-of-decay}.

\subsection{Fourier Characters}\label{subsec:Characters_on_Omega}

We begin by introducing a collection of character functions on $\Omega^{\calU,m}$. The characters are indexed by pairs $(M,\bfa)$ where $M\in \calM_{\calU,\leq m}$ is a ``partial matching'' and $\bfa:M\rightarrow\ZNk\setminus\{0\}$ is a labeling on the edges in $M$. To facilitate the definition of the character functions, we first introduce the following probability values.

\begin{definition}\label{def:Psi}
For integers $n,m$ such that $n\geq m\geq 0$, we define $\Psi(n,m,0):=1$, and for $1\leq d\leq m$ we define inductively $\Psi(n,m,d):=mn^{-k}\cdot \Psi(n-1,m-1,d-1)$. 
\end{definition}

It is easy to see that $\Psi(n,m,d)$ is equal to the probability that a fixed partial matching $M\in \calM_{\calU,d}$, where $\calU$ is a $k$-universe of cardinality $n$, is contained in a uniformly random matching drawn from $\calM_{\calU,m}$. We are now ready to define the character functions:

\begin{definition}\label{def:characters}
For a matching $M\in\calM_{\calU,\leq m}$ and a map $\bfa:M\rightarrow \ZNk\setminus\{0\}$, we call $(M,\bfa)$ a character index on $\Omega^{\calU,m}$ and define the character function $\psi_{M,\bfa}:\Omega\rightarrow\bC$ by
$$\psi_{M,\bfa}(\bfy):=\Psi(|\calU|,m,|M|)^{-1/2}\cdot\prod_{e \in M}\chi_{\bfa(e)}(\bfy(e)).$$
Specially, we define $\chi_{\bfa(e)}(\bfy(e)) =0$ when $\bfy(e) = \nil$ (see \Cref{subsec:general_notations} for how $\chi_{\bfa(e)}(\cdot)$ is defined on standard inputs).
\end{definition}

Note that the character functions defined above do not form a complete basis for $L^{2}(\Omega^{\calU,m})$. Nevertheless, they will be sufficient for our purpose. The following proposition shows that these characters form an orthonormal set.

\begin{proposition}\label{prop:ortho_set}
For character indices $(M_{1},\bfa_{1})$ and $(M_{2},\bfa_{2})$ on $\Omega^{\calU,m}$, we have $\inp{ \psi_{M_{1},\bfa_{1}}}{\psi_{M_{2},\bfa_{2}}}=\mathbbm{1}\{M_{1}=M_{2}\text{ and }\bfa_{1}=\bfa_{2}\}$.
\end{proposition}

\begin{proof}
We divide the argument into cases. 

\textbf{Case 1:} $M_{1}\cup M_{2}$ is not a matching. Then for each $\bfy\in \Omega^{\calU,m}$, there exists an hyperedge $e\in M_{1}\cup M_{2}$ with $\bfy(e)=\nil$. This forces $$\psi_{M_{1},\bfa_{1}}(\bfy)\cdot \overline{\psi_{M_{2},\bfa_{2}}(\bfy)}=0$$ for each $\bfy\in \Omega^{\calU,m}$, and thus $\inp{ \psi_{M_{1},\bfa_{1}}}{\psi_{M_{2},\bfa_{2}}}=0$.

\textbf{Case 2:} $M_{1}\cup M_{2}$ is a matching. For each $e\in M_{1}\cup M_{2}$, let 
$$\bfa(e):=\begin{cases}
\bfa_{1}(e), &\text{if }e\in M_{1}\setminus M_{2},\\
\bfa_{1}(e)-\bfa_{2}(e), &\text{if }e\in M_{1}\cap M_{2},\\
-\bfa_{2}(e), &\text{if }e\in M_{2}\setminus M_{1}.
\end{cases}
$$
By \Cref{def:characters} we have
\begin{equation}\label{eq:inner_product_characters}
\inp{ \psi_{M_{1},\bfa_{1}}}{\psi_{M_{2},\bfa_{2}}}= \Psi(|\calU|,m,|M_{1}|)^{-1/2}\cdot\Psi(|\calU|,m,|M_{2}|)^{-1/2}\cdot\Exu{\bfy\in\Omega^{\calU,m}}{\prod_{e\in M_{1}\cup M_{2}}\chi_{\bfa(e)}(\bfy(e))}.
\end{equation}
For a uniformly random $\bfy\in\Omega^{\calU,m}$ conditioned on $\supp(\bfy)\supseteq M_{1}\cup M_{2}$, the labels $\{\bfy(e)\}_{e\in M_{1}\cup M_{2}}$ are independent and each uniformly distributed on $\ZNk$. Therefore, \eqref{eq:inner_product_characters} implies that $\inp{ \psi_{M_{1},\bfa_{1}}}{\psi_{M_{2},\bfa_{2}}}=0$ unless $\bfa(e)=0$ for all $e\in M_{1}\cup M_{2}$. Since $\bfa_{1}(e)\neq 0$ for all $e\in M_{1}$ and $\bfa_{2}(e)\neq 0$ for all $e\in M_{2}$ by the definition of characters, it follows that $\inp{ \psi_{M_{1},\bfa_{1}}}{\psi_{M_{2},\bfa_{2}}}=0$ unless $M_{1}=M_{2}$ and $\bfa_{1}=\bfa_{2}$, in which case the right hand side of \eqref{eq:inner_product_characters} clearly evaluates to 1.
\end{proof}

\subsection{Discrete Derivatives}\label{subsec:derivatives}

In this subsection, we introduce the notion of discrete derivatives for functions over $\Omega^{\calU,m}$, as well as a related notion of globalness. These concepts are necessary for the statement of the projected level-$d$ inequality.

Among all elements of the space $\Omega^{\calU,m}$, we sometimes need to consider those labeled matchings that contain a fixed partial matching $S\in \calM_{\calU,\leq m}$. Such labeled matchings are determined by two choices:

\begin{enumerate}
    \item Assigning labels to the edges in \( S \), which corresponds to selecting an element from \( \Map{S}{\ZNk} \).
    \item Choosing the remaining labeled matching on \( \calU_{\setminus S}\), which corresponds to an element in the space \( \Omega^{\calU_{\setminus S},\,m-|S|} \).
\end{enumerate}
This leads to the following definition:

\begin{definition}\label{def:embedding}
For a matching $S\in \calM_{\calU,\leq m}$, there is a canonical embedding 
\[
\mathfrak{i}:\Omega^{\calU_{\setminus S},\,m-|S|}\times\Map{S}{\ZNk}\hookrightarrow \Omega^{\calU,m}.
\]
This embedding proceeds by mapping a pair $(\bfy,\bfz)$ from the left hand side to the labeled matching $\boldsymbol{\xi}\in \Omega^{\calU,m}$ defined by 
\begin{enumerate}
\item $\boldsymbol{\xi}(e)=\bfy(e)$ for $e\in \tprod (\calU_{\setminus S})$,
\item $\boldsymbol{\xi}(e)=\bfz(e)$ for $e\in S$, and 
\item $\boldsymbol{\xi}(e)=\nil$ for all other $e\in \tprod \calU$. 
\end{enumerate}
\end{definition}

The following definition provides the key gadget for defining discrete derivatives on $\Omega^{\calU,m}$.

\begin{definition}\label{def:gadget_H}
Let $S$ be a finite set. For any map $\bfz\in \Map{S}{\ZNk}$, we define its Hamming weight $\|\bfz\|_{\sfH}$ to be the number of edges $e\in S$ such that $\bfz(e)$ is a nonzero element of the Abelian group $\ZNk$. We define a function $H_{S}:\Map{S}{\ZNk}\rightarrow\mathbb{C}$ by letting
$$H_{S}(\bfz):=(-1)^{\|\bfz\|_{\sfH}}\left(N^{k}-1\right)^{|S|-\|\bfz\|_{\sfH}}.$$
\end{definition}

We are now ready to define the discrete derivative operators.

\begin{definition}\label{def:derivative}
Consider a function $f:\Omega^{\calU,m}\rightarrow\mathbb{C}$. For a matching $S\in \calM_{\calU,\leq m}$ and a label $\bfx\in \Map{S}{\ZNk}$, we define a function $D_{S,\bfx}[f]:\Omega^{\calU,m}_{\setminus S}\rightarrow\mathbb{R}$ by 
\[
D_{S,\bfx}[f](\bfy):=\Exu{\bfz:S\rightarrow\ZNk}{H_{S}(\bfz)\cdot f\Big(\fraki(\bfy,\bfx-\bfz)\Big)},
\]
where the embedding $\fraki:\Omega^{\calU,m}_{\setminus S}\times \Map{S}{\ZNk}\rightarrow\Omega^{\calU,m}$ is as in \Cref{def:embedding}.
\end{definition}

The discrete derivative operators provide a means of measuring the ``globalness'' of a \emph{function} on $\Omega^{\calU,m}$, beyond the previous globalness notion (\Cref{def:global_set}) which is defined only for \emph{subsets} of $\Omega^{\calU,m}$. The following definition of derivative-based globalness parallels \cite[Definition 4.4]{KLM23} and \cite[Definition 4.8]{FMW25}.

\begin{definition}\label{def:derivative-based-global}
Let $r,\lambda>0$ and $1\leq p<\infty$. For a function $f:\Omega^{\calU,m}\rightarrow\mathbb{C}$, we say it is $(r,\lambda,d)$-$L^{p}$-global if for every matching $S\in\calM_{\calU,\leq d}$ and label $\bfx\in \Map{S}{\ZNk}$, we have $\left\|D_{S,\bfx}f\right\|_{p}\leq r^{|S|}\lambda$.
\end{definition}

The following proposition (similar to \cite[Lemma 4.9]{KLM23} and \cite[Proposition 4.10]{FMW25}) shows that the two notions of globalness are closely related: the globalness of a subset in the sense of \Cref{def:global_set} implies the derivative-based globalness of the indication function of the subset as in \Cref{def:derivative-based-global}. 

\begin{proposition}\label{prop:connecting-two-globalness}
Suppose a subset $A\subseteq \Omega^{\calU,m}$ is a global set (in the sense of \Cref{def:global_set}). Let $1_{A}:\Omega^{\calU,m}\rightarrow\{0,1\}$ be the indicator function of $A$. Then for every $1\leq p<\infty$, the function $1_{A}$ is $(4, \|1_{A}\|_{p}, m)$-$L^{p}$-global.
\end{proposition}

\begin{proof}
Consider an arbitrary matching $S\in\calM_{\calU,\leq m}$, and let $\fraki:\Omega^{\calU,m}_{\setminus S}\times \Map{S}{\ZNk}\hookrightarrow \Omega^{\calU,m}$ be the embedding defined in \Cref{def:embedding}. For any fixed $\bfz\in\Map{S}{\ZNk}$, by \Cref{def:global_set}, the function $1_{A}(\fraki(\cdot,\bfz)):\Omega^{\calU,m}_{\setminus S}\rightarrow \{0,1\}$ is the indicator function of a set of size at most $2^{|S|}\cdot|A| \cdot |\Omega^{\calU,m}_{\setminus S}|/|\Omega^{\calU,m}|$. As $1_{A}$ is Boolean-valued, we get 
$\left\|1_{A}(\fraki(\cdot,\bfz))\right\|_{p}^{p}\leq 2^{|S|}\cdot \left\|1_{A}\right\|_{p}^{p}$. Therefore, for any $\bfx\in \Map{S}{\ZNk}$, we have (using the Minkowski inequality)
\begin{align*}
\left\|D_{S,\bfx}[1_{A}]\right\|_{p}&=
\left\|\Exu{\bfz: S\rightarrow \ZNk}{H_{S}(\bfz)\cdot 1_{A}\Big(\mathfrak{i}(\cdot,\bfx-\bfz)\Big)}\right\|_{p}
\leq \Exu{\bfz: S\rightarrow \ZNk}{\left|H_{S}(\bfz)\right|\cdot\left\|1_{A}\Big(\fraki(\cdot,\bfx-\bfz)\Big)\right\|_{p}}\\
&\leq \Exu{\bfz: S\rightarrow \ZNk}{|H_{S}(\bfz)|}\cdot 2^{|S|/p}\cdot \|1_{A}\|_{p}\leq 4^{|S|}\cdot\|1_{A}\|_{p},
\end{align*}
where we used the simple calculation
\[\Exu{\bfz: S\rightarrow \ZNk}{|H_{S}(\bfz)|}=\left(\frac{1}{N^{k}}\cdot (N^{k}-1)+\frac{N^{k}-1}{N^{k}}\cdot 1\right)^{|S|}\leq 2^{|S|}.\qedhere\]
\end{proof}

\subsection{Level-$d$ Projection}\label{subsec:projections}

As suggested in its name, the ``projected level-$d$ inequality'' studies the projection of a function onto the linear subspace spanned by a collection of level-$d$ character functions. This is formalized in the following two definitions.

\begin{definition}\label{def:frakX}
For a nonnegative integer $d$, we denote by $\frakX^{\calU,d}$ the collection of pairs $(M,\bfa)$ where $M\in \calM_{\calU,d}$ and $\bfa\in \Map{M}{\ZNk\setminus \{0\}}$. We also write $\frakX^{\calU,\leq m}:=\bigcup_{0\leq d\leq m}\frakX^{\calU,d}$. 
\end{definition}

\begin{definition}\label{def:level-d-projection}
Define the operator $P_{\frakX}^{=d}:L^{2}(\Omega^{\calU,m})\rightarrow L^{2}(\Omega^{\calU,m})$ to be the orthogonal projection onto the linear subspace of the Hilbert space $L^{2}(\Omega^{\calU,m})$ spanned by the characters $\psi_{M,\bfa}$, where $(M,\bfa)$ ranges in $ \frakX^{\calU,d}$.
\end{definition}

Using the fact that $\{\psi_{M,\bfa}:M\in\calM_{U,d}\text{ and }\bfa\in \Map{M}{\ZNk\setminus\{0\}}\}$ forms an orthonormal set (see~\Cref{prop:ortho_set}), we have the following direct formula for projections.

\begin{proposition}\label{prop:formula-of-projection}
Given an integer $d\geq 0$, for each function $f:\Omega^{U,m}\rightarrow\mathbb{C}$ we have
$$P_{\frakX}^{=d}[f](\bfy):=\sum_{(M,\bfa)\in\frakX^{\calU,d}}\langle f,\psi_{M,\bfa}\rangle \cdot \psi_{M,\bfa}(\bfy).$$
\end{proposition}
\begin{proof}
    Follows immediately from~\Cref{prop:ortho_set}.
\end{proof}

We are now ready to state the projected level-$d$ inequality, which is proved in \Cref{app:global_hypercontractivity}:

\begin{restatable}[Projected level-$d$ inequality]{theorem}{projleveld}\label{thm:level-d-inequality}
Fix integers $d,m$ such that $|\calU|\geq 2km$ and $m\geq 2(d+1)$. Suppose $f:\Omega^{\calU,m}\rightarrow\bC$ is both $(r,\lambda_{1},d)$-$L^{1}$-global and $(r,\lambda_{2},d)$-$L^{2}$-global, where $d\leq \log(\lambda_{2}/\lambda_{1})$ and $r\geq 1$. Then
\begin{equation}\label{eq:projected-level-d}
\left\|P_{\frakX}^{=d}f\right\|_{2}^{2}\leq \lambda_{1}^{2}\left(\frac{10^{5}r^{2}\log(\lambda_{2}/\lambda_{1})}{d}\right)^{d}.
\end{equation}
\end{restatable}

We note that since the bound provided by~\Cref{thm:level-d-inequality} grows (exponentially) with the level~$d$, a projection onto the span of character functions across multiple levels can be roughly bounded by the level-$d$ bound corresponding to the highest level involved. We formalize this observation in the following corollary.

\begin{corollary}\label{cor:projected-level-d}
Under the same conditions as \Cref{thm:level-d-inequality}, for any real number $\ell\in [1,\log(\lambda_{2}/\lambda_{1})]$ we have
\[
\sum_{d=1}^{\lfloor \ell\rfloor}\left\|P^{=d}_{\frakX}[f]\right\|_{2}^{2}\leq \ell\cdot\lambda_{1}^{2}\left(\frac{10^{5}r^{2}\log(\lambda_{2}/\lambda_{1})}{\ell}\right)^{\ell}.
\]
\end{corollary}
\begin{proof}
It suffices to observe that the expression on the right hand side of \eqref{eq:projected-level-d} is monotone increasing in $d$ in the range $1\leq d\leq \log(\lambda_{2}/\lambda_{1})$, even when $d$ takes non-integral values.
\end{proof}

\subsection{Singular Value Decomposition}\label{subsec:SVD}

Having formalized the projected level-$d$ inequality on $\Omega^{\calU,m}$, we now proceed to the second step of this section: analyzing the operator $\bdP^{\calU,m}_{\mu}$. A key property of this operator is that it admits a clean singular value decomposition: it maps character functions on $\Omega^{\calU,m}$ to scalar multiples of character functions on $\ZmodN^{\bigcup \calU}$. We remark that this map is not 1-to-1, as distinct characters on $\Omega^{\calU,m}$ may be mapped to the same character on $\ZmodN^{\bigcup \calU}$.

Given a character index $b \in \ZmodN^{\bigcup \calU}$, the following definition identifies all characters on $\Omega^{\calU,m}$ that are mapped by $\bdP^{\calU,m}_{\mu}$ to nonzero scalar multiples of $\chi_b$.

\begin{definition}\label{def:character_group}
Given a character index \(b \in \mathbb{Z}_N^{\bigcup\calU}\), define \(\calX^{\circ}(b)\) to be the collection of character indices \((M, \bfa) \in \frakX^{\calU, m}\) satisfying:
\begin{enumerate}[label=(\arabic*)]
    \item For every \(e \in M\), we have \(b_{|e} = \bfa(e)\);
    \item For every \(e \in M\), the vector \(b_{|e} \in \mathbb{Z}_N^k\) is nonzero on at least two coordinates;
    \item If a vertex \(v \in \tcup\calU\) does not appear in any edge of \(M\), then \(b_v = 0\).
\end{enumerate}
If condition (2) is removed, the resulting (larger) collection is denoted \(\calX(b)\), omitting the superscript \(\circ\). Note that $\calX(b)$ is empty if $|\supp(b)|> km$. 
\end{definition}

The condition (2) in \Cref{def:character_group} is especially important since it captures the property that not too many characters on $\Omega^{\calU,m}$ are associated with a same character on $\ZmodN^{\calV\times [n]}$ under $\bdP^{\calU,m}_{\mu}$, as shown by the following lemma. 

\begin{lemma}[{\cite[Lemma 6.9]{CGSV24}}]\label{lem:CGSV_noise_rate}
For any character index $b\in \ZmodN^{\calU}$ with $|\supp(b)|=d$, if $|\calU|>100km$ then
\[
\sum_{(M,\bfa)\in \calX^{\circ}(b)}\Psi(|\calU|,m,|M|)\leq \left(\frac{100k^{3}md}{|\calU|^{2}}\right)^{d/2}.
\]
\end{lemma}

We are now ready to present the singular value decomposition lemma. In particular, in the proof of the lemma, we will show why the one-wise independence of $\mu$ corresponds to the condition (2) in \Cref{def:character_group}.

\begin{lemma}[Singular value decomposition]\label{lem:singular_value}
Fix a character index $b\in \ZmodN^{\bigcup\calU}$ and a one-wise independent distribution $\mu$ over $\ZNk$. There exist complex numbers $R(M,\bfa)$, each with absolute value at most 1, for all character indices $(M,\bfa)\in \calX^{\circ}(b)$, such that
\[
\Big\langle \bdP^{\calU,m}_{\mu}[f], \chi_b \Big\rangle_{L^2\big(\mathbb{Z}_N^{\bigcup\calU}\big)}
= \sum_{(M, \bfa) \in \calX^{\circ}(b)} R(M, \bfa)\cdot\left\langle f, \sqrt{\Psi(|\calU|,m,|M|)}\psi_{M,\bfa} \right\rangle
\]
holds for any function $f\in L^{2}(\Omega^{\calU,m})$.
\end{lemma}

\begin{proof}
Recall from Section~\ref{subsec:general_notations} that $\mu(\cdot)$ denotes the probability mass function of $\mu$. For each character index \(t \in \mathbb{Z}_N^k\), define
\[
r(t) := \sum_{z \in \mathbb{Z}_N^k} \mu(z)\, \overline{\chi_t(z)} = N^k \cdot \widehat{\mu}(t).
\]
Since \(\mu\) is assumed to be one-wise independent, we know that \(\widehat{\mu}(t) = 0\) for any \(t \in \mathbb{Z}_N^k\) with exactly one nonzero coordinate. Thus, \(r(t) = 0\) for such \(t\). Additionally, we have \(r(0) = 1\) and \(|r(t)| \leq 1\) for all \(t\) since $|r(t)|\leq  \mathbb{E}_{z \sim \mu}\, \left|\overline{\chi_t(z)}\right|=1$.

From \Cref{def:Markov_kernel}, we can express
\[
\bfP^{\calU,m}_{\mu}(x, \bfy) = \frac{1}{|\calM_{\calU, m}|} \prod_{e \in \supp(\bfy)} \mu(x_{|e} - \bfy(e)).
\]
It is easy to see that \(\left\langle \bfP^{\calU,m}_{\mu}(\cdot,\bfy),\chi_{b}\right\rangle_{L^{2}\big(\mathbb{Z}_N^{\bigcup\calU}\big)}\neq 0\) only if 
\begin{equation}\label{eq:calYb}
\text{any }v\in \supp(b)\text{ is contained in some edge }e\in \supp(\bfy).
\end{equation}
Let $\calY(b)$ denote the collection of $\bfy\in\Omega^{\calU,m}$ satisfying \eqref{eq:calYb}. We can then calculate:
\begin{align}
\Big\langle \bdP^{\calU,m}_{\mu}[f], \chi_b \Big\rangle_{L^2\big(\mathbb{Z}_N^{\bigcup\calU}\big)}
&=  \sum_{\bfy\in\Omega^{\calU,m}}f(\bfy)\left\langle\bfP^{\calU,m}_{\mu}(\cdot,\bfy),\chi_{b}\right\rangle_{L^{2}\big(\mathbb{Z}_N^{\bigcup\calU}\big)}\nonumber \\
&= \frac{1}{|\calM_{\calU, m}|} \sum_{\bfy\in \calY(b)}\left(f(\bfy) \prod_{e \in \supp(\bfy)} \Exu{z \in \mathbb{Z}_N^k} {\mu(z - \bfy(e))\, \overline{\chi_{b_{|e}}(z)} }\right) \nonumber \\
&= \frac{1}{|\calM_{\calU, m}|} \sum_{\bfy\in \calY(b)}\left( f(\bfy) \prod_{e \in \supp(\bfy)} \widehat{\mu}(b_{|e})\, \overline{\chi_{b_{|e}}(\bfy(e))}\right) \nonumber \\
&= \frac{1}{|\calM_{\calU, m}| \cdot (N^k)^{\alpha n}} \sum_{\bfy\in \calY(b)}\left( f(\bfy) \prod_{e \in \supp(\bfy)} r(b_{|e})\, \overline{\chi_{b_{|e}}(\bfy(e))}\right)\nonumber \\
&= \Exu{\bfy \in \Omega^{\calU, m}}{\ind{\bfy\in\calY(b)}\cdot f(\bfy) \prod_{e \in \supp(\bfy)} r(b_{|e})\, \overline{\chi_{b_{|e}}(\bfy(e))}}. \label{eq:SVD_x_side}
\end{align}

Now consider a fixed \(\bfy \in \Omega^{\calU,m}\). If $\bfy\not\in\calY(b)$, we have $\psi_{M,\bfa}(\bfy)=0$ for all $(M,\bfa)\in \calX(b)$. If $\bfy\in\calY(b)$, we have $\psi_{M,\bfa}(\bfy)\neq 0$ for exactly one $(M,\bfa)\in \calX(b)$. --- namely, the unique pair \((M,\bfa)\) with \(M \subseteq \supp(\bfy)\). Therefore, we have:
\begin{equation}\label{eq:SVD_y_side}
\sum_{(M,\bfa) \in \calX(b)} R(M,\bfa) \sqrt{\Psi(|\calU|,m,|M|)}\cdot\overline{\psi_{M,\bfa}(\bfy)}
= \ind{\bfy\in\calY(b)}\cdot \prod_{e \in \supp(\bfy)} r(b_{|e})\,\overline{\chi_{b_{|e}}(\bfy(e))},
\end{equation}
where \(R(M,\bfa) := \prod_{e \in M} r(\bfa(e))\). Since each $r(\bfa(e))$ has absolute value at most 1, so does $R(M,\bfa)$.

Comparing \eqref{eq:SVD_x_side} and \eqref{eq:SVD_y_side}, we conclude:
\[
\Big\langle \bdP^{\calU,m}_{\mu}[f], \chi_b \Big\rangle_{L^2\big(\mathbb{Z}_N^{\bigcup\calU}\big)}
= \sum_{(M, \bfa) \in \calX(b)} R(M, \bfa)\cdot\left\langle f, \sqrt{\Psi(|\calU|,m,|M|)}\psi_{M,\bfa} \right\rangle.
\]
Finally, observe that any \((M, \bfa) \in \calX(b)\) with \(R(M, \bfa) \neq 0\) must satisfy condition (2) of \Cref{def:character_group}, and hence lies in \(\calX^{\circ}(b)\), as desired.
\end{proof}

The following corollary summarizes what we have revealed about the singular value decomposition of the operator $\bdP^{\calU,m}_{\mu}$.

\begin{corollary}\label{cor:singular_value}
Fix a one-wise independent distribution \(\mu\) over \(\mathbb{Z}_N^k\). For any character index \(b \in \mathbb{Z}_N^{\bigcup \calU}\), the adjoint operator \(\left(\bdP^{\calU,m}_{\mu}\right)^{\dagger} : L^2\big(\mathbb{Z}_N^{\bigcup \calU}\big) \rightarrow L^2(\Omega^{\calU,m})\) maps the character function \(\chi_b \in L^2\big(\mathbb{Z}_N^{\bigcup \calU}\big)\) to a function \(\widetilde{\chi}_b \in L^2(\Omega^{\calU,m})\) satisfying the following:
\begin{enumerate}[label=(\arabic*)]
    \item If \(|\supp(b)| = \ell\), then \(\|\widetilde{\chi}_b\|_2^2 \leq \left(100k^3 m \ell\, |\calU|^{-2}\right)^{\ell/2}\).
    
    \item If \(|\supp(b)| = \ell\), then \(\widetilde{\chi}_b\) lies in the linear subspace of \(L^2(\Omega^{\calU,m})\) spanned by the character functions \(\psi_{M,\bfa}\) for \((M,\bfa) \in \bigcup_{d=1}^{\lfloor \ell/2\rfloor}\frakX^{\calU, d}\).

    \item For distinct indices \(b, b' \in \mathbb{Z}_N^{\bigcup \calU}\), the functions \(\widetilde{\chi}_b\) and \(\widetilde{\chi}_{b'}\) are orthogonal.
\end{enumerate}
\end{corollary}

\begin{proof}
For $(M,\bfa)\in \calX^{\circ}(b)$, let $R(M,\bfa)$ be the complex numbers from \Cref{lem:singular_value}. The conclusion of \Cref{lem:singular_value} implies that
\begin{equation}\label{eq:transposed_operator_on_characters}
\left(\bdP^{\calU,m}_{\mu}\right)^{\dagger}[\chi_{b}]=\sum_{(M,\bfa)\in \calX^{\circ}(b)}\overline{R(M,\bfa)}\cdot \sqrt{\Psi(|\calU|,m,|M|)}\cdot \psi_{M,\bfa}.
\end{equation}
The three statements in the corollary can then be easily deduced.

For statement (1): since each $R(M,\bfa)$ has absolute value at most 1, it follows that 
\[
\left\|\left(\bdP^{\calU,m}_{\mu}\right)^{\dagger}[\chi_{b}]\right\|_{2}^{2}\leq \sum_{(M,\bfa)\in\calX^{\circ}}\Psi(|\calU|,m,|M|)\leq \left(\frac{100k^{3}m\ell}{|\calU|^{2}}\right)^{\ell/2},
\]
due to \Cref{lem:CGSV_noise_rate}.

For statement (2): condition (2) of \Cref{def:character_group}, for any $(M,\bfa)\in \calX^{\circ}(b)$, we have $|M|\leq |\supp(b)|/2=\ell/2$, and hence the statement follows from \eqref{eq:transposed_operator_on_characters}.

For statement (3): conditions (1) and (3) of \Cref{def:character_group} ensures that $\calX(b)$ and $\calX(b')$ are disjoint. The orthogonality then follows from \eqref{eq:transposed_operator_on_characters} and \Cref{prop:ortho_set}.
\end{proof}

\subsection{Proof of \Cref{lem:global-decay}}\label{subsec:proof-of-decay}

We now prove \Cref{lem:global-decay}, restated below.
\globalimpliesdecay*

\begin{proof}
For any character index $b\in \ZmodN^{\bigcup\calU}$, we define the function \(\widetilde{\chi}_{b}=\left(\bdP^{\calU,m}_{\mu}\right)^{\dagger}[\chi_{b}]\) as in \Cref{cor:singular_value}. For any positive integer $\ell$, we have
\begin{align*}
\left\|\bdP^{\calU,m}_{\mu}[\varphi_A]^{=\ell}\right\|_{2}^{2}&=\sum_{\substack{b \in \mathbb{Z}_N^{\bigcup \calU} \\ |\supp(b)| = \ell}} 
\left| \left\langle \bdP^{\calU,m}_{\mu}[\varphi_A], \chi_b \right\rangle \right|^2
= 
\sum_{\substack{b \in \mathbb{Z}_N^{\bigcup \calU} \\ |\supp(b)| = \ell}} 
\left| \left\langle \varphi_A, \widetilde{\chi}_b \right\rangle \right|^2 \\
&\leq 
\left( \frac{100k^3 m \ell}{|\calU|^2} \right)^{\ell/2}
\sum_{\substack{b \in \mathbb{Z}_N^{\bigcup \calU} \\ |\supp(b)| = \ell}} 
\left| \left\langle \varphi_A, \frac{\widetilde{\chi}_b}{\| \widetilde{\chi}_b \|_2} \right\rangle \right|^2
\tag*{(by \Cref{cor:singular_value}(1))} \\
&\leq 
\left( \frac{100k^3 m \ell}{|\calU|^2} \right)^{\ell/2}
\sum_{d=1}^{\lfloor \ell/2\rfloor}\left\| P^{=d}_{\frakX}[\varphi_A] \right\|_2^2
\tag*{(by \Cref{cor:singular_value}(2)(3))}.
\end{align*}
Since $A$ is a global set in $\Omega^{\calU,m}$ of size $2^{-w}\cdot\left|\Omega^{\calU,m}\right|$, it follows from \Cref{prop:connecting-two-globalness} that $\varphi_{A}=2^{w}\cdot 1_{A}$ is both $(4,1,m)$-$L^{1}$-global and $(4,2^{w/2},m)$-$L^{2}$-global. If $\ell\leq w$, we can apply \Cref{cor:projected-level-d} to the final line of the above display and get
\[
\left\|\bdP^{\calU,m}_{\mu}[\varphi_A]^{=\ell}\right\|_{2}^{2}\leq \left(\frac{100k^{3}m\ell}{|\calU|^{2}}\right)^{\ell/2}\cdot \frac{\ell}{2}\left(\frac{10^{5}\cdot 4\cdot(w/2)}{\ell/2}\right)^{\ell/2}= \frac{\ell}{2}\left(\frac{4\cdot 10^{7}k^{3}mw}{|\calU|^{2}}\right)^{\ell/2}\leq \left(\frac{w}{|\calU|}\right)^{\ell/2},
\]
since $|\calU|\geq 10^{8}k^{3}m$ and $\frac{\ell}{2}\leq 2^{\ell/2}$. If $\ell>w/2$, we note that \(\sum_{d=1}^{\lfloor \ell/2\rfloor}\left\| P^{=d}_{\frakX}[\varphi_A]\right\|_2^2\leq \|\varphi_{A}\|_{2}^{2}\). Therefore, we have
\[
\left\|\bdP^{\calU,m}_{\mu}[\varphi_A]^{=\ell}\right\|_{2}^{2}\leq \left(\frac{100k^{3}m\ell}{|\calU|^{2}}\right)^{-\ell/2}\cdot 2^{w}\leq 2^{w}\cdot \left(\frac{\ell}{4|\calU|}\right)^{\ell/2},
\]
since $|\calU|\geq 10^{8}k^{3}m$. Combining the above two displays, we conclude for $1\leq \ell\leq km$ that
\[\left\|\bdP^{\calU,m}_{\mu}[\varphi_A]^{=\ell}\right\|_{2}^{2}\leq F(|\calU|,\ell,w/2),\]
and thus the proof is complete.
\end{proof}

\bibliography{reference}

@inproceedings{Yos11,
author = {Yoshida, Yuichi},
title = {Optimal constant-time approximation algorithms and (unconditional) inapproximability results for every bounded-degree CSP},
year = {2011},
isbn = {9781450306911},
publisher = {Association for Computing Machinery},
address = {New York, NY, USA},
url = {https://doi.org/10.1145/1993636.1993725},
doi = {10.1145/1993636.1993725},
booktitle = {Proceedings of the Forty-Third Annual ACM Symposium on Theory of Computing},
pages = {665–674},
numpages = {10},
location = {San Jose, California, USA},
series = {STOC '11}
}

@article{CGSV24,
author = {Chou, Chi-Ning and Golovnev, Alexander and Sudan, Madhu and Velusamy, Santhoshini},
title = {Sketching Approximability of All Finite CSPs},
year = {2024},
issue_date = {April 2024},
publisher = {Association for Computing Machinery},
address = {New York, NY, USA},
volume = {71},
number = {2},
issn = {0004-5411},
url = {https://doi.org/10.1145/3649435},
doi = {10.1145/3649435},
journal = {J. ACM},
month = apr,
articleno = {15},
numpages = {74}
}

@inproceedings{Tre01,
author = {Trevisan, Luca},
title = {Non-approximability results for optimization problems on bounded degree instances},
year = {2001},
isbn = {1581133499},
publisher = {Association for Computing Machinery},
address = {New York, NY, USA},
url = {https://doi.org/10.1145/380752.380839},
doi = {10.1145/380752.380839},
booktitle = {Proceedings of the Thirty-Third Annual ACM Symposium on Theory of Computing},
pages = {453–461},
numpages = {9},
location = {Hersonissos, Greece},
series = {STOC '01}
}

@article{KLM23,
  title={Sharp hypercontractivity for global functions},
  author={Keller, Nathan and Lifshitz, Noam and Marcus, Omri},
  journal={arXiv preprint arXiv:2307.01356},
  year={2023} 
}

@article{FederVardi,
  author       = {Tom\'as Feder and Moshe Y. Vardi},
  title        = {The Complexity of Constraint Satisfaction Problems},
  journal      = {Proceedings of the 25th Annual ACM Symposium on Theory of Computing (STOC)},
  year         = {1993},
  note         = {Posits the CSP Dichotomy Conjecture},
}

@article{Zhuk,
  author       = {Dmitriy Zhuk},
  title        = {A Proof of the CSP Dichotomy Conjecture},
  journal      = {Journal of the ACM},
  volume       = {67},
  number       = {5},
  pages        = {30:1--30:78},
  year         = {2020},
  doi          = {10.1145/3402029},
}

@inproceedings{Bulatov,
  author       = {Andrei A. Bulatov},
  title        = {A Dichotomy Theorem for Nonuniform CSPs},
  booktitle    = {Proceedings of the 58th IEEE Symposium on Foundations of Computer Science (FOCS)},
  year         = {2017}
}

@inproceedings{KMW06,
author = {Kuhn, Fabian and Moscibroda, Thomas and Wattenhofer, Roger},
title = {The price of being near-sighted},
year = {2006},
isbn = {0898716055},
publisher = {Society for Industrial and Applied Mathematics},
address = {USA},
booktitle = {Proceedings of the Seventeenth Annual ACM-SIAM Symposium on Discrete Algorithm},
pages = {980–989},
numpages = {10},
location = {Miami, Florida},
series = {SODA '06}
}

@inproceedings{Tre96,
author = {Trevisan, Luca},
title = {Positive Linear Programming, Parallel Approximation and {PCP}s},
year = {1996},
isbn = {3540616802},
publisher = {Springer-Verlag},
address = {Berlin, Heidelberg},
booktitle = {Proceedings of the Fourth Annual European Symposium on Algorithms},
pages = {62–75},
numpages = {14},
series = {ESA '96}
}

@article{fotakis1997linear,
  title={Linear programming and fast parallel approximability},
  author={Fotakis, Dimitris A and Spirakis, Paul G},
  journal={Unpublished manuscript},
  year={1997}
}

@inproceedings{KK19,
author = {Kapralov, Michael and Krachun, Dmitry},
title = {An optimal space lower bound for approximating MAX-CUT},
year = {2019},
isbn = {9781450367059},
publisher = {Association for Computing Machinery},
address = {New York, NY, USA},
url = {https://doi.org/10.1145/3313276.3316364},
doi = {10.1145/3313276.3316364},
booktitle = {Proceedings of the 51st Annual ACM SIGACT Symposium on Theory of Computing},
pages = {277–288},
numpages = {12},
location = {Phoenix, AZ, USA},
series = {STOC 2019}
}

@article{FMW25,
  title={Multi-Pass Streaming Lower Bounds for Approximating Max-Cut},
  author={Fei, Yumou and Minzer, Dor and Wang, Shuo},
  journal={arXiv preprint arXiv:2503.23404},
  year={2025}
}

@article{laurent2012semidefinite,
  title={Semidefinite optimization},
  author={Laurent, Monique and Vallentin, Frank},
  journal={Lecture Notes, available at http://page. mi. fu-berlin. de/fmario/sdp/laurentv. pdf},
  year={2012}
}

@misc{odonnell2021analysisbooleanfunctions,
      title={Analysis of Boolean Functions}, 
      author={Ryan O'Donnell},
      year={2021},
      eprint={2105.10386},
      archivePrefix={arXiv},
      primaryClass={cs.DM},
      url={https://arxiv.org/abs/2105.10386}, 
}

@book{mohri2018foundations,
  title={Foundations of machine learning},
  author={Mohri, Mehryar and Rostamizadeh, Afshin and Talwalkar, Ameet},
  year={2018},
  publisher={MIT press}
}

@book{peleg2000distributed,
  title={Distributed computing: a locality-sensitive approach},
  author={Peleg, David},
  year={2000},
  publisher={SIAM}
}

@inproceedings{Rag08,
author = {Raghavendra, Prasad},
title = {Optimal algorithms and inapproximability results for every CSP?},
year = {2008},
isbn = {9781605580470},
publisher = {Association for Computing Machinery},
address = {New York, NY, USA},
url = {https://doi.org/10.1145/1374376.1374414},
doi = {10.1145/1374376.1374414},
booktitle = {Proceedings of the Fortieth Annual ACM Symposium on Theory of Computing},
pages = {245–254},
numpages = {10},
location = {Victoria, British Columbia, Canada},
series = {STOC '08}
}

@inproceedings{KKS14,
  title={Streaming lower bounds for approximating MAX-CUT},
  author={Kapralov, Michael and Khanna, Sanjeev and Sudan, Madhu},
  booktitle={Proceedings of the Twenty-Sixth Annual ACM-SIAM Symposium on Discrete Algorithms},
  pages={1263--1282},
  year={2014},
  organization={SIAM}
}

@inproceedings{raz1997separation,
  title={Separation of the monotone NC hierarchy},
  author={Raz, Ran and McKenzie, Pierre},
  booktitle={Proceedings 38th Annual Symposium on Foundations of Computer Science},
  pages={234--243},
  year={1997},
  organization={IEEE}
}

@inproceedings{goos2017query,
  title={Query-to-communication lifting for BPP},
  author={G{\"o}{\"o}s, Mika and Pitassi, Toniann and Watson, Thomas},
  booktitle={2017 IEEE 58th Annual Symposium on Foundations of Computer Science (FOCS)},
  pages={132--143},
  year={2017},
  organization={IEEE}
}

@inproceedings{chou2020optimal,
  title={Optimal streaming approximations for all boolean max-2csps and max-ksat},
  author={Chou, Chi-Ning and Golovnev, Alexander and Velusamy, Santhoshini},
  booktitle={2020 IEEE 61st Annual Symposium on Foundations of Computer Science (FOCS)},
  pages={330--341},
  year={2020},
  organization={IEEE}
}

@inproceedings{saxena2025streaming,
  title={Streaming Algorithms via Local Algorithms for Maximum Directed Cut},
  author={Saxena, Raghuvansh R and Singer, Noah G and Sudan, Madhu and Velusamy, Santhoshini},
  booktitle={Proceedings of the 2025 Annual ACM-SIAM Symposium on Discrete Algorithms (SODA)},
  pages={3392--3408},
  year={2025},
  organization={SIAM}
}

@inproceedings{KKS15,
author = {Kapralov, Michael and Khanna, Sanjeev and Sudan, Madhu},
title = {Streaming lower bounds for approximating MAX-CUT},
year = {2015},
publisher = {Society for Industrial and Applied Mathematics},
address = {USA},
booktitle = {Proceedings of the Twenty-Sixth Annual ACM-SIAM Symposium on Discrete Algorithms},
pages = {1263–1282},
numpages = {20},
location = {San Diego, California},
series = {SODA '15}
}

@article{SudanSurvey,
  title={Streaming and Sketching Complexity of CSPs: A survey},
  author={Sudan, Madhu},
  journal={arXiv preprint arXiv:2205.02744},
  year={2022}
}

@article{hwang2024oblivious,
  title={Oblivious Algorithms for Maximum Directed Cut: New Upper and Lower Bounds},
  author={Hwang, Samuel and Singer, Noah G and Velusamy, Santhoshini},
  journal={arXiv preprint arXiv:2411.12976},
  year={2024}
}

@inproceedings{chou2022linear,
  title={Linear space streaming lower bounds for approximating {CSP}s},
  author={Chou, Chi-Ning and Golovnev, Alexander and Sudan, Madhu and Velingker, Ameya and Velusamy, Santhoshini},
  booktitle={Proceedings of the 54th Annual ACM SIGACT Symposium on Theory of Computing},
  pages={275--288},
  year={2022}
}

@inproceedings{saxena2023improved,
  title={Improved streaming algorithms for Maximum Directed Cut via smoothed snapshots},
  author={Saxena, Raghuvansh R and Singer, Noah G and Sudan, Madhu and Velusamy, Santhoshini},
  booktitle={2023 IEEE 64th Annual Symposium on Foundations of Computer Science (FOCS)},
  pages={855--870},
  year={2023},
  organization={IEEE}
}

@inproceedings{guruswami2017streaming,
  title={Streaming complexity of approximating max 2csp and max acyclic subgraph},
  author={Guruswami, Venkatesan and Velingker, Ameya and Velusamy, Santhoshini},
  booktitle={Approximation, Randomization, and Combinatorial Optimization. Algorithms and Techniques (APPROX/RANDOM 2017)},
  pages={8--1},
  year={2017},
  organization={Schloss Dagstuhl--Leibniz-Zentrum f{\"u}r Informatik}
}

@inproceedings{Khot02,
  author       = {Subhash Khot},
  title        = {On the Power of Unique 2-Prover 1-Round Games},
  booktitle    = {Proceedings of the 17th Annual {IEEE} Conference on Computational
                  Complexity, Montr{\'{e}}al, Qu{\'{e}}bec, Canada, May 21-24,
                  2002},
  pages        = {25},
  publisher    = {{IEEE} Computer Society},
  year         = {2002},
  url          = {https://doi.org/10.1109/CCC.2002.1004334},
  doi          = {10.1109/CCC.2002.1004334},
  bibsource    = {dblp computer science bibliography, https://dblp.org}
}

@INPROCEEDINGS{AKSY20,
  author={Assadi, Sepehr and Kol, Gillat and Saxena, Raghuvansh R. and Yu, Huacheng},
  booktitle={2020 IEEE 61st Annual Symposium on Foundations of Computer Science (FOCS)}, 
  title={Multi-Pass Graph Streaming Lower Bounds for Cycle Counting, MAX-CUT, Matching Size, and Other Problems}, 
  year={2020},
  volume={},
  number={},
  pages={354-364},
  doi={10.1109/FOCS46700.2020.00041}}

@inproceedings{AN21,
author = {Assadi, Sepehr and N, Vishvajeet},
title = {Graph streaming lower bounds for parameter estimation and property testing via a streaming XOR lemma},
year = {2021},
isbn = {9781450380539},
publisher = {Association for Computing Machinery},
address = {New York, NY, USA},
url = {https://doi.org/10.1145/3406325.3451110},
doi = {10.1145/3406325.3451110},
booktitle = {Proceedings of the 53rd Annual ACM SIGACT Symposium on Theory of Computing},
pages = {612–625},
numpages = {14},
location = {Virtual, Italy},
series = {STOC 2021}
}

@article{Assadi,
author = {Assadi, Sepehr},
title = {Recent Advances in Multi-Pass Graph Streaming Lower Bounds},
year = {2023},
issue_date = {September 2023},
publisher = {Association for Computing Machinery},
address = {New York, NY, USA},
volume = {54},
number = {3},
issn = {0163-5700},
url = {https://doi.org/10.1145/3623800.3623808},
doi = {10.1145/3623800.3623808},
journal = {SIGACT News},
month = sep,
pages = {48–75},
numpages = {28}
}

@inproceedings{goldreich1997property,
  title={Property testing in bounded degree graphs},
  author={Goldreich, Oded and Ron, Dana},
  booktitle={Proceedings of the twenty-ninth annual ACM symposium on Theory of computing},
  pages={406--415},
  year={1997}
}

@article{parnas2002testing,
  title={Testing the diameter of graphs},
  author={Parnas, Michal and Ron, Dana},
  journal={Random Structures \& Algorithms},
  volume={20},
  number={2},
  pages={165--183},
  year={2002},
  publisher={Wiley Online Library}
}

@article{kaufman2004tight,
  title={Tight bounds for testing bipartiteness in general graphs},
  author={Kaufman, Tali and Krivelevich, Michael and Ron, Dana},
  journal={SIAM Journal on computing},
  volume={33},
  number={6},
  pages={1441--1483},
  year={2004},
  publisher={SIAM}
}

@article{blais2012property,
  title={Property testing lower bounds via communication complexity},
  author={Blais, Eric and Brody, Joshua and Matulef, Kevin},
  journal={computational complexity},
  volume={21},
  number={2},
  pages={311--358},
  year={2012},
  publisher={Springer}
}

@article{parnas2007approximating,
  title={Approximating the minimum vertex cover in sublinear time and a connection to distributed algorithms},
  author={Parnas, Michal and Ron, Dana},
  journal={Theoretical Computer Science},
  volume={381},
  number={1-3},
  pages={183--196},
  year={2007},
  publisher={Elsevier}
}

@inproceedings{behnezhad2024approximating,
  title={Approximating maximum matching requires almost quadratic time},
  author={Behnezhad, Soheil and Roghani, Mohammad and Rubinstein, Aviad},
  booktitle={Proceedings of the 56th Annual ACM Symposium on Theory of Computing},
  pages={444--454},
  year={2024}
}

@article{mahabadi20250,
  title={A 0.51-Approximation of Maximum Matching in Sublinear $n^{1.5}$ Time},
  author={Mahabadi, Sepideh and Roghani, Mohammad and Tarnawski, Jakub},
  journal={arXiv preprint arXiv:2506.01669},
  year={2025}
}

@inproceedings{bury2015sublinear,
  title={Sublinear estimation of weighted matchings in dynamic data streams},
  author={Bury, Marc and Schwiegelshohn, Chris},
  booktitle={Algorithms-ESA 2015: 23rd Annual European Symposium, Patras, Greece, September 14-16, 2015, Proceedings},
  pages={263--274},
  year={2015},
  organization={Springer}
}
\bibliographystyle{alpha}

\appendix

\section{The Decomposition Lemma}\label{app:regularity_decomposition}
In this appendix, we prove~\Cref{lem:regularity_decomposition}. The proof is a rather straightforward (sometimes verbatim) adaptation of \cite[Section 2]{FMW25}. The main difference from \cite{FMW25} is that here we refine the analysis so that the decomposition applies to protocols with communication complexity as large as $\Theta(\sqrt{n})$, rather than just $\Theta(n^{1/3})$. Nevertheless, this added strength is not needed for the current paper --- we still only invoke \Cref{lem:regularity_decomposition} for protocols of complexity $O(n^{1/3})$, as in \cite{FMW25}.

\subsection{The Set Decomposition Lemma}
The first step towards the decomposition of communication protocols is the decomposition of sets of possible inputs to a single player. Recall from \Cref{def:communication_game} that the input to a player is chosen in a labeled matching space of the form $\Omega^{\calU,m}$, where $\calU$ is a $k$-universe and $m\leq|\calU|$ is a nonnegative integer. We develop the following set decomposition lemma that applies to a large subset $A\subseteq \Omega^{\calU, m}$, analogous to~\cite[Lemma 2.6]{FMW25}. 

Recall from \Cref{def:restrictions} that restrictions on $\Omega^{\calU,m}$ are given by elements of $\Omega^{\calU,\leq m}$.

\begin{lemma}[Set decomposition]\label{lem:reg}
    Fix a $k$-universe $\calU$ and a nonnegative integer $m\leq |\calU|$. Let $\bfz' \in \Omega^{\calU,\leq m}$ be a restriction on $\Omega^{\calU,m}$, and take any subset $A\subseteq \Omega^{\calU,m}_{\bfz'}$. Then we can decompose $A$ into a disjoint union of subsets $A_{(1)},A_{(2)},\dots,A_{(\ell )}$ such that:
    \begin{enumerate}[label=(\arabic*)]
        \item {\bf Globalness:} for each $i\in [\ell]$, there exists a restriction $\bfz_{(i)}$ that subsumes $\bfz'$ such that $A_{(i)}\subseteq \Omega_{\bfz_{(i)}}$ and $A_{(\ell)}$ is $\bfz_{(\ell)}$-global.
        \item {\bf Size of the restrictions}: the restrictions $\bfz_{(i)}$ satisfy the following inequality:
        \[\sum_{i=1}^\ell \frac{\left|A_{(i)}\right|}{|A|} \left(\left|\supp(\bfz_{(i)})\right| +\log_{2} \left(\frac{\left|\Omega_{\bfz_{(i)}}^{\calU, m}\right|}{\left|A_{(i)}\right|}\right)\right)\leq \left|\supp(\bfz')\right|+\log_{2}\left(\frac{\left|\Omega_{\bfz'}^{\calU, m }\right|}{|A|}\right)+2.\]
    \end{enumerate}
\end{lemma}

\begin{proof}
The proof of the lemma is algorithmic, and the decomposition algorithm is given in \Cref{alg:decomposition}.
\begin{algorithm}
\DontPrintSemicolon
\SetKwInOut{Input}{Input}\SetKwInOut{Output}{Output}
    \caption{Decompose$(A,\bfz')$}\label{alg:decomposition}
    \Input{a restriction $\bfz'$ and a set $A\subseteq \Omega_{\bfz'}^{\calU, m }$}
    \Output{a sequence of sets $A_{(1)},\dots,A_{(\ell )}$ and a sequence of restrictions $z_{(1)},\dots,z_{(\ell )}$}
    $i\gets 0$\;
    \While {$A$ is not $\bfz'$-global}{
        $i \gets i+1$\;
        find a restriction $\bfz_{(i)}$ with largest possible support size such that $\bfz_{(i)}$ subsumes $\bfz'$ and 
        \begin{gather}\label{eq:restriction-maximality}
        \frac{\left|A\cap \Omega_{\bfz_{(i)}}^{\calU,m}\right|}{\left|\Omega_{\bfz_{(i)}}^{\calU,m}\right|}> 2^{\left|\supp(\bfz_{(i)})\right|-|\supp(\bfz')|}\cdot\frac{|A|}{\left|\Omega_{\bfz'}^{\calU,m}\right|}
        \end{gather}
        
        $A_{(i)}\gets A\cap\Omega_{\bfz_{(i)}}^{\calU,m }$\;
        $A\gets A\setminus A_{(i)}$\; 
    }
    \If{$A$ is nonempty (and $\bfz'$-global)}{
        $i \gets i+1$\;
        $A_{(i+1)} \gets A$\; 
        $z_{(i+1)}\gets \bfz'$\;
    }
    $\ell \leftarrow i$\; 
\end{algorithm}

To check the first property, assume on the contrary that some $A_{(i)}$ is not $\bfz_{(i)}$-global. Then by \Cref{def:global_set}, there must exists a restriction $\bfz$ that subsumes $\bfz_{(i)}$ and 
\begin{equation}\label{eq:assume-not-global}
\frac{\left|A_{(i)}\cap \Omega_{\bfz}^{\calU,m}\right|}{\left|\Omega_{\bfz}^{\calU,m}\right|}>2^{|\supp(\bfz)|-\left|\supp(\bfz_{(i)})\right|}\cdot\frac{\left|A_{(i)}\cap \Omega_{\bfz_{(i)}}^{\calU,m}\right|}{\left|\Omega_{\bfz_{(i)}}^{\calU,m}\right|}.
\end{equation}
Since the choice of $\bfz_{(i)}$ satisfies \eqref{eq:restriction-maximality} and $A_{(i)}=A\cap \Omega_{\bfz_{(i)}}$, we can combine \eqref{eq:assume-not-global} and \eqref{eq:restriction-maximality} and get
\[
\frac{\left|A\cap \Omega_{\bfz}^{\calU,m}\right|}{\left|\Omega_{\bfz}^{\calU,m}\right|}>2^{|\supp(\bfz)|-\left|\supp(\bfz')\right|}\cdot\frac{\left|A\cap \Omega_{\bfz'}^{\calU,m}\right|}{\left|\Omega_{\bfz'}^{\calU,m }\right|}.
\]
This contradicts the maximality of the choice of $\bfz_{(i)}$ since $|\supp(\bfz)|>\left|\supp(\bfz_{(i)})\right|$.

We now check the second property, and towards that end we denote $A_{(\geq i)} = \bigcup_{j=i}^\ell A_{(j)}$. 
Taking logs of~\eqref{eq:restriction-maximality} gives 
\[
\log_{2} \frac{\left|\Omega^{\calU,m}_{\bfz_{(i)}}\right|}{\left|A_{(i)}\right|} + \left| \supp(\bfz_{(i)})\right| \leq \log_{2} \frac{\left|\Omega^{\calU,m}_{\bfz'}\right|}{\left|A_{(\geq i)}\right|} + |\supp(\bfz')| = \log_{2} \frac{\left|\Omega^{\calU,m}_{\bfz'}\right|}{|A|} + \log_{2} \frac{|A|}{\left|A_{(\geq i)}\right|}+ |\supp(\bfz')|. 
\]
Multiplying this inequality by $|A_{(i)}|/|A|$ and summing over all $i\in [\ell]$ we get: 
\begin{align*}
    \sum_{i=1}^\ell \frac{|A_{(i)}|}{|A|} \left(|\supp(z_{(i)})|+\log_{2} \frac{\left|\Omega^{\calU,m}_{\bfz_{(i)}}\right|}{|A_{(i)}|}\right)&\leq \log_{2}\frac{\left|\Omega^{\calU,m}_{\bfz'}\right|}{|A|}+|\supp(\bfz')|+\sum_{i=1}^\ell \frac{|A_{(i)}|}{|A|}\cdot \log_{2} \frac{|A|}{|A_{(\geq i)}|}\\ 
    &\leq \log_{2}\frac{\left|\Omega^{\calU,m}_{\bfz'}\right|}{|A|}+|\supp(\bfz')|+\int_0^1 \log_{2} \frac{1}{1-x} \, \mathrm{d} x\\
    &\leq \log_{2}\frac{\left|\Omega^{\calU,m}_{\bfz'}\right|}{|A|} +|\supp(\bfz')|+2,
\end{align*}
where the second transition is because $\sum_{i=1}^\ell \frac{|A_{(i)}|}{|A|}\cdot \log_{2} \frac{|A|}{|A_{(\geq i)}|}$ is a lower Riemann sum for
the function $f(x) = \log_{2}\left(1/(1-x)\right)$ with points $x_{(i)} = \frac{|A_{(1)}|}{|A|}+\dots+\frac{|A_{(i-1)}|}{|A|}$, and the last transition is by a direct calculation.
\end{proof}

\subsection{From Arbitrary Protocols to Global Protocols}

Now we shift our focus from subsets of a single labeled matching space $\Omega^{\calU,m}$ to communication protocols for the game $\DIHP(G,n,\alpha,K)$ (defined in \Cref{def:communication_game}).
For the remainder of this appendix, we fix a distribution-labeled $k$-graph $G=(\calV,\calE,N,(\mu_{\sfe})_{\sfe\in\calE})$ along with an integer $K > 0$ and a parameter $\alpha > 0$. 

We will show how to use~\Cref{lem:reg} to transform any protocol into a \emph{global protocol}. 
Before doing so, we must formally define what we mean by a “global protocol” and how we measure its cost.
A global protocol proceeds in discrete communication rounds, with exactly one player speaking in each round.
The globalness requirement demands that, at the end of every round, the set of inputs consistent with the transcript so far is global.
To quantify the amount of information revealed during the communication process, we introduce the following potential function for structured rectangles (as defined in \Cref{def:structured_rectangle}).

\begin{definition}
    For restrictions $\boldsymbol{\zeta}=\left(\bfz^{(\sfe,j)}\right)_{(\sfe,j)\in \calE\times [K]}$ and a rectangle $R= \prod_{(\sfe, j) \in \calE\times [K]}A^{(\sfe, j)}$ such that $A^{(\sfe, j)}\subseteq \Omega^{\calU_\sfe,m}_{\bfz^{(\sfe,j)}}$,  we define the potential of $(\boldsymbol{\zeta},R)$ as: 
    \begin{align*}
    \phi(\boldsymbol{\zeta},R):= \sum_{(\sfe,j)\in \calE\times [K]} \left|\supp(\bfz^{(\sfe,j)})\right|+ \log_{2} \left(\frac{\left|\Omega^{\calU,m}_{\bfz^{(i)}}\right|}{|A^{(i)}|}
    \right). 
    \end{align*}
\end{definition}

The formal definition of global protocols is given as follows.
\begin{definition}\label{def:global_protocol}
    A communication protocol $\Pi$ for $\textsf{DIHP}(G,n,\alpha,K)$ is called an $r$-round global communication protocol if it specifies the following procedure of communications: 
    \begin{itemize}
        \item the $K|\calE|$ players take turns to send messages according $\Pi$; 
        \item there are at most $r$ rounds of communications, and there is only one player sending message in a single round;
        \item the length of message in each round of communications is not bounded; instead, from the perspective of rectangles, after each round of communications, a $\boldsymbol{\zeta}$-global rectangle $R$ is further partitioned into a disjoint union of rectangles $R_{(1)},\dots,R_{(\ell)}$ such that: (1) $R_{(i)}$ is $\boldsymbol{\zeta}_{(i)}$-global; (2) $\boldsymbol{\zeta}_{(i)}$ subsumes $\boldsymbol{\zeta}$; (3) the following inequality holds: 
        \begin{align*}
            \sum_{i=1}^\ell\frac{|R_{(i)}|}{|R|}\phi(\boldsymbol{\zeta}_{(i)},R_{(i)})\leq \phi(\boldsymbol{\zeta},R) + 3. 
        \end{align*}
    \end{itemize}
\end{definition}
Note that in the above definitions of global protocols, we measure a protocol by its average potential but not its communication cost. 

We now show an explicit construction of a global protocol $\Pi^{\reff}$ given any communication protocol $\Pi$, and show that $\adv(\Pi^{\reff})\geq \adv(\Pi)$. 
\begin{lemma}\label{lem:arbitrary_to_global}
    Given a communication protocol $\Pi$ for $\mathsf{DIHP}(G,n,\alpha,K)$ with communication complexity at most $r$, we can construct an $r$-round global protocol $\Pi^{\reff}$ for $\mathsf{DIHP}(G,n,\alpha,K)$ such that for any leaf rectangle $R$ of $\Pi^{\reff}$, the output of $\Pi$ is constant on $R$. 
\end{lemma}
\begin{proof}
         We start with some setup. For convenience, given an arbitrary communication protocol $\Pi$ with $|\Pi|=r$, we consider its tree structure. Without loss of generality, assume that at each round, a player sends exactly one bit of message. In this case, the communication tree is a binary tree. Furthermore, we extend the tree so that all leaf nodes lie at the same depth. In particular, these modifications do not increase the communication cost or decrease the advantage of $\Pi$.  Each node $u$ on the tree has an associated rectangle $R_u = \prod_{(\sfe, j) \in \calE\times [K]}A^{(\sfe, j)}$. We use $\mathcal{N}_d$ to denote the set of all rectangles (nodes) of $\Pi$ of depth $d$, where root node is of depth $0$. In particular, $\mathcal{N}_{r}$ denotes the set of all leaf rectangles (nodes) of $\Pi$, and each leaf rectangle (node) is labeled with an output, either ``1'' or ``0''. 

    With the setup described above, we now construct the global protocol $\Pi^{\reff}$ (where the superscript ``$\reff$'' stands for ``refined''). 
    The formal construction of $\Pi^{\reff}$ is described in Algorithm \ref{alg:refinement}, but it is helpful to think
    of the construction slightly less formally. Note that viewing the protocol $\Pi$ as a communication tree, we have that each node in it corresponds to a rectangle. Thus, 
    in $\Pi^{\reff}$ we proceed going over the nodes of this tree, starting with the root node of $\Pi$, and decompose each one of these rectangle into structured rectangles (as defined in \Cref{def:structured_rectangle}) using~\Cref{lem:reg}. 
\begin{algorithm}
\DontPrintSemicolon
\SetKwInOut{Input}{Input}\SetKwInOut{Output}{Output}
    \caption{Construction of the global protocol $\Pi^{\reff}$}\label{alg:refinement}
    \Input{player $(\sfe,j)$ gets input $\bfy^{(\sfe,j)}\in\Omega^{\calU_\sfe, \alpha n }$}
    \Output{a bit $\texttt{ans}\in\{0,1\}$}
    initialize: $v\leftarrow$ the root of $\Pi$; for every $(\sfe,j)\in \calE\times [K]$, $A^{(\sfe,j)} \leftarrow \Omega^{\calU_\sfe,\alpha n}$, $\bfz^{(\sfe,j)}  = \emptyset$; $R\leftarrow \prod_{(\sfe,j)} A^{(\sfe,j)}$\;
    \While {\text{$v$ is not a leaf node}}{
        suppose player $(\sfe,j)$ communicates a bit at node $v$ according to $\Pi$\;
        let $A^{(\sfe,j)}=A_{0}\cup A_{1}$ be the partition\footnotemark~at $v$ according to $\Pi$\;
        let $b \in\{0,1\}$ be such that $\bfy^{(\sfe,j)}\in A_{b}$\;
        player $(\sfe,j)$ sends $b$, and we update $A^{(\sfe,j)}\leftarrow A_{b}$, $R\leftarrow \prod_{(\sfe,j)} A^{(\sfe,j)}$, $v\gets v_{b}$\;
        \If{$A^{(\sfe,j)}$ is not $\bfz^{(\sfe,j)}$-global}{
            $(A_{(1)},\bfz_{(1)}),\dots,(A_{(\ell)},\bfz_{(\ell)})\leftarrow\mathrm{Decompose}(A^{(\sfe,j)},\bfz^{(\sfe,j)})$ 
             \tcp*{run \Cref{alg:decomposition}}
            
            let $t\in [\ell]$ be such that $\bfy^{(\sfe,j)}\in A_{(t)}$\; 
            
            player $(\sfe,j)$ sends $t$, and we update $A^{(\sfe,j)}\leftarrow A_{_{(t)}}, \bfz^{(\sfe,j)} \leftarrow \bfz_{(t)}$, $R\leftarrow \prod_{(\sfe,j)} A^{(\sfe,j)}$\;
        }
    }
    let $\texttt{ans} = 1$ if $\mathcal{D}_{\mathrm{yes}}(R)\geq \mathcal{D}_{\mathrm{no}}(R)$, otherwise let $\texttt{ans} = 0$\;
    output $\texttt{ans}$\;
\end{algorithm}
\footnotetext{
Strictly speaking, the protocol $\Pi$ at node $v$ does not directly divide the set $A^{(\sfe,j)}$ itself, since $A^{(\sfe,j)}$ is a set dynamically maintained during the execution of the \emph{refined} protocol $\Pi^{\reff}$. However, there always exists a superset $A^{(\sfe,j)}_{\text{original}}\supseteq A^{(\sfe,j)}$ that is divided at node $v$ into two subsets based on the message of player $(\sfe,j)$ in the original protocol $\Pi$. The partition $A^{(\sfe,j)}=A_{0}\cup A_{1}$ is then the restriction of the partition of $A^{(\sfe,j)}_{\text{original}}$ according to $\Pi$. 
}
    
    \paragraph{$\Pi^{\mathrm{ref}}$ is global:}
    to analyze the protocol $\Pi^{\mathrm{ref}}$, we 
    define a round of communication as the event in which a player $(\sfe,j)$ sends both a bit $b$ and an integer $t$ (see Lines 3–10 in Algorithm \ref{alg:refinement}). Note that every rectangle produced by the refined protocol $\Pi^{\reff}$ after rounds of communications is global with respect to some associated restriction. We will keep track of the restriction corresponding to each rectangle $R$. We define $\mathcal{N}^{\mathrm{ref}}_d$ to be the set of all restriction-rectangle pairs $(\boldsymbol{\zeta}, R)$ that are generated by $\Pi^{\mathrm{ref}}$ after the first $d$ rounds of communication. There are two subtle differences between $\mathcal{N}_d$ and $\mathcal{N}^{\mathrm{ref}}_d$:
\begin{itemize}
    \item $\mathcal{N}_d$ is a set of rectangles, whereas $\mathcal{N}^{\mathrm{ref}}_d$ is a set of pairs, each consisting of a restriction and a rectangle.
    \item The notion of ``depth'' differs: a rectangle $R \in \mathcal{N}_d$ is obtained by protocol $\Pi$ after exactly $d$ bits have been communicated, while a pair $(\boldsymbol{\zeta}, R) \in \mathcal{N}^{\mathrm{ref}}_d$ is produced by protocol $\Pi^{\mathrm{ref}}$ after $d$ rounds of communication, with each round involving the transmission of a bit $b \in \{0,1\}$ and an integer $t$.
\end{itemize}
First, we show that $\Pi^{\reff}$ is a global protocol as per~\Cref{def:global_protocol}. The discussion above shows that (1) $\Pi^{\reff}$ has exactly $r$ rounds of communications; (2) after $0\leq d\leq r$ rounds of communications, the resulting rectangles in $\mathcal{N}_d^{\reff}$ are all global with respect to restrictions; (3) for all $(\boldsymbol{\zeta},R)\in \mathcal{N}_{d-1}^{\reff}$ and $(\boldsymbol{\zeta}',R')\in \mathcal{N}_{d}^{\reff}$ such that $R'\subseteq R$, we have $\boldsymbol{\zeta}'$ subsume $\boldsymbol{\zeta}$. Thus, we 
have the first two items in~\Cref{def:global_protocol}, and
we next show the third item.

It suffices to upper bound the potential increment after each round of communications. Assume that after $d$ rounds of communication according to $\Pi^{\reff}$, we obtain a structured rectangle $(\boldsymbol{\zeta},R)\in\mathcal{N}_{d}^{\reff}$ and player $(\sfe,j)$ will speak in the next round. The communication of the player $(\sfe,j)$ divides $A^{(\sfe,j)}$ into two parts $A_0,A_1$, which decomposes the rectangle $R$ into two disjoint rectangles $R_0,R_1$ via the message of $b$ (see lines 3 to 6 in Algorithm \ref{alg:refinement}). In lines 7 to 10, $R_0$ and $R_1$ are further decomposed into several structured rectangles separately by the message $t\in [\ell]$. We have:
\begin{align*}
    \sum_{\substack{(\boldsymbol{\zeta}',R')\in \mathcal{N}_{d+1} ^{\reff}\\ R'\subseteq R}}\frac{|R'|}{|R|}\cdot  \phi(\boldsymbol{\zeta}',R') &= \frac{|R_0|}{|R|}\sum_{\substack{(\boldsymbol{\zeta}',R')\in \mathcal{N}_{d+1} ^{\reff}\\ R'\subseteq R_0}}\frac{|R'|}{|R_0|}\cdot \phi(\boldsymbol{\zeta}',R')+ \frac{|R_1|}{|R|}\sum_{\substack{(\boldsymbol{\zeta}',R')\in \mathcal{N}_{d+1} ^{\reff}\\ R'\subseteq R_1}}\frac{|R'|}{|R_1|}\cdot \phi(\boldsymbol{\zeta}',R')\\
    &\leq \frac{|R_0|}{|R|}\left(\phi(\boldsymbol{\zeta},R)+\log_{2}\left(\frac{|R|}{|R_0|}\right)+2\right)+  \frac{|R_1|}{|R|}\left(\phi(\boldsymbol{\zeta},R)+\log_{2}\left(\frac{|R|}{|R_1|}\right)+2\right) \\
    &\leq \phi(\boldsymbol{\zeta},R)+ 3,
\end{align*}
where the second transition is by~\Cref{lem:reg}, and the last transition comes from the fact that the binary entropy is upper bounded by $1$.

\paragraph{$\Pi$ is constant on leaf rectangles of $\Pi^\reff$:} 
    by construction, it is easy to see that each leaf rectangle $R$ of $\Pi$ is subdivided into several subrectangles by the refined protocol $\Pi^{\reff}$, thus the statment holds. 
\end{proof}
\begin{remark}
    The structured rectangles induced by the refined protocol $\Pi^\reff$ are exactly the collection $\calR$ that we want to construct in Lemma \ref{lem:regularity_decomposition} (after deleting some bad pairs). We will formally prove it in Section \ref{sec:bound_bad_pairs}. 
\end{remark}

\subsection{Bounding the Weights of ``Bad'' Rectangles}\label{sec:bound_bad_pairs}
Having transformed arbitrary protocols to global ones, we proceed to analyze global protocols. Assume that we have a $r$-round global protocol $\Pi$. To prove Lemma \ref{lem:regularity_decomposition} for $\Pi$, it suffices to prove that only a small fraction of the structured rectangles in $\calR$ are not $(10^5 r)$-good (as defined in \Cref{def:good_rec}). To this end, we make the following definitions.

\begin{definition}
Let $\Pi$ be an $r$-round global protocol.  
For each $0 \leq d \leq r$, we let $\mathcal{R}^{d}(\Pi)$ denote the collection of structured rectangles obtained after $d$ rounds of communication in $\Pi$.  
In particular, we write $\calR^\leaf(\Pi)=\calR^r(\Pi)$ for the set of structured rectangles at the leaves of the protocol tree. For $d\geq r+1$, the collection $\calR^{d}(\Pi)$ is an emptyset.
\end{definition}

\begin{definition}\label{def:deleting_bad_pairs}
Let $\Pi$ be a global protocol for $\DIHP(G,n,\alpha,K)$ and let $W$ be a positive real number. We define
$\calR^{\mathrm{bad}}(\Pi,W)$ as the following set of structured rectangles:
\begin{align*}
    \calR^{\mathrm{bad}}(\Pi,W) :=\left\{(\boldsymbol{\zeta},R)\in\calR^\leaf: \text{$(\boldsymbol{\zeta},R)$ is not $W$-good as per \Cref{def:good_rec}}\right\}.
\end{align*}
\end{definition}
The main lemma of this section is the following.
\begin{lemma}\label{lem:analysis_of_global_protocols}
    For any fixed distribution-labeled $k$-graph $G$, integer $K>0$ and parameter $\alpha>0$, there exists a constant $\eta>0$ such that if $r\leq  \eta \sqrt{n}$, any $r$-round global protocol $\Pi$ for $\mathsf{DIHP}(G,n,\alpha,K)$ satisfies 
    \begin{align*}
        \sum_{(\bdzeta,R)\in \calR^{\bad}(\Pi,10^5r)}\calD_{\no}(R) \leq 0.01. 
    \end{align*}
\end{lemma}

We observe that the desired \Cref{lem:regularity_decomposition} (restated below) follows immediately from \Cref{lem:analysis_of_global_protocols}. 

\DecompositionLemma*

\begin{proof}[Proof of \Cref{lem:regularity_decomposition} assuming \Cref{lem:analysis_of_global_protocols}]
Let $r = |\Pi|$, and apply \Cref{lem:arbitrary_to_global} to transform $\Pi$ into an $r$-round global protocol $\Pi^{\mathrm{ref}}$.  
Consider the collection
\[
    \calR := \calR^{\mathrm{leaf}}\!\left(\Pi^{\mathrm{ref}}\right) 
    \setminus \calR^{\mathrm{bad}}\!\left(\Pi^{\mathrm{ref}}, 10^{5} r\right).
\]
We claim that $\calR$ satisfies the three conditions stated in the lemma.  

The second condition follows directly from \Cref{def:deleting_bad_pairs}, and the third condition follows from the guarantee of \Cref{lem:arbitrary_to_global}.  
The first condition follows from \Cref{lem:analysis_of_global_protocols} together with the obvious identity
\[
    \sum_{(\boldsymbol{\zeta}, R) \in \calR^{\mathrm{leaf}}(\Pi^{\mathrm{ref}})} 
    \calD_{\no}(R) = 1.\qedhere
\]
\end{proof}

The remainder of this section is devoted for the proof of \Cref{lem:analysis_of_global_protocols}, modulo two technical claims that are proved in~\Cref{app:techinical_lemma}. The following notations will be useful in the proof.

\begin{notation}
For a set $T$ of $k$-hyperedges on the vertex set $\calV \times [n]$, we write $V(T)$ for the set of vertices incident to at least one hyperedge in $T$. For a restriction sequence $\boldsymbol{\zeta} = \left( \bfz^{(\sfe,j)} \right)_{(\sfe,j) \in \calE \times [K]}$, we define
\[
V(\bdzeta):=\bigcup_{(\sfe,j)\in \calE\times [K]}V\left(\supp(\bfz^{(\sfe,j)})\right).
\]
\end{notation}

\begin{proof}[Proof of \Cref{lem:analysis_of_global_protocols}]
First, we define the following two subcollections $\calR_{1},\calR_{2}\subseteq \calR^{\leaf}(\Pi)$:
    \begin{enumerate}[label=(\arabic*)]
        \item $\mathcal{R}_1$ is the collection of structured rectangles $(\boldsymbol{\zeta}, R)\in \mathcal{R}^\leaf(\Pi)$ such that $\phi(\zeta, R)> 10^5 r$; 
        \item $\mathcal{R}_2$ is the collection of 
        structured rectangles $(\boldsymbol{\zeta}, R)\in \mathcal{R}^\leaf(\Pi)$ such that either the hyperedge sets $\left(\supp(\bfz^{(\sfe,j)})\right)_{(\sfe,j)\in \calE\times [K]}$ are not pairwise disjoint or their union contains a cycle. 
    \end{enumerate}
    It is easy to see that $\calR^\bad(\Pi,10^5 r) \subseteq \calR_1\cup \calR_2 $. The pairs that violate the first condition in the definition of goodness (\Cref{def:good_rec}) are included in $\calR_2$, while the pairs violating the second or the third conditions are included in $\calR_1$. It suffices to prove the following bounds regarding the two subcollections $\mathcal{R}_1, \mathcal{R}_2$:
    \begin{enumerate}[label=(\arabic*)]
        \item $\sum_{(\boldsymbol{\zeta}, R)\in \mathcal{R}_1} \Dno(R)\leq 0.005$;
        \item $\sum_{(\boldsymbol{\zeta},R)\in \mathcal{R}_2\setminus\mathcal{R}_1} \Dno (R)\leq 0.005$.
    \end{enumerate}

    \paragraph{Upper bound for $\mathcal{R}_1$.} To upper bound of the total weight of structured rectangles $(\boldsymbol{\zeta}, R)$ with $\phi(\boldsymbol{\zeta},R)\geq  10^5 r$, we first bound the weighted sum of potentials $\phi(\boldsymbol{\zeta},R)$ over all leaf pairs $(\boldsymbol{\zeta},R)\in \mathcal{R}^{\leaf}(\Pi)$. More precisely, we show that
    \begin{align}\label{eq:weighted_sum}
        \sum_{(\boldsymbol{\zeta},R)\in \mathcal{R}^{\leaf}(\Pi)} \frac{|R|}{\left|\prod_{(\sfe, j) \in \calE \times [K]} \Omega^{\calU_{\sfe}, \alpha n}\right|}\cdot \phi (\boldsymbol{\zeta},R) \leq 3\cdot r,
    \end{align}
    and the proof proceeds by induction argument on the depth $d$: we prove that for all $d$,
    \begin{align}\label{eq:induction_hypo}
    \sum_{(\boldsymbol{\zeta},R)\in \mathcal{R}^d(\Pi)} \frac{|R|}{\left|\prod_{(\sfe, j) \in \calE \times [K]} \Omega^{\calU_{\sfe}, \alpha n}\right|}\cdot \phi(\boldsymbol{\zeta},R)\leq 3\cdot d.
    \end{align}
    When $d=0$ the statement is clear as $\mathcal{R}^0(\Pi)$ only
    contains the trivial rectangle $\prod_{(\sfe, j) \in \calE \times [K]} \Omega^{\calU_{\sfe}, \alpha n}$, whose potential equals $0$.
    Let $d>0$ and assume that~\eqref{eq:induction_hypo} holds for $d-1$.  We have
    \begin{align*}
        &\quad\sum_{(\boldsymbol{\zeta}',R')\in\mathcal{R}^d(\Pi)} \frac{|R'|}{\left|\prod_{(\sfe, j) \in \calE \times [K]} \Omega^{\calU_{\sfe}, \alpha n}\right|}\cdot \phi(\boldsymbol{\zeta}',R') \\ 
        &= \sum_{(\boldsymbol{\zeta},R)\in \mathcal{R}^{d-1}(\Pi)} \frac{|R|}{\left|\prod_{(\sfe, j) \in \calE \times [K]} \Omega^{\calU_{\sfe}, \alpha n}\right|}\sum_{\substack{(\boldsymbol{\zeta}',R')\in \mathcal{R}^d(\Pi)\\ R'\subseteq R,}} \frac{|R'|}{|R|}\cdot \phi(\boldsymbol{\zeta}',R') \\
        &\leq \sum_{(\boldsymbol{\zeta},R)\in \mathcal{R}^{d-1}(\Pi)} \frac{|R|}{\left|\prod_{(\sfe, j) \in \calE \times [K]} \Omega^{\calU_{\sfe}, \alpha n}\right|} \cdot \left(\phi(\boldsymbol{\zeta},R)+3 \right)\leq 3(d-1) + 3 = 3d,
    \end{align*}
    where the second transition is by~\Cref{def:global_protocol}, and the last transition is due to the inductive hypothesis. This completes the inductive step, and in particular establishes~\eqref{eq:weighted_sum}.
    
    The bound on $\mathcal{R}_1$ now follows by Markov's inequality applied on~\eqref{eq:weighted_sum}:
    \begin{align*}
        \sum_{(\boldsymbol{\zeta},R)\in \mathcal{R}_1} \Dno(R) &= \sum_{(\boldsymbol{\zeta},R)\in \mathcal{R}_1} \frac{|R|}{\left|\prod_{(\sfe, j) \in \calE \times [K]} \Omega^{\calU_{\sfe}, \alpha n}\right|} \\ &\leq \frac{1}{10^{5}r}\sum_{(\bdzeta,R)\in \mathcal{R}_{\leaf}}\frac{|R|}{\left|\prod_{(\sfe, j) \in \calE \times [K]} \Omega^{\calU_{\sfe}, \alpha n}\right|}\cdot \phi(\boldsymbol{\zeta},R) \\&\leq \frac{3\cdot r}{10^5  r}<0.005.
    \end{align*}
    \paragraph{Upper bound for $\mathcal{R}_2\setminus \mathcal{R}_1$.} For an integer $d\in\{1,2,\dots,r\}$ and a structured rectangle $(\bdzeta,R)\in \calR^{d}(\Pi)$, we write \[(\bdzeta,R)\mapsto \left(\widetilde{\bdzeta},\widetilde{R}\right)\] 
    for the unique parent structured rectangle $\left(\widetilde{\bdzeta},\widetilde{R}\right)\in \calR^{d-1}(\Pi)$ such that $R\subseteq \widetilde{R}$. When we write a summation \[\sum_{(\bdzeta,R)\mapsto \left(\widetilde{\bdzeta},\widetilde{R}\right)}(\cdot),\]
    we mean the sum over all structured rectangles $(\bdzeta,R)\in \bigcup_{d=1}^{r}\calR^d(\Pi)$ together with their respective parents $\left(\widetilde{\bdzeta},\widetilde{R}\right)$.
    
    The main component of the proof is to obtain the following inequality for some constant $J$:
    \begin{align}\label{eq:cycle}
        \sum_{(\boldsymbol{\zeta},R)\in \mathcal{R}_2\setminus \mathcal{R}_1} \Dno(R) \leq \frac{J}{n}\cdot \sum_{\substack{(\bdzeta,R)\mapsto \left(\widetilde{\bdzeta},\widetilde{R}\right)\\ |V(\boldsymbol{\zeta})|\leq 10^5 r}}\Dno(R) \left(|V(\bdzeta)|^2 - \left|V(\widetilde{\bdzeta})\right|^2\right).
    \end{align}
    Indeed, once we have \eqref{eq:cycle}, a change of variables on the right-hand side gives
    \begin{align}
    \sum_{(\boldsymbol{\zeta},R)\in \mathcal{R}_2\setminus \mathcal{R}_1} \Dno(R) &\leq \frac{J}{n}\left(\sum_{\substack{(\bdzeta,R)\mapsto \left(\widetilde{\bdzeta},\widetilde{R}\right)\\ |V(\boldsymbol{\zeta})|\leq 10^5 r}}\calD_{\no}(R)|V(\bdzeta)|^{2}-\sum_{\substack{(\bdzeta',R')\mapsto (\bdzeta,R)\\ |V(\boldsymbol{\zeta}')|\leq 10^5 r}}\calD_{\no}(R')|V(\bdzeta)|^{2}\right) \nonumber\\
    &=\frac{J}{n}\cdot \sum_{d=0}^{r}\sum_{\substack{(\bdzeta,R)\in \calR^{d}(\Pi)\\ |V(\bdzeta)|\leq 10^5 r}}|V(\bdzeta)|^{2}\cdot \calD\Big(\mathrm{Reduced}(\bdzeta,R)\Big), \label{eq:change_of_variables}
    \end{align}
    where for any $(\bdzeta,R)\in \calR^{d}(\Pi)$ with $0\leq d\leq r$, we define 
    \[\mathrm{Reduced}(\bdzeta,R):=R\setminus\left(\bigcup_{(\bdzeta',R')} R'\right),\] with the union taken over all children pairs $(\bdzeta',R')\in\calR^{d+1}(\Pi)$ such that $|V(\bdzeta')|\leq 10^5r$ and $(\bdzeta',R')\mapsto (\bdzeta,R)$. Note that since any leaf pair $(\bdzeta,R)\in \calR^{r}(\Pi)$ has no children in $\calR^{r+1}(\Pi)$ (which is empty), for these structured rectangles we have $\mathrm{Reduced}(\bdzeta,R)\defeq R$.

    Since the collection of rectangles $\mathrm{Reduced}(\bdzeta,R)$ for $(\bdzeta,R)\in \bigcup_{d=0}^{r}\calR^{d}(\Pi)$ are clearly pairwise disjoint, the sum of their $\calD_{\no}$ weights is at most 1. We can thus conclude from \eqref{eq:change_of_variables}
    \begin{align*}
    \sum_{(\boldsymbol{\zeta},R)\in \mathcal{R}_2\setminus \mathcal{R}_1} \Dno(R)\leq \frac{J}{n}\cdot (10^5r)^{2}\sum_{d=0}^{r}\sum_{\substack{(\bdzeta,R)\in \calR^{d}(\Pi)\\ |V(\bdzeta)|\leq 10^5 r}} \calD\left(\mathrm{Reduced}(\bdzeta,R)\right)\leq \frac{J}{n}\cdot (10^5r)^{2}.
    \end{align*}
    When $\eta$ is chosen to be sufficiently small and $r\leq \eta\sqrt{n}$, the value of $\frac{J}{n}\cdot (10^5r)^{2}$ is at most $0.005$, as desired.

    \paragraph{Proof of the inequality \eqref{eq:cycle}.} 
    Suppose $(\bdzeta,R)$ is a structured rectangle in $\bigcup_{d=1}^{r}\calR^{d}(\Pi)$, where the restriction sequence $\bdzeta$ is written out as $\left(\bfz^{(\sfe,j)}\right)_{(\sfe,j)\in \calE\times [K]}$. Let $\left(\widetilde{\bdzeta},\widetilde{R}\right)$ be parent of $(\bdzeta,R)$. We define the subset   $B(\boldsymbol{\zeta},R)\subseteq R$ as the set of all joint inputs $\left(\bfy^{(\sfe,j)}\right)_{(\sfe,j)\in \calE\times K}\in R$ satisfying the following: there exists a player $(\sfe,j)\in \calE\times [K]$ and an edge $e\in \supp \left(\bfy^{(\sfe,j)}\right)\setminus \supp \left(\bfz^{(\sfe,j)}\right)$ such that 
    \[\left|V(\{e\}) \cap V(\boldsymbol{\zeta})\right|\geq 2\quad\text{and}\quad \left| V(\{e\}) \cap \left(V(\boldsymbol{\zeta}) \setminus V(\widetilde{\boldsymbol{\zeta}})\right)\right|\geq 1.\]
    We make the following two claims, the proofs of  which are deferred to \Cref{app:techinical_lemma}. 
    \begin{claim}\label{claim:cycle}
    For every structured rectangle $(\boldsymbol{\zeta},R)\in \mathcal{R}_2\setminus \mathcal{R}_1$, there exists an integer $d\in [r]$ and a structured rectangle $(\boldsymbol{\zeta}^{\anc},R^\anc)\in \mathcal{R}^{d}(\Pi)$ such that $|V(\boldsymbol{\zeta}^\anc)|\leq 10^5 r$ and $R\subseteq B(\boldsymbol{\zeta}^\anc,R^\anc)$. 
    \end{claim}
    \begin{claim}\label{lem:low_cycle_probability}
    There exists a constant $J$ (depending only on $G,\alpha,K$) such that for any structured rectangle $(\boldsymbol{\zeta},R)\in\bigcup_{d=1}^{r}\calR^d(\Pi)$ with $|V(\boldsymbol{\zeta})| \leq \sqrt{n}$, if $\left(\widetilde{\bdzeta},\widetilde{R}\right)$ is its parent, then
    \begin{align*}
    \mathcal{D}_{\mathrm{no}}(B(\boldsymbol{\zeta},R)) \leq \frac{J}{n}\cdot \Dno(R) \left(|V(\bdzeta)|^2 - \left|V(\widetilde{\bdzeta})\right|^2\right).
    \end{align*}
\end{claim}
    \Cref{claim:cycle} implies
    \begin{align}\label{eq:cycle_2}
        \sum_{(\boldsymbol{\zeta},R)\in \mathcal{R}_2\setminus \mathcal{R}_1} \Dno(R) \leq \sum_{d=1}^{r}\sum_{\substack{\left(\boldsymbol{\zeta},R\right)\in\mathcal{R}^d(\Pi) \\ |\boldsymbol{\zeta}|\leq 10^5 r}} \Dno(B(\boldsymbol{\zeta},R)),
    \end{align}
    since the structured rectangles in $\calR_{2}\setminus \calR_{1}$ are pairwise disjoint. Then plugging \Cref{lem:low_cycle_probability} into \eqref{eq:cycle_2} yields the desired inequality \eqref{eq:cycle}, as long as $\eta\leq 10^{-5}$.  
\end{proof}

    \subsection{Proofs of the Technical Claims}\label{app:techinical_lemma}
    In this subsection, we finish the proofs of \Cref{claim:cycle} and \Cref{lem:low_cycle_probability}, thereby completing the proof of \Cref{lem:analysis_of_global_protocols}. Throughout the proofs, we keep the notations introduced in the proof of \Cref{lem:analysis_of_global_protocols}.
    We first prove \Cref{claim:cycle}.
    
    \begin{proof}[Proof of \Cref{claim:cycle}.]
        Fix a structured rectangle $(\boldsymbol{\zeta},R)\in \mathcal{R}_2\setminus \mathcal{R}_1$. By definition, $(\boldsymbol{\zeta},R)\in \mathcal{R}^{\mathrm{leaf}}(\Pi)=\mathcal{R}^r(\Pi)$ and $|V(\boldsymbol{\zeta})|\le 10^5 r$.  
        
        There is a unique path $P$ in the protocol tree of $\Pi$ that starts at this leaf $(\boldsymbol{\zeta},R)$ and terminates at the root rectangle (namely the whole space $\prod_{(\sfe,j)\in\mathcal{E}\times[K]}\Omega^{\mathcal{U}_{\sfe},\alpha n}$), obtained by iteratively replacing each node by its parent. In other words, each rectangle on the path $P$ is the parent of its predecessor, and the sequence ends at the root. Let  
    \[
\left(\bdzeta'=\left(\bfz'^{(\sfe,j)}\right)_{(\sfe,j)\in\calE\times[K]},\,R'\right)
    \]  
    be the first structured rectangle along this path that satisfies the first condition of goodness (in \Cref{def:good_rec}), namely:  
    \begin{itemize}
        \item The support sets $\bigl(\supp(\bfz^{(\sfe,j)})\bigr)_{(\sfe,j)\in \calE\times [K]}$ are pairwise disjoint;  
        \item The union $\bigcup_{(\sfe,j)\in \calE\times [K]}\supp(\bfz^{(\sfe,j)})$ is cycle-free.  
    \end{itemize}
        
    Note that by the definition of $\calR_{2}$, the initial structured rectangle on the path $P$ --- namely, the leaf $(\bdzeta,R)$ --- does not satisfy this condition. On the other hand, the final structured rectangle on the path, i.e., the whole space $\prod_{(\sfe,j)\in\mathcal{E}\times[K]}\Omega^{\mathcal{U}_{\sfe},\alpha n}$ with empty restrictions, clearly does. Hence $(\bdzeta',R')$ is well defined as an element of $\calR^{d}(\Pi)$ for some $d\in\{0,1,\dots,r-1\}$.

    We let
    \[
\left(\bdzeta''=\left(\bfz''^{(\sfe,j)}\right)_{(\sfe,j)\in\calE\times[K]},\,R''\right)
    \]
    be the predecessor of $(\bdzeta',R')$ on the path. By the choice of $(\bdzeta',R')$, we know that the structured rectangle $(\bdzeta'',R'')$ does not satisfy the first condition of goodness. In particular, $\bdzeta''\neq \bdzeta'$. Furthermore, because exactly one player speaks in each round of the protocol $\Pi$, there exists exactly one player $(\sfe^*,j^*)\in\calE\times [K]$ for which $\bfz''^{(\sfe^*,j^*)} \neq \bfz'^{(\sfe^*,j^*)}$. 

    The fact that $(\bdzeta'',R'')$ does not satisfy the first condition of goodness gives rise to the following two cases.
    
    \paragraph{Case 1: $\bigcup_{(\sfe,j)\in \calE\times [K]} \supp\left(\bfz''^{(\sfe,j)}\right)$ contains a cycle.} Since $(\bdzeta',R')$ satisfies the first condition of goodness, we know that in this case, the union

    \begin{equation}\label{eq:union_E}
    E:=\bigcup_{(\sfe,j)\in \calE\times [K]}\supp\left(\bfz'^{(\sfe,j)}\right)
    \end{equation}
    is cycle-free, while appending the matching
    \begin{equation}\label{eq:appending_M}
    M:=\supp\left(\bfz''^{(\sfe^*,j^*)}\right)\setminus \supp\left(\bfz'^{(\sfe^*,j^*)}\right)
    \end{equation}
    to $E$ creates some cycle. We claim that there must exist an edge $e\in M$ such that at least two of its vertices already lie in $V(E)$. Indeed, suppose otherwise. Then there would exist a set of edges in $E\cup M$, with $\ell_{1}$ edges from $E$ and $\ell_{2}$ edges from $M$, that covers at most $(\ell_{1}+\ell_{2})(k-1)$ vertices (see \Cref{subsec:general_notations} for the definition of cycle-freeness). Since $E$ is cycle-free the $\ell_{1}$ edges from $E$ together cover at least $\ell_{1}(k-1)+1$ vertices. Then each of the $\ell_{2}$ edges from $M$ would then contribute at least $k-1$ new vertices, giving an additional $\ell_{2}(k-1)$. Thus the total number of covered vertices would be at least $(\ell_{1}+\ell_{2})(k-1)+1$, contradicting the assumption. Hence such an edge $e$ must exist.

    Now let \[\left(\bdzeta^{\anc}=\left({\bfz^\anc}^{(\sfe,j)}\right)_{(\sfe,j)\in \calE\times [K]},\,R^{\anc}\right)\] be the last structured rectangle along the segment of the path $P$ from $(\bdzeta',R')$ to the root such that $V(\bdzeta^{\anc})$ contains at least two vertices of $e$. We clearly have $\left|V(\bdzeta^{\anc})\right|\leq |V(\bdzeta)|\leq 10^5r$. It remains to show $
R \subseteq B(\bdzeta^{\anc},R^{\anc})$. Since $R \subseteq R''$, it suffices to show that $R'' \subseteq B(\bdzeta^{\anc},R^{\anc})$. Indeed, as $e$ belongs to the difference
\[
\supp(\bfz''^{(\sfe^*,j^*)}) \setminus \supp(\bfz'^{(\sfe^*,j^*)}),
\]
we deduce that $e$ also belongs to
\[
\supp(\bfy''^{(\sfe^*,j^*)}) \setminus \supp({\bfz^{\anc}}^{(\sfe^*,j^*)}),
\]
since this is an enlargement of the larger set and a shrinking of the smaller one. Consequently, we have 
\[
\left(\bfy''^{(\sfe,j)}\right)_{(\sfe,j)\in\calE\times[K]} \in B(\bdzeta^{\anc},R^{\anc}),
\]
because 
\[
\left|V(\{e\}) \cap V(\bdzeta^{\anc})\right| \geq 2
\quad \text{and} \quad
\left|V(\{e\}) \cap \left(V(\bdzeta^{\anc}) \setminus V(\widetilde{\bdzeta^{\anc}})\right)\right| \geq 1,
\]
where $\widetilde{\bdzeta^{\anc}}$ denotes the restriction sequence of the parent of $(\bdzeta^{\anc},R^{\anc})$.
    
        \paragraph{Case 2: The support sets $\bigl(\supp(\bfz^{(\sfe,j)})\bigr)_{(\sfe,j)\in \calE\times [K]}$ are not pairwise disjoint.} In this case, we still define the edge sets $E$ and $M$ as in \eqref{eq:union_E} and \eqref{eq:appending_M}. Now there must exists an edge $e\in M$ that already belongs to $E$. In particular, at least two vertices of $e$ belongs to $V(E)$. We can again let $(\bdzeta^{\anc},R^{\anc})$ be the last structured rectangle along the segment of the path $P$ from $(\bdzeta',R')$ to the root such that $V(\bdzeta^{\anc})$ contains at least two vertices of $e$. The same conclusions as in Case 1 follows.
    \end{proof}
    Next, we prove \Cref{lem:low_cycle_probability}. 
    \begin{proof}[Proof of \Cref{lem:low_cycle_probability}.]

    Let $\bdzeta=\left(\bfz^{(\sfe,j)}\right)_{(\sfe,j)\in \calE\times [K]}$ and $R=\prod_{(\sfe,j)\in \calE\times [K]}A^{(\sfe,j)}$. For each $(\sfe^*,j^*)\in \calE\times [K]$, we define $B^{(\sfe^*,j^*)}(\boldsymbol{\zeta},R)$ to be the set of all joint inputs 
    \[(\bfy^{(\sfe,j)})_{(\sfe,j)\in \mathcal{E}\times [K]}\in R\]
    whose $(\sfe^*,j^*)$-coordinate, $\bfy^{(\sfe^*,j^*)}$, satisfies the following: there exists an edge 
    \[e\in \supp \left(\bfy^{(\sfe^*,j^*)}\right)\setminus \supp \left(\bfz^{(\sfe^*,j^*)}\right)\] such that
   \[\left|V(\{e\}) \cap V(\boldsymbol{\zeta})\right|\geq 2\quad\text{and}\quad\left| V(\{e\}) \cap \left(V(\boldsymbol{\zeta}) \setminus V(\Tilde{\boldsymbol{\zeta}})\right)\right|\geq 1.\]

   By definition, $B(\boldsymbol{\zeta},R)= \bigcup_{(\sfe,j)\in \calE\times [K]} B^{(\sfe,j)} (\boldsymbol{\zeta},R)$. Therefore, it suffices to show that for any fixed player $(\sfe,j)\in \calE\times [K]$, we have \begin{equation}\label{eq:cycle_prob_fixed_player}
   \mathcal{D}_{\mathrm{no}}
    \left(B^{(\sfe,j)}(\boldsymbol{\zeta},R)\right)\leq \frac{J}{|\calE|K\cdot n}\cdot  \mathcal{D}_{\mathrm{no}}(R) \left(|V(\boldsymbol{\zeta})|^2 - \left|V(\widetilde{\boldsymbol{\zeta}})\right|^2\right)
    \end{equation}
    for some constant $J$. In the remaining of the proof, we fix a player $(\sfe,j)\in \calE\times[K]$.
    
    We define the edge set
    \[E:=\left\{e\in \tprod (\calU_{\sfe})_{\setminus \supp(\bfz^{(\sfe,j)})}\,\middle|\, \left|V(\{e\}) \cap V(\boldsymbol{\zeta})\right|\geq 2\text{ and }\left| V(\{e\}) \cap \left(V(\boldsymbol{\zeta}) \setminus V(\Tilde{\boldsymbol{\zeta}})\right)\right|\geq 1\right\}.\]
    It is easy to see that
    \begin{equation}\label{eq:cycle_prob_size_of_E}
    |E|\leq \left|V(\bdzeta)\setminus V(\widetilde{\bdzeta})\right|\cdot \left|V(\bdzeta)\right|\cdot n^{k-2}\leq n^{k-2}\left(|V(\boldsymbol{\zeta})|^2 - \left|V(\widetilde{\boldsymbol{\zeta}})\right|^2\right).
    \end{equation}
    For each edge $e\in E$, we define the restricted domain
    \[
    \Omega[e]:=\left\{\bfy\in \Omega^{\calU_{\sfe},\alpha n}_{\bfz^{(\sfe,j)}}\,\middle|\, e\in \supp(\bfy)\right\}.
    \]
    By the definitions of $B^{(\sfe,j)}(\bdzeta,R)$ and the uniformity of $\calD_{\no}$, we have
    \begin{equation}\label{eq:cycle_prob_ratio_D}
    \frac{\mathcal{D}_{\mathrm{no}}
    \left(B^{(\sfe,j)}(\boldsymbol{\zeta},R)\right)}{
    \Dno(R)
    }=\frac{\left|A^{(\sfe,j)}\cap\bigcup_{e\in E}\Omega[e]\right|}{\left|A^{(\sfe,j)}\right|}\leq \sum_{e\in E}\frac{\left|A^{(\sfe,j)}\cap \Omega[e]\right|}{\left|A^{(\sfe,j)}\right|}.
    \end{equation}
    The $\bfz^{(\sfe,j)}$-globalness of $A^{(\sfe,j)}$ implies that for each $e\in E$,
    \begin{equation}\label{eq:cycle_prob_globalness}
    \frac{\left|A^{(\sfe,j)}\cap \Omega[e]\right|}{\left|A^{(\sfe,j)}\right|}
    \leq
    2\cdot \frac{\left|\Omega^{\calU_{\sfe},\alpha n}_{\bfz^{(\sfe,j)}}\cap\Omega[e]\right|}{\left|\Omega^{\calU_{\sfe},\alpha n}_{\bfz^{(\sfe,j)}}\right|}
    =2\cdot \frac{\alpha n-\left|\supp(\bfz^{(\sfe,j)})\right|}{\left(n-\left|\supp(\bfz^{(\sfe,j)})\right|\right)^{k}}\leq \frac{J}{|\calE|K\cdot n^{k-1}}  \end{equation}
    for some constant $J>0$.
    Here, the last transition uses the assumption that $|V(\boldsymbol{\zeta})|\leq \sqrt{n}$. Combining \eqref{eq:cycle_prob_size_of_E}, \eqref{eq:cycle_prob_ratio_D} and \eqref{eq:cycle_prob_globalness} yields the desired inequality \eqref{eq:cycle_prob_fixed_player}.
\end{proof}

\section{Global Hypercontractivity in $\Omega$}
\label{app:global_hypercontractivity}

In this appendix we prove \Cref{thm:level-d-inequality}. The proof is a rather straightforward adaptation of \cite[Section 4]{FMW25}. We begin in \Cref{subsec:derivatives_compose,subsec:operators_commute} by establishing basic properties of the derivative operators (defined in \Cref{def:derivative}) and the projection operators (defined in \Cref{def:level-d-projection}). In \Cref{subsec:compartison_product,subsec:hypercontractive}, we incorporate a result from \cite{KLM23} by comparing our labeled matching space $\Omega^{\calU,m}$ with a product space. Finally, we conclude with the proof of \Cref{thm:level-d-inequality} in \Cref{subsec:level-d}.

The following two notations will be used throughout this appendix.

\begin{notation}
Suppose $S$ and $T$ are disjoint finite sets. For two maps $\bfx_{1}:S\rightarrow \ZNk$ and $\bfx_{2}:T\rightarrow\ZNk$, we define their concatenation $\bfx_{1}\uplus \bfx_{2}:S\sqcup T\rightarrow \ZNk$ by setting $(\bfx_{1}\uplus \bfx_{2})(e):=\bfx_{1}(e)$ for $e\in S$ and $(\bfx_{1}\uplus \bfx_{2})(e):=\bfx_{2}(e)$ for $e\in T$.
\end{notation}

\begin{notation}
Suppose $S$ and $M$ are finite sets such that $S\subseteq M$. For a map $\bfa:M\rightarrow\ZNk$, we define $\bfa_{|S}:S\rightarrow \ZNk$ to be the restriction of $\bfa$ to $S$, and define $\bfa_{\setminus S}:M\setminus S\rightarrow\ZNk$ to be the restriction of $\bfa$ to $M\setminus S$. 
\end{notation}

\subsection{Derivatives Compose}\label{subsec:derivatives_compose}

We observe that the gadget function $H_{S}$ (defined in \Cref{def:gadget_H}) used in the definition of the discrete derivative operators has the following simple Fourier decomposition.

\begin{lemma}\label{lem:derivative_primitive}
We have the identity
$H_{S}=\sum_{\bfa:S\rightarrow\ZNk\setminus\{0\}}\chi_{\bfa}.$
\end{lemma}
\begin{proof}
Straightforward calculation shows that for any $\bfz\in \Map{S}{\ZNk}$,
\begin{align*}
\sum_{\bfa:S\rightarrow\ZNk\setminus\{0\}}\chi_{\bfa}(\bfz)&=\sum_{\bfa:S\rightarrow\ZNk\setminus\{0\}}\left(\prod_{e\in S}\chi_{\bfa(e)}(\bfz(e))\right)=\prod_{e\in S}\left(\sum_{\bfa(e)\in \ZNk\setminus\{0\}}\chi_{\bfa(e)}(\bfz(e))\right)\\
&=\prod_{e\in S}\left(N^{k}\cdot \ind{\bfz(e)=0}-\chi_{0}(\bfz(e))\right)=\prod_{e\in S}\left(N^{k}\cdot \ind{\bfz(e)=0}-1\right)\\
&=\prod_{e\in S}\left(N^{k}-1\right)^{\ind{\bfz(e)=0}}(-1)^{\ind{\bfz(e)\neq 0}}=H_{S}(\bfz).\qedhere
\end{align*}
\end{proof}

This Fourier decomposition allows us to prove that the derivative operators (defined in \Cref{def:derivative}) compose with each other in the following natural way.

\begin{lemma}\label{lem:derivatives-compose}
Suppose $S$ and $T$ are vertex disjoint matchings in $\calM_{\calU,\leq m}$ with $|S\cup T|\leq m$. Fix labels $\bfx_{1}:S\rightarrow\ZNk$ and $\bfx_{2}:T\rightarrow\ZNk$. For any $f:\Omega^{\calU,m}\rightarrow\mathbb{C}$ we have
\[
D_{S,\bfx_{1}}D_{T,\bfx_{2}}[f]=D_{S\cup T,\bfx_{1}\uplus \bfx_{2}}[f].
\]
\end{lemma}
\begin{proof}
We have the following three canonical embeddings from \Cref{def:embedding}:
\begin{align*}
\fraki&:\Omega^{\calU,m}_{\setminus (S\cup T)}\times \Map{S\cup T}{\ZNk}\hookrightarrow \Omega^{\calU,m},\\
\fraki_{1}&:\Omega^{\calU,m}_{\setminus (S\cup T)}\times\Map{S}{\ZNk}\hookrightarrow \Omega^{\calU,m}_{\setminus T},\text{ and}\\
\fraki_{2}&:\Omega^{\calU,m}_{\setminus (S\cup T)}\times\Map{S}{\ZNk}\times \Map{T}{\ZNk}\hookrightarrow \Omega^{\calU,m}.
\end{align*}
Directly from \Cref{def:derivative} we get
\begin{align*}
D_{S,\bfx_{1}}D_{T,\bfx_{2}}[f](\bfy)&=\Exu{\bfz_{1}:S\rightarrow\ZNk}{H_{S}(\bfz_{1})\cdot D_{T,\bfx_{2}}[f]\Big(\fraki_{1}(\bfy,\bfx_{1}-\bfz_{1})\Big)}\\
&=\Exu{\bfz_{1}:S\rightarrow\ZNk}{H_{S}(\bfz_{1})\cdot \Exu{\bfz_{2}:T\rightarrow\ZNk}{H_{T}(\bfz_{2})\cdot f\Big(\fraki_{2}(\bfy,\bfx_{1}-\bfz_{1},\bfx_{2}-\bfz_{2})\Big)}}\\
&=\Exu{\bfz:S\cup T\rightarrow\ZNk}{H_{S\cup T}(\bfz)\cdot f\Big(\fraki(\bfy,(\bfx_{1}\uplus \bfx_{2})-\bfz\Big)}\\
&=D_{S\cup T,\bfx_{1}\uplus \bfx_{2}}[f](\bfy).\qedhere
\end{align*}
\end{proof}

\Cref{lem:derivatives-compose} implies the following important corollary about the derivative-based globalness notion (defined in \Cref{def:derivative-based-global}).

\begin{corollary}\label{cor:globalness-of-derivative}
If $f:\Omega^{\calU,m}\rightarrow\mathbb{C}$ is $(r,\lambda,d)$-$L^{p}$-global, then for any matching $S\in\calM_{\calU,\leq d}$ and label $\bfx\in \Map{S}{\ZNk}$, the derivative $D_{S,\bfx}[f]$ is $(r,r^{|S|}\lambda ,d-|S|)$-$L^{p}$-global.
\end{corollary}

\begin{proof}
For each matching $T\in \calM_{\calU_{\setminus S},\,\leq d-|S|}$, we know that $S\cup T\in \calM_{\calU,\leq d}$. So by the assumption that $f$ is $(r,\lambda,d)$-$L^{p}$-global, we have $\left\|D_{S\cup T,\bfx''}f\right\|_{p}\leq r^{|S|+|T|}\lambda$ for any label $\bfx''\in \Map{S\cup T}{\ZNk}$.
By \Cref{lem:derivatives-compose} it follows that $\left\|D_{T,\bfx'}\left[D_{S,\bfx}f\right]\right\|_{p}\leq r^{|T|}\cdot r^{|S|}\lambda$ for any label $\bfx'\in \Map{T}{\ZNk}$, as required.
\end{proof}

\subsection{Projections Commutes with Derivatives}\label{subsec:operators_commute}

The goal of this subsection is to show that the derivative operators ``commute'' with the projection operators defined in \Cref{def:level-d-projection}. For that purpose, we first compute the derivatives of character functions.

\begin{proposition}\label{prop:derivative-for-pure-functions}
On the space $\Omega^{\calU,m}$, given two pairs $(S,\bfx),(M,\bfa)\in \frakX^{\calU,\leq m}$, we have
\[
D_{S,\bfx}[\psi_{M,\bfa}]=\begin{cases}
\Psi(|\calU|,m,|S|)^{-1/2}\cdot \chi_{\bfa_{|S}}(\bfx)\cdot\psi_{M\setminus S,\,\bfa_{\setminus S}}&\text{if }S\subseteq M,\\
0 &\text{if }S\not\subseteq M,
\end{cases}
\]
We note that in the above equation, $\psi_{M,\bfa}$ is a character on $\Omega^{\calU,m}$, while $\psi_{M\setminus S,\,\bfa_{\setminus S}}$ is a character on $\Omega^{\calU,m}_{\setminus S}$.
\end{proposition}

\begin{proof} We consider the following two cases respectively.

\textbf{Case 1:} $S\not\subseteq M$. If $M\cup S$ is not a matching, then $\psi_{M,\bfa}(\fraki(\bfy,\bfz))=0$ for all $\bfy\in\Omega^{\calU,m}_{\setminus S}$ and $\bfz\in\Map{S}{\ZNk}$, and hence $D_{S,\bfx}[\psi_{M,\bfa}]=0$ by definition. If $M\cup S$ is a matching, pick an edge $e\in S\setminus M$ and let $\bfx_{\setminus \{e\}}$ be the restriction of $\bfx$ to $S\setminus\{e\}$. It is easy to see that the value of the function $h:=D_{S\setminus\{e\},\,\bfx_{\setminus\{e\}}}[\psi_{M,\bfa}]$ at an input $\bfy'\in \Omega^{\calU,m}_{\setminus(S\setminus\{e\})}$ does not depend on the coordinate $\bfy'(e)$. Therefore by \Cref{def:derivative}, for any $\bfy\in \Omega^{\calU,m}_{\setminus S}$ we can pick an arbitrary $\bfy'\in  \Omega^{\calU,m}_{\setminus(S\setminus\{e\})}$ that extends $\bfy$ and have
$$D_{\{e\},\bfx(e)}[h](\bfy)=\frac{1}{N^{k}}\cdot \left(N^{k}-1\right)\cdot h(\bfy')+\frac{N^{k}-1}{N^{k}}\cdot (-1)\cdot h(\bfy')=0.$$
Now by \Cref{lem:derivatives-compose} we have $D_{S,\bfx}[\psi_{M,\bfa}]=D_{\{e\},\bfx(e)}[h]=0$.

\textbf{Case 2:} $S\subseteq M$. By \Cref{def:characters,def:derivative}, for $\bfy\in\Omega^{\calU,m}_{\setminus S}$ we have
\begin{equation}\label{eq:der_of_char_0}
D_{S,\bfx}[\psi_{M,\bfa}](\bfy)=\Psi(|\calU|,m,|M|)^{-1/2}\prod_{e\in M\setminus S}\chi_{\bfa(e)}(\bfy(e))\cdot \Exu{\bfz:S\rightarrow\ZNk}{H_{S}(\bfz)\prod_{e\in S}\chi_{\bfa(e)}(\bfx(e)-\bfz(e))}
\end{equation}
Using \Cref{lem:derivative_primitive}, we have
\begin{align}
\Exu{\bfz:S\rightarrow\ZNk}{H_{S}(\bfz)\cdot\prod_{e\in S}\chi_{\bfa(e)}(\bfx(e)-\bfz(e))}&=\chi_{\bfa}(\bfx)\cdot \inp{H_{S}}{\chi_{\bfa_{|S}}}\nonumber\\
&=\chi_{\bfa}(\bfx)\cdot \inp{\sum_{\bfa':S\rightarrow\ZNk\setminus\{0\}}\chi_{\bfa'}}{\chi_{\bfa_{|S}}}=\chi_{\bfa}(\bfx). \label{eq:der_of_char_1}
\end{align}
By \Cref{def:derivative} again we have
\begin{equation}\label{eq:der_of_char_2}
\prod_{e\in M\setminus S}\chi_{\bfa(e)}(\bfy(e))=\Psi\left(|\calU_{\setminus S}|,m-|S|,|M\setminus S|\right)^{1/2}\cdot\psi_{M\setminus S,\,\bfa_{\setminus S}}(\bfy).
\end{equation}
Finally, by \Cref{def:Psi} we know that
\begin{equation}\label{eq:der_of_char_3}
\Psi(|\calU|,m,|M|)^{-1/2}=\Psi(|\calU|,m,|S|)^{-1/2}\cdot \Psi\left(|\calU_{\setminus S}|,m-|S|,|M\setminus S|\right)^{-1/2}.
\end{equation}
Plugging \eqref{eq:der_of_char_1}, \eqref{eq:der_of_char_2}, and \eqref{eq:der_of_char_3} into \eqref{eq:der_of_char_0} yields the conclusion. 
\end{proof}

The following lemma then shows that the projection operators (defined in \Cref{def:level-d-projection}) commute with the derivative operators, up to the obvious change in degrees. 

\begin{lemma}\label{lem:derivative-projection-commute}
Given an integer $d\leq m$, a function $f:\Omega^{\calU,m}\rightarrow\mathbb{C}$ and any character $(S,\bfx)\in \frakX^{\calU,\leq m}$ with $|S|\leq d$, we have $$D_{S,\bfx}P_{\frakX}^{=d}[f]=P_{\frakX}^{=d-|S|}D_{S,\bfx}[f].$$
\end{lemma}

\begin{proof}
On the one hand, it follows from \Cref{prop:formula-of-projection,prop:derivative-for-pure-functions} that
\begin{equation}\label{eq:derivative-of-level-d}
D_{S,\bfx}P_{\frakX}^{=d}[f]=\Psi(|\calU|,m,|S|)^{-1/2}\sum_{\substack{(M,\bfa)\in \frakX^{\calU,d}\\
M\supseteq S}}\langle f,\psi_{M,\bfa}\rangle\cdot \chi_{\bfa_{|S}}(\bfx)\cdot \psi_{M\setminus S,\,\bfa_{\setminus S}}.
\end{equation}
Note that in the above equation $\psi_{M,\bfa}$ is the character on $\Omega^{\calU,m}$ while $\psi_{M\setminus S,\,\bfa_{\setminus S}}$ is the character on $\Omega^{\calU,m}_{\setminus S}$.

On the other hand, for any character $(T,\bfa)\in \calM_{\calU_{\setminus S},\, d-|S|}$, we can calculate (using \Cref{lem:derivative_primitive} in the second transition)
\begin{align*}
&\quad\Psi\Big(|\calU_{\setminus S}|,m-|S|,|T|\Big)^{1/2}\cdot\left\langle D_{S,\bfx}[f],\psi_{T,\bfa'}\right\rangle\\
&=  
\Exu{\bfy\in\Omega^{\calU,m}_{\setminus S}}{\Exu{\bfz:S\rightarrow\ZNk}{H_{S}(\bfz)\cdot f\Big(\fraki(\bfy,\bfx-\bfz)\Big)}\cdot \overline{\prod_{e\in T}\chi_{\bfa'(e)}(\bfy(e))}}\\
&=\Exu{\bfy\in \Omega^{\calU,m}_{\setminus S},\,\bfz:S\rightarrow\ZNk}{\left(\sum_{\bfa'':S\rightarrow\ZNk\setminus\{0\}}\chi_{\bfa''}(\bfx-\bfz)\right)\cdot f\Big(\fraki(\bfy,\bfz)\Big)\cdot \overline{\prod_{e\in T}\chi_{\bfa'(e)}(\bfy(e))}}\\
&=\sum_{\bfa'':S\rightarrow\ZNk\setminus\{0\}}\chi_{\bfa''}(\bfx)\cdot\Exu{\bfy\in \Omega^{\calU,m}_{\setminus S},\,\bfz:S\rightarrow\ZNk}{f\Big(\fraki(\bfy,\bfz)\Big)\cdot \overline{\prod_{e\in S}\chi_{\bfa''(e)}(\bfz(e))\prod_{e\in T}\chi_{\bfa'(e)}(\bfy(e))}}\\
&=\sum_{\substack{\bfa:S\cup T\rightarrow\ZNk\setminus\{0\}\\
\bfa_{\setminus S}=\bfa'
}}\chi_{\bfa_{|S}}(\bfx)\cdot\Exu{\bfy\in \Omega^{\calU,m}_{\setminus S},\,\bfz:S\rightarrow\ZNk}{f\Big(\fraki(\bfy,\bfz)\Big)\cdot \overline{\prod_{e\in S}\chi_{\bfa(e)}(\bfz(e))\prod_{e\in T}\chi_{\bfa(e)}(\bfy(e))}}\\
&=\sum_{\substack{\bfa:S\cup T\rightarrow\ZNk\setminus\{0\}\\
\bfa_{\setminus S}=\bfa'
}}\chi_{\bfa_{|S}}(\bfx)\cdot\Psi(|\calU|,m,|S|)^{-1}\cdot\Exu{\boldsymbol{\xi}\in\Omega^{\calU,m}}{f(\boldsymbol{\xi})\cdot \overline{\prod_{e\in S\cup T}\chi_{\bfa(e)}(\boldsymbol{\xi}(e))}}\\
&=\Psi(|\calU|,m,|S|)^{-1}\cdot\sum_{\bfa':S\rightarrow\ZNk\setminus\{0\}}\chi_{\bfa'}(\bfx)\cdot\Psi(|\calU|,m,|S\cup T|)^{1/2}\cdot\left\langle f,\psi_{S\cup T,\bfa'\uplus\bfa}\right\rangle.
\end{align*}
In the fifth transition above, we use the facts that $\prod_{e\in S\cup T}\chi_{\bfa(e)}(\boldsymbol{\xi}(e))$ is nonzero only if $\boldsymbol{\xi}$ lies in the image of the embedding $\fraki:\Omega^{\calU,m}_{\setminus S}\times \Map{S}{\ZNk}\hookrightarrow\Omega^{\calU,m}$, and that this image has size $\Psi(|\calU|,m,|S|)\cdot\left|\Omega^{\calU,m}\right|$. 
Now plugging into the above display the relation
$$\Psi\Big(|\calU_{\setminus S}|,m-|S|,|T|\Big)^{1/2}=\Psi(|\calU|,m,|S|)^{-1/2}\cdot \Psi(|\calU|,m,|S\cup T|)^{1/2}$$
yields
$$\left\langle D_{S,\bfx}[f],\psi_{T,\bfa'}\right\rangle=\Psi(|\calU|,m,|S|)^{-1/2}\cdot \sum_{\bfa':S\rightarrow\ZNk\setminus\{0\}}\chi_{\bfa'}(\bfx)\cdot\left\langle f,\psi_{S\cup T,\bfa'\uplus \bfa}\right\rangle.$$
Comparing the above with \eqref{eq:derivative-of-level-d}, one can see that $D_{S,\bfx}P_{\frakX}^{=d}[f]=P_{\frakX}^{=d-|S|}D_{S,\bfx}[f]$.
\end{proof}

\subsection{Comparison with Product Space}\label{subsec:compartison_product}

In this subsection, we compare our labeled matching space $\Omega^{\calU,m}$ with the following product space:

\begin{definition}
Fix $p\in (0,1)$. Consider a random element $Z$ of $\Map{\tprod\calU}{\ZNk\cup\{\nil\}}$ such that for each $e\in \tprod\calU$, the value $Z(e)$ is independent and identically distributed according to
$$\Pr{Z(e)=\nil}=1-p,\quad\text{and}\quad\Pr{Z(e)=z}=N^{-k}p\text{ for each }z\in \ZNk.$$
The ground set $\Map{\tprod\calU}{\ZNk\cup\{\nil\}}$ endowed with the distribution of such a random element $Z$ is a probability space, which we denote by $\Gamma^{\calU,p}$. 
\end{definition}

\subsubsection{Mimicking on a Given Level}

We show that for a fixed degree $d\leq m$, one can choose an appropriate parameter $p$ such that the degree-$d$ Fourier level of the product space $\Gamma^{\calU,p}$ ``mimics'' the level-$d$ characters of the space $\Omega^{\calU,m}$. To make this notion precise, we start with an abstract collection of Fourier coefficients indexed by $\frakX^{\calU,d}$ (as defined in \Cref{def:frakX}), and compare the corresponding “Fourier inverse functions” they define on the two spaces.  
\begin{definition}\label{def:Fourier_inverse_Omega}
Fix a nonnegative integers $d$ such that $d\leq m$. For any map $\varphi:\frakX^{\calU,d}\rightarrow\bC$, we associate with it a ``Fourier inverse'' function $\varphi^{\vee}_{(m)}\in L^{2}(\Omega^{\calU,m})$ by
$$\varphi^{\vee}_{(m)}(\bfy):=\sum_{(M,\bfa)\in \frakX^{\calU,d}}\varphi(M,\bfa)\cdot\prod_{e\in M}\chi_{\bfa(e)}(\bfy(e)).$$
\end{definition}

\begin{definition}
Fix $p\in (0,1)$. For any map $\varphi:\frakX^{\calU,d}\rightarrow\bC$, we define a function $\varphi^{\natural}_{(p)}\in L^{2}(\Gamma^{\calU,p})$. Specifically, for every $\bfy\in \Map{\tprod\calU}{\ZNk\cup\{\nil\}}$ we let
$$\varphi^{\natural}_{(p)}(\bfy):=\sum_{(M,\bfa)\in \frakX^{\calU,d}}\left|\varphi(M,\bfa)\right|\cdot\prod_{e\in M}\chi_{\bfa(e)}(\bfy(e)).$$
\end{definition}

The following proposition demonstrates that for an appropriate parameter $p$, the two Fourier inverse functions have the same $L^{2}$-norm.
\begin{proposition}\label{prop:L2_comparison}
Fix integers $d,m$ such that $d\leq m\leq |\calU|$, and let $p=\Psi(|\calU|,m,d)^{1/d}$. For any $\varphi:\frakX^{\calU,d}\rightarrow\bC$, we have $\left\|\varphi^{\vee}_{(m)}\right\|_{2}=\left\|\varphi^{\natural}_{(p)}\right\|_{2}$.
\end{proposition}
\begin{proof}
Simply expanding $\left\|\varphi^{\vee}_{(m)}\right\|_{2}^{2}=\Exu{\bfy\in \Omega^{\calU,m}}{\varphi^{\vee}_{(m)}(\bfy)\cdot \overline{\varphi^{\vee}_{(m)}(\bfy)}}$ yields
\begin{equation}\label{eq:expanding_Omega}
\left\|\varphi^{\vee}_{(m)}\right\|_{2}^{2}=\sum_{(M,\bfa)\in \frakX^{\calU,d}}\left|\varphi(M,\bfa)\right|^{2}\cdot \Psi(|\calU|,m,d),
\end{equation}
and expanding
$\left\|\varphi^{\natural}_{(p)}\right\|_{2}^{2}=\Exu{\bfy\sim\Gamma^{\calU,m}}{\varphi^{\natural}_{(p)}(\bfy)\cdot \overline{\varphi^{\natural}_{(m)}(\bfy)}}$ yields
\begin{equation}\label{eq:expanding_Gamma}
\left\|\varphi^{\natural}_{(p)}\right\|_{2}^{2}=\sum_{(M,\bfa)\in \frakX^{\calU,d}}\left|\varphi(M,\bfa)\right|^{2}\cdot p^{d}.
\end{equation}
The conclusion then follows by comparing \eqref{eq:expanding_Omega} with \eqref{eq:expanding_Gamma} and using $p^{d}=\Psi(|\calU|,m,d)$.
\end{proof}
\subsubsection{Mimicking on other Levels}

We claim that not only does the parameter $p$ chosen in \Cref{prop:L2_comparison} not only ensures that the degree-$d$ Fourier levels of $\Gamma^{\calU,p}$ and $\Omega^{\calU,m}$ resemble each other, but also that the nearby levels --- those not too far from degree $d$ --- exhibit a similar approximation under the same choice of $p$. This effect on nearby levels is primarily due to the fact that the probability parameters $\Psi(n,m,d)$ (defined in \Cref{def:Psi}) grows approximately exponentially in $d$ when $d$ is small, as formalized below.  

\begin{proposition}\label{prop:approximate-product}
Fix integers $n,m,d$ such that $n\geq 2km$ and $m\geq 2(d+1)$. Let $p=\Psi(n,m,d)^{1/d}$.
\begin{enumerate}[label=(\arabic*)]
\item For $\ell\in\{0,1,\dots,d\}$ we have $p^{\ell}\leq \Psi(n,m,\ell)\leq (2p)^{\ell}$.
\item For $\ell\in\{d,d+1,\dots,m\}$ we have $\Psi(n,m,\ell)\leq p^{\ell}$.
\end{enumerate}
\end{proposition}
\begin{proof}
For $i\in\{0,1,\dots,m-1\}$, we have
$$\frac{\Psi(n-i-1,m-i-1,1)}{\Psi(n-i,m-i,1)}=\frac{(n-i)^{k}}{(n-i-1)^{k}}\cdot\frac{m-i-1}{m-i}<1.$$
So $\Psi(n-i,m-i,1)$ is decreasing in $i$. Furthermore, 
\[
\frac{\Psi(n,m,1)}{\Psi(n-d,m-d,1)}=\frac{(n-d)^{k}}{n^{k}}\cdot\frac{m}{m-d}\leq \frac{m}{m-d}\leq 2.
\]
Therefore, for $\ell\in\{0,1,\dots,d\}$ we have
\begin{align*}
\Psi(n,m,\ell)&=\prod_{i=0}^{\ell-1}\Psi(n-i,m-i,1)\leq 2^{\ell}\cdot\Psi(n-d,m-d,1)^{\ell}\\
&\leq 2^{\ell}\cdot\left(\prod_{i=0}^{d-1}\Psi(n-i,m-i,1)\right)^{\ell/d}=2^{\ell}\cdot\Psi(n,m,d)^{\ell/d}=(2p)^{\ell},
\end{align*}
as well as
$$\Psi(n,m,\ell)=\prod_{i=0}^{\ell-1}\Psi(n-i,m-i,1)\geq \left(\prod_{i=0}^{d}\Psi(n-i,m-i,1)\right)^{\ell/d}=\Psi(n,m,d)^{\ell/d}=p^{\ell}.$$
For $\ell\geq d$ we have
\begin{align*}
\Psi(n,m,\ell)&=\Psi(n,m,d)\cdot\prod_{i=d}^{\ell-1}\Psi(n-i,m-i,1)\leq \Psi(n,m,d)\cdot\Psi(n-d,m-d,1)^{\ell-d}\\
&\leq \Psi(n,m,d)\cdot\left(\prod_{i=0}^{d-1}\Psi(n-i,m-i,1)\right)^{(\ell-d)/d}=\Psi(n,m,d)^{\ell/d}=p^{\ell}.\qedhere
\end{align*}
\end{proof}

\subsubsection{Comparison of $q$-Norms}

Recall that \Cref{prop:L2_comparison} establishes the equality of the $L^2$-norms of the two Fourier inverse functions corresponding to the same collection of Fourier coefficients. Equipped with \Cref{prop:approximate-product}, we now extend this comparison (but with an inequality instead of equality) to the $L^q$-norms of these functions, for any positive even integer $q$.

\begin{lemma}\label{lem:comparison_q_norm}
Fix integers $d,m$ such that $d\leq m\leq |\calU|$, and let $p=\Psi(|\calU|,m,d)^{1/d}$. For any $\varphi:\frakX^{\calU,d}\rightarrow\bC$ and any positive integer $q$, we have $\left\|\varphi^{\vee}_{(m)}\right\|_{2q}\leq \left\|\varphi^{\natural}_{(p)}\right\|_{2q}$.
\end{lemma}

\begin{proof}
Let $\mathsf{GOOD}$ be the collection of all sequences $(M_{1},\bfa_{1}),\dots,(M_{2q},\bfa_{2q})\in \calM_{\calU,d}$ such that for every $e\in \tprod\calU$, 
$$
\sum_{i=1}^{q}\widetilde{\bfa_{i}}(e)=\sum_{i=q+1}^{2q}\widetilde{\bfa_{i}}(e),
$$
where $\widetilde{\bfa_{i}}:\tprod\calU\rightarrow \ZNk$ is the extension of $\bfa_{i}$ by value 0 on $(\tprod\calU)\setminus M_{i}$, for each $i\in [2q]$. It is easy to see that for the expected value
\begin{equation}\label{eq:GOOD-expected-1}
\mathscr{E}_{\Omega}\Big((M_{1},\bfa_{1}),\dots,(M_{2q},\bfa_{2q})\Big):=\Exu{\bfy\in \Omega^{\calU,m}}{\prod_{i=1}^{q}\left(\prod_{e\in M_{i}}\chi_{\bfa_{i}(e)}(\bfy(e))\right) \prod_{i=q+1}^{2q}\left(\prod_{e\in M_{i}}\overline{\chi_{\bfa_{i}(e)}(\bfy(e))}\right)}
\end{equation}
equals $\Psi\Big(|\calU|,m,\big|M_{1}\cup\dots\cup M_{2q}\big|\Big)$ if $((M_{1},\bfa_{1}),\dots,(M_{2q},\bfa_{2q}))\in \textsf{GOOD}$ and equals 0 otherwise. The same conclusion also holds if the expected value in \eqref{eq:GOOD-expected-1} is evaluated not for a random element $\bfy\in \Omega^{\calU,m}$ but for $\bfy$ sampled from $\Gamma^{\calU,p}$:
\begin{equation}\label{eq:GOOD-expected-2}
\mathscr{E}_{\Gamma}\Big((M_{1},\bfa_{1}),\dots,(M_{2q},\bfa_{2q})\Big):=\Exu{\bfy\in \Gamma^{\calU,p}}{\prod_{i=1}^{q}\left(\prod_{e\in M_{i}}\chi_{\bfa_{i}(e)}(\bfy(e))\right) \prod_{i=q+1}^{2q}\left(\prod_{e\in M_{i}}\overline{\chi_{\bfa_{i}(e)}(\bfy(e))}\right)}
\end{equation}
equals $p^{|M_{1}\cup\dots\cup M_{2q}|}$ if $((M_{1},\bfa_{1}),\dots,(M_{2q},\bfa_{2q}))\in \textsf{GOOD}$ and equals 0 otherwise. Therefore the values of \eqref{eq:GOOD-expected-1} and \eqref{eq:GOOD-expected-2} are always nonnegative real numbers. Furthermore, since $p^{|M_{1}\cup\dots\cup M_{2q}|}$ is always at least $\Psi\Big(|\calU,m,\big|M_{1}\cup\dots\cup M_{2q}\big|\Big)$, by \Cref{prop:approximate-product}(2), the value of \eqref{eq:GOOD-expected-1} is always at most the value of \eqref{eq:GOOD-expected-2}, whether the sequence $((M_{1},\bfa_{1}),\dots,(M_{2q},\bfa_{2q}))$ belongs to $\textsf{GOOD}$ or not. Now since $ \left\|\varphi^{\vee}_{(m)}\right\|_{2q}^{2q}=
\Exu{\bfy\in \Omega^{\calU,m}}{\varphi^{\vee}_{(m)}(\bfy)^{q}\cdot\overline{\varphi^{\vee}_{(m)}(\bfy)^{q}}}$ expands into
\[\sum_{(M_{1},\bfa_{1}),\dots,(M_{2q},\bfa_{2q})\in \frakX^{\calU,d}}\left(\prod_{i=1}^{q}\varphi(M_{i},\bfa_{i})\prod_{i=q+1}^{2q}\overline{\varphi(M_{i},\bfa_{i})}\cdot\mathscr{E}_{\Omega}\Big((M_{1},\bfa_{1}),\dots,(M_{2q},\bfa_{2q})\Big)\right)
\]
and $\left\|\varphi^{\natural}_{(p)}\right\|_{2q}^{2q}=
\Exu{\bfy\sim \Gamma^{\calU,p}}{\varphi^{\natural}_{(p)}(\bfy)^{q}\cdot\overline{\varphi^{\natural}_{(p)}(\bfy)^{q}}}$ expands into
\[\sum_{(M_{1},\bfa_{1}),\dots,(M_{2q},\bfa_{2q})\in \frakX^{\calU,d}}\left(\prod_{i=1}^{q}\left|\varphi(M_{i},\bfa_{i})\right|\prod_{i=q+1}^{2q}\left|\varphi(M_{i},\bfa_{i})\right|\cdot\mathscr{E}_{\Gamma}\Big((M_{1},\bfa_{1}),\dots,(M_{2q},\bfa_{2q})\Big)\right),
\]
by term-wise comparison it follows that $\left\|\varphi^{\vee}_{(m)}\right\|_{2q}^{2q}\leq \left\|\varphi^{\natural}_{(p)}\right\|_{2q}^{2q}$.
\end{proof}

\subsubsection{Comparison of Derivatives}

The final comparison required between the spaces $\Gamma^{\calU,p}$ and $\Omega^{\calU,m}$ concerns their respective derivative operators. For the space $\Gamma^{\calU,p}$, we adopt the following definition of formal derivatives that act purely on Fourier coefficients.

\begin{definition}\label{def:formal_derivative}
Fix a nonnegative integer $d\leq|\calU|$. For any matching $S\in \calM_{\calU,\leq d}$ and $\bfx\in \Map{S}{\ZNk}$, we define the formal derivative operator $\widehat{D_{S,\bfx}}:\Map{\frakX^{\calU,d}}{\bC}\rightarrow\Map{\frakX^{\calU_{\setminus S},\,d-|S|}}{\bC}$ as follows. For each $\varphi:\frakX^{\calU,d}\rightarrow\bC$ and $(T,\bfa)\in \frakX^{\calU_{\setminus S},\,d-|S|}$, let
$$\left(\widehat{D_{S,\bfx}}[\varphi]\right)(T,\bfa):=\sum_{\bfa':S\rightarrow\ZNk\setminus\{0\}}\chi_{\bfa'}(\bfx)\cdot \varphi(S\cup T, \bfa'\uplus \bfa).$$
\end{definition}

Our comparison lemma for derivatives is as follows.
\begin{lemma}\label{lem:comparison_2_norm}
Fix integers $d,m$ such that $d\leq m\leq |\calU|$, and let $p=\Psi(|\calU|,m,d)^{1/d}$. For any $\varphi:\frakX^{\calU,d}\rightarrow\bC$, any $S\in \calM_{\calU,\leq d}$ and any $\bfx\in \Map{S}{\ZNk}$, we have $$\left\|D_{S,\bfx}\big[\varphi^{\vee}_{(m)}\big]\right\|_{2}\geq 2^{-|S|/2}\cdot \left\|\left(\widehat{D_{S,\bfx}}[\varphi]\right)^{\natural}_{(p)}\right\|_{2}.$$
\end{lemma}

\begin{proof}
We first note that it follows easily from \Cref{def:characters,def:formal_derivative} and \Cref{prop:derivative-for-pure-functions} that
$$D_{S,\bfx}\big[\varphi^{\vee}_{(m)}\big]=\left(\widehat{D_{S,\bfx}}[\varphi]\right)^{\vee}_{(m-|S|)}$$
as functions on $\Omega^{\calU, d}_{\setminus S}$. Therefore we may write for convenience $\varphi':=\widehat{D_{S,\bfx}}[\varphi]$ and it suffices to show 
\begin{equation}\label{eq:derivative_comparison_goal}
\left\|(\varphi')^{\vee}_{(m-|S|)}\right\|_{2}\geq 2^{-|S|/2} \left\|(\varphi')^{\natural}_{(p)}\right\|_{2}.
\end{equation}

Similarly to \eqref{eq:expanding_Omega} and \eqref{eq:expanding_Gamma}, we have the identities
$$
\left\|(\varphi')^{\vee}_{(m)}\right\|_{2}^{2}=\sum_{(T,\bfa)\in \frakX^{\calU_{\setminus S},\,d-|S|}}\left|\varphi(T,\bfa)\right|^{2}\cdot \Psi\Big(|\calU_{\setminus S}|,m-|S|,d-|S|\Big),
$$
and 
$$
\left\|(\varphi')^{\natural}_{(p)}\right\|_{2}^{2}=\sum_{(T,\bfa)\in \frakX^{\calU_{\setminus S},\,d-|S|}}\left|\varphi(T,\bfa)\right|^{2}\cdot p^{d-|S|}.
$$
Now note that by \Cref{prop:approximate-product}(1) we have
$$
\Psi\Big(|\calU_{\setminus S}|,m-|S|,d-|S|\Big)=\frac{\Psi(|\calU|,m,d)}{\Psi(|\calU|,m,|S|)}\geq \frac{p^{d}}{(2p)^{|S|}}=2^{-|S|}p^{d-|S|},
$$
and combining the above three displays leads to the desired conclusion \eqref{eq:derivative_comparison_goal}.
\end{proof}

The next lemma is a standard identity for formal derivatives over product spaces (cf. \cite[Theorem 4.6]{KLM23} for analogous statements).

\begin{lemma}\label{lem:derivative_identity}
Fix $p\in (0,1)$. For any $\varphi:\frakX^{\calU,d}$, we have the identity
\[ 
\sum_{S\in\calM_{\calU,\leq d}} p^{|S|}\Exu{\bfx:S\rightarrow\ZNk}{\left\|\left(\widehat{D_{S,\bfx}}[\varphi]\right)^{\natural}_{(p)}\right\|_{2}^{2}}= 2^{d}\left\|\varphi^{\natural}_{(p)}\right\|_{2}^{2}.
\]
\end{lemma}
\begin{proof}
Expanding the 2-norm on the left hand side yields
\begin{align}
\left\|\left(\widehat{D_{S,\bfx}}[\varphi]\right)^{\natural}_{(p)}\right\|_{2}^{2}&=\Exu{\bfy\in \Gamma^{\calU_{\setminus S},\,p}}{\left(\widehat{D_{S,\bfx}}[\varphi]\right)^{\natural}_{(p)}(\bfy)\cdot \overline{\left(\widehat{D_{S,\bfx}}[\varphi]\right)^{\natural}_{(p)}(\bfy)}}\nonumber\\
&=\sum_{(T,\bfa)\in \frakX^{\calU_{\setminus S},\,d-|S|}}\left|\widehat{D_{S,\bfx}}[\varphi](T,\bfa)\right|^{2}\cdot p^{d-|S|}.\label{eq:derivative_identity_1}
\end{align}
By \Cref{def:formal_derivative} we have
\begin{align}
&\quad\Exu{\bfx:S\rightarrow \ZNk}{\sum_{(T,\bfa)\in \frakX^{\calU_{\setminus S},\,d-|S|}}\left|\widehat{D_{S,\bfx}}[\varphi](T,\bfa)\right|^{2}}\nonumber\\
&=\sum_{(T,\bfa)\in \frakX^{\calU_{\setminus S},\,d-|S|}}\Exu{\bfx:S\rightarrow \ZNk}{\left|\sum_{\bfa':S\rightarrow\ZNk\setminus\{0\}}\chi_{\bfa'}(\bfx)\cdot \varphi(S\cup T,\bfa'\uplus\bfa)\right|^{2}}\nonumber\\
&=\sum_{(T,\bfa)\in \frakX^{\calU_{\setminus S},\,d-|S|}}\left(\sum_{\bfa'_{1},\bfa'_{2}:S\rightarrow\ZNk\setminus\{0\}}\varphi(S\cup T,\bfa'_{1}\uplus \bfa)\overline{\varphi(S\cup T,\bfa'_{2}\uplus \bfa)}\cdot \Exu{\bfx:S\rightarrow\ZNk}{\chi_{\bfa'_{1}-\bfa'_{2}}(\bfx)}\right).\nonumber\\
&=\sum_{(T,\bfa)\in \frakX^{\calU_{\setminus S},\,d-|S|}}\sum_{\bfa':S\rightarrow\ZNk\setminus\{0\}}\left|\varphi(S\cup T,\bfa'\uplus\bfa)\right|^{2}=\sum_{\substack{(M,\bfa'')\in \frakX^{\calU,d}\\ M\supseteq S}}\left|\varphi(M,\bfa'')\right|^{2}.\label{eq:derivative_identity_2}
\end{align}
Combining \eqref{eq:derivative_identity_1} and \eqref{eq:derivative_identity_2}, we obtain the desired result
\begin{align*}
\sum_{S\in\calM_{\calU,\leq d}} p^{|S|}\Exu{\bfx:S\rightarrow\ZNk}{\left\|\left(\widehat{D_{S,\bfx}}[\varphi]\right)^{\natural}_{(p)}\right\|_{2}^{2}}&=p^{d}\sum_{S\in \calM_{\calU,\leq d}}\sum_{\substack{(M,\bfa'')\in \frakX^{\calU,d}\\ M\supseteq S}}\left|\varphi(M,\bfa'')\right|^{2}\\
&=p^{d}\cdot 2^{d}\sum_{(M,\bfa'')\in \frakX^{\calU,d}}\left|\varphi(M,\bfa'')\right|^{2}=2^{d}\left\|\varphi^{\natural}_{(p)}\right\|_{2}^{2},
\end{align*}
where in the last transition we use \eqref{eq:expanding_Gamma}.
\end{proof}

\subsection{The Hypercontractive Inequality}\label{subsec:hypercontractive}

We are now prepared to derive a derivative-based hypercontractive inequality for our space $\Omega^{\calU,m}$ by leveraging the corresponding inequality already established for product spaces. 
In particular, applying the general result of \cite[Theorem~4.1]{KLM23} to the product space $\Gamma^{\calU,p}$ yields the following.\footnote{Strictly speaking, \cite{KLM23} proves the inequality only for real-valued functions, whereas our application requires it for complex-valued functions. Nevertheless, the proof in \cite{KLM23} extends to the complex setting without difficulty. Alternatively, one can apply \cite[Theorem~4.1]{KLM23} separately to the real and imaginary parts of $\varphi^{\natural}_{(p)}$, which introduces an additional factor of $2^{2q}$ on the right-hand side of \eqref{eq:KLM-theorem}, easily absorbed into the other parameters.}

\begin{lemma}[{\cite[Theorem 4.1]{KLM23}}]\label{lem:KLM-theorem}
Let $p\in (0,1)$. Suppose $q$ is a positive integer and $\rho\in (0,\frac{1}{3\sqrt{2q}})$. For any $\varphi:\frakX^{\calU,d}\rightarrow\bC$, we have\footnote{%
Compared with the inequality in \cite[Theorem~4.1]{KLM23}, the right-hand side of \eqref{eq:KLM-theorem} contains an additional factor of $p^{|S|}$.  
This arises because the expectation in \eqref{eq:KLM-theorem} is taken over $\bfx:S\to\mathbb{Z}_n^k$ rather than over $\bfx:S\to\mathbb{Z}_n^k\cup\{\nil\}$; whenever $\bfx$ contains a $\nil$, the derivative with respect to $\bfx$ is zero.%
}
\begin{equation}\label{eq:KLM-theorem}
\left\|\varphi^{\natural}_{(p)}\right\|_{2q}^{2q}\leq \rho^{-2dq}\sum_{S\in\calM_{\calU,\leq d}}\beta^{2q|S|}(2q)^{-q|S|}p^{|S|}\Exu{\bfx:S\rightarrow\ZNk}{\left\|\left(\widehat{D_{S,\bfx}}[\varphi]\right)^{\natural}_{(p)}\right\|_{2}^{2q}},
\end{equation}
where
$\beta:=\rho\sqrt{2q}\left(1+\frac{4(q-1)}{\ln(\rho^{-1}(2q)^{-1/2})}\right)$.
\end{lemma}

We combine \Cref{lem:KLM-theorem} with the comparison results established in \Cref{subsec:compartison_product} to obtain the desired hypercontractive inequality for our space $\Omega^{\calU,m}$.

\begin{lemma}[Derivative-based hypercontractivity]\label{lem:derivative-based-inequality}
Fix integers $d,m$ such that $|\calU|\geq 2km$ and $m\geq 2(d+1)$. Fix $r>0$ and integer $q\geq 1$. For $f:\Omega^{\calU,m}\rightarrow\bC$ such that $f$ lies in the subspace $\lspan\left\{\psi_{M,\bfa}:(M,\bfa)\in \frakX^{\calU,d}\right\}$ of $L^{2}(\Omega^{\calU,m})$, we have
\[
\left\|f\right\|_{2q}^{2q}\leq 2^{d}\rho^{-2dq}\left\|f\right\|_{2}^{2}\cdot\max_{\substack{S\in\calM_{\calU,\leq d}\\ \bfx:S\rightarrow\ZNk}}\left(r^{-|S|}\left\|D_{S,\bfx}f\right\|_{2}\right)^{2q-2},
\]
where
\begin{equation}\label{eq:rho-formula}
\rho:=\frac{1}{4\sqrt{2}}\min\left\{ q^{-1/2},q^{-1}r^{-\frac{q-1}{q}}\right\}.
\end{equation}
\end{lemma}
\begin{proof}
Since $f\in \lspan\left\{\psi_{M,\bfa}:(M,\bfa)\in \frakX^{\calU,d}\right\}$, we may define $\varphi:\frakX^{\calU,d}\rightarrow\bC$ by $$\varphi(M,\bfa):=\Psi(|\calU|,m,|M|)^{-1/2}\cdot \left\langle f,\psi_{M,\bfa}\right\rangle$$ for $(M,\bfa)\in \frakX^{\calU,d}$, and then by \Cref{def:characters,def:Fourier_inverse_Omega} we have $f=\varphi^{\vee}_{(m)}$. 

Let $\beta:=\rho\sqrt{2q}\left(1+\frac{4(q-1)}{\ln(\rho^{-1}(2q)^{-1/2})}\right)$, as in \Cref{lem:KLM-theorem}. Using $\rho^{-1}(2q)^{-1/2}\geq 4$ we get $\beta\leq \rho\sqrt{2q}\cdot 4q$. Now the other upper bound $\rho\leq \frac{1}{4\sqrt{2}}q^{-1}r^{-\frac{q-1}{q}}$ yields $\beta^{2q}q^{-q}r^{2q-2}\leq 1$. 

According to \Cref{lem:KLM-theorem}, we have
\[
\left\|\varphi^{\natural}_{(p)}\right\|_{2q}^{2q}\leq \rho^{-2dq}\sum_{S\in\calM_{\calU,\leq d}}\beta^{2q|S|}(2q)^{-q|S|} p^{|S|}\Exu{\bfx:S\rightarrow\ZNk}{\left\|\left(\widehat{D_{S,\bfx}}[\varphi]\right)^{\natural}_{(p)}\right\|_{2}^{2q}}.
\]
We can apply \Cref{lem:comparison_q_norm,lem:comparison_2_norm} to the above inequality and get
\[
\left\|f\right\|_{2q}^{2q}\leq \rho^{-2dq}\sum_{S\in\calM_{\calU,\leq d}}\beta^{2q|S|}q^{-q|S|} p^{|S|}\Exu{\bfx:S\rightarrow\ZNk}{\left\|\left(\widehat{D_{S,\bfx}}[\varphi]\right)^{\natural}_{(p)}\right\|_{2}^{2}\cdot\left\|D_{S,\bfx}[f]\right\|_{2}^{2q-2}}.
\]
Using the estimations obtained in the preceding paragraph, we simplify the above into
\begin{align*}
\left\|f\right\|_{2q}^{2q}&\leq \rho^{-2dq}\sum_{S\in\calM_{\calU,\leq d}} p^{|S|}r^{-(2q-2)|S|}\Exu{\bfx:S\rightarrow\ZNk}{\left\|\left(\widehat{D_{S,\bfx}}[\varphi]\right)^{\natural}_{(p)}\right\|_{2}^{2}\cdot\left\|D_{S,\bfx}[f]\right\|_{2}^{2q-2}}\\
&\leq \rho^{-2dq}\left(\sum_{S\in\calM_{\calU,\leq d}} p^{|S|}\Exu{\bfx:S\rightarrow\ZNk}{\left\|\left(\widehat{D_{S,\bfx}}[\varphi]\right)^{\natural}_{(p)}\right\|_{2}^{2}}\right)\max_{\substack{S\in\calM_{\calU,\leq d}\\ \bfx:S\rightarrow\ZNk}}\left(r^{-|S|}\left\|D_{S,\bfx}f\right\|_{2}\right)^{2q-2}\\
&=2^{d}\rho^{-2dq}\left\|\varphi^{\natural}_{(p)}\right\|_{2}^{2}\cdot\max_{\substack{S\in\calM_{\calU,\leq d}\\ \bfx:S\rightarrow\ZNk}}\left(r^{-|S|}\left\|D_{S,\bfx}f\right\|_{2}\right)^{2q-2},
\end{align*}
where in the last transition we use \Cref{lem:derivative_identity}. The proof is concluded by noting that by \Cref{prop:L2_comparison} we have that $\left\|\varphi^{\natural}_{(p)}\right\|_{2}^{2}=\|f\|_{2}^{2}$. 
\end{proof}

\subsection{Proof of the Level-$d$ Inequality}\label{subsec:level-d}
Now we complete the proof of the level-$d$ inequality, restated below, by an induction argument. 
\projleveld*

\begin{proof}

The conclusion in the case $d=0$ simply comes from $\Exu{\bfy}{f(\bfy)}^{2}\leq \lambda_{1}^{2}$, which holds by the $L^1$-globalness assumption. We proceed by an induction on $d$. 
Towards this end, fix $d\geq 1$ and assume that the statement holds for all $d'<d$.

Fix $S\neq \emptyset$ and $\bfx:S\rightarrow\ZNk$, so that by \Cref{cor:globalness-of-derivative} we know that $D_{S,\bfx}[f]$ is both $(r,r^{|S|}\lambda_{1},d-|S|)$-$L^{1}$-global and $(r,r^{|S|}\lambda_{2},d-|S|)$-$L^{2}$-global. 
Our first goal 
will be to show that $P_{\frakX}^{=d}[f]$ has discrete derivatives with small norms, and towards this end we use the induction hypothesis. 
Since $|\calU_{\setminus S}|\geq 2k(m-|S|)$ and $m-|S|\geq 2(d-|S|+1)$, we can apply the induction hypothesis on $D_{S,\bfx}[f]:\Omega^{\calU,m}_{\setminus S}\rightarrow \bC$. Combining with \Cref{lem:derivative-projection-commute}, we get
\begin{align}
\left\|D_{S,\bfx}P_{\frakX}^{=d}[f]\right\|_{2}^{2}&= \left\|P_{\frakX}^{=d-|S|}D_{S,\bfx}[f]\right\|_{2}^{2}\\
&\leq r^{2|S|}\lambda_{1}^{2}\left(\frac{10^{5}r^{2}\log(\lambda_{2}/\lambda_{1})}{d-|S|}\right)^{d-|S|}\nonumber\\
&= 10^{5(d-|S|)}\lambda_{1}^{2}r^{2d}\log^{d-|S|}(\lambda_{2}/\lambda_{1})d^{-(d-|S|)}\left(1+\frac{|S|}{d-|S|}\right)^{d-|S|}\nonumber\\
&\leq 10^{5(d-|S|)}\lambda_{1}^{2}r^{2d}\log^{d-|S|}(\lambda_{2}/\lambda_{1})d^{-(d-|S|)}\cdot 10^{5|S|}\nonumber\\
&= \lambda_{1}^{2}\left(\frac{10^{5}r^{2}\log(\lambda_{2}/\lambda_{1})}{d}\right)^{d}\left(\frac{\sqrt{d}}{\log^{1/2}(\lambda_{2}/\lambda_{1})}\right)^{2|S|}\notag\\
&=(r')^{2|S|}(\lambda')^{2},\label{eq:globalness-of-level-d}
 \end{align}
 where we let 
\[\lambda'=\lambda_{1}\left(\frac{10^{5}r^{2}\log(\lambda_{2}/\lambda_{1})}{d}\right)^{d/2}\quad\text{ and }\quad
r'=\frac{\sqrt{d}}{\log^{1/2}(\lambda_{2}/\lambda_{1})}.
\]
We intend to apply~\Cref{lem:derivative-based-inequality}, and for that we pick
\[q=\left\lfloor\frac{4\log (\lambda_{2}/\lambda_{1})}{d}\right\rfloor\quad\text{ and }\quad\rho=  \frac{1}{4\sqrt{2}}\min\left\{ q^{-1/2},q^{-1}(r')^{-\frac{q-1}{q}}\right\}.\]
This choice of parameters ensure that $\rho^{-2}\leq 10^{3}q$, and thus
\begin{align}
2^{d}\rho^{-2dq}\lambda_{1}^{2q}\left(\lambda_{2}/\lambda_{1}\right)^{2} &\leq \lambda_{1}^{2q}\left(2\rho^{-2}(\lambda_{2}/\lambda_{1})^{2/(dq)}\right)^{dq}\leq \lambda_{1}^{2q}(10^{4}q)^{dq}\nonumber\\
&\leq \lambda_{1}^{2q}\left(\frac{10^{5}r^{2}\log(\lambda_{2}/\lambda_{1})} {d}\right)^{dq}=(\lambda')^{2q}.\label{eq:parameter-mess-to-lambda}
\end{align}
Since $P_{\frakX}^{=d}$ is an orthogonal projection, we have
\begin{align*}
\left\|P_{\frakX}^{=d}f\right\|_{2}^{4q}
=\left\langle f,P_{\frakX}^{=d}f\right\rangle^{2q}
&\leq \left\|P_{\frakX}^{=d}f\right\|_{2q}^{2q}\cdot \|f\|_{2q/(2q-1)}^{2q}\\
&\leq \left\|P_{\frakX}^{=d}f\right\|_{2q}^{2q}\cdot \|f\|_{1}^{2q-2}\cdot \|f\|_{2}^{2}\\
&\leq 2^{d}\rho^{-2dq}\lambda_{1}^{2q-2}\lambda_{2}^{2}\left\|P_{\frakX}^{=d}f\right\|_{2}^{2}\cdot \max_{\substack{S\in\calM_{\calU,\leq d}\\ \bfx:S\rightarrow\ZNk}}\left((r')^{-|S|}\left\|D_{S,\bfx}P_{\frakX}^{=d}f\right\|_{2}\right)^{2q-2}\\
&\leq (\lambda')^{2q}\cdot\left\|P_{\frakX}^{=d}f\right\|_{2}^{2}\max\left(\left\|P_{\frakX}^{=d}f\right\|_{2}^{2q-2},(\lambda')^{2q-2}\right),
\end{align*}
where in the second and third transitions we used H\"{o}lder's inequality, the fourth transition is by~\Cref{lem:derivative-based-inequality}, and the last transition is by 
\eqref{eq:parameter-mess-to-lambda} and \eqref{eq:globalness-of-level-d}. It follows that $\left\|P_{\frakX}^{=d}f\right\|_{2}\leq \lambda'$, as desired.
\end{proof}
\end{document}